\documentclass[final,onefignum,onetabnum]
{siamart220329}

\usepackage{braket,amsfonts}
\usepackage{enumitem}
\usepackage{array}
\usepackage{xcolor,witharrows,subcaption,multirow}

\usepackage{algorithmic}
\graphicspath{ {./figures/} }
\usepackage{graphicx,epstopdf}

\Crefname{ALC@unique}{Line}{Lines}

\newcommand{\R}{\mathbb{R}}

\usepackage{mathrsfs}
\newcommand{\f}{\mathscr{F}}
\newcommand{\g}{\mathscr{G}}

\newcommand{\simplef}{f}
\newcommand{\fhat}{\widehat{f}}

\newcommand{\tmdl}{T_\textrm{model}}

\newcommand{\bs}{\boldsymbol}
\mathchardef\mhyphen="2D 

 \usepackage{ulem}
\newtheorem{assumption}{Assumption}

\usepackage{xspace}
\usepackage{bold-extra}
\usepackage[most]{tcolorbox}

\colorlet{texcscolor}{blue!50!black}
\colorlet{texemcolor}{red!70!black}
\colorlet{texpreamble}{red!70!black}
\colorlet{codebackground}{black!25!white!25}

\usepackage{appendix}
\renewcommand{\appendixname}{Supplementary Material}

\lstdefinestyle{siamlatex}{%
  style=tcblatex,
  texcsstyle=*\color{texcscolor},
  texcsstyle=[2]\color{texemcolor},
  keywordstyle=[2]\color{texemcolor},
  moretexcs={cref,Cref,maketitle,mathcal,text,headers,email,url},
}

\tcbset{%
  colframe=black!75!white!75,
  coltitle=white,
  colback=codebackground, 
  colbacklower=white, 
  fonttitle=\bfseries,
  arc=0pt,outer arc=0pt,
  top=1pt,bottom=1pt,left=1mm,right=1mm,middle=1mm,boxsep=1mm,
  leftrule=0.3mm,rightrule=0.3mm,toprule=0.3mm,bottomrule=0.3mm,
  listing options={style=siamlatex}
}

\newtcblisting[use counter=example]{example}[2][]{%
  title={Example~\thetcbcounter: #2},#1}

\newtcbinputlisting[use counter=example]{\examplefile}[3][]{%
  title={Example~\thetcbcounter: #2},listing file={#3},#1}

\DeclareTotalTCBox{\code}{ v O{} }
{ 
  fontupper=\ttfamily\color{black},
  nobeforeafter,
  tcbox raise base,
  colback=codebackground,colframe=white,
  top=0pt,bottom=0pt,left=0mm,right=0mm,
  leftrule=0pt,rightrule=0pt,toprule=0mm,bottomrule=0mm,
  boxsep=0.5mm,
  #2}{#1}

\patchcmd\newpage{\vfil}{}{}{}
\begin{tcbverbatimwrite}{tmp_\jobname_header.tex}

\title{Nearest Neighbors GParareal: Improving Scalability of Gaussian Processes for Parallel-in-Time Solvers
}

\author{Guglielmo Gattiglio\thanks{Department of Statistics, University of Warwick, Coventry, CV4 7AL, UK (\email{Guglielmo.Gattiglio@warwick.ac.uk}, \email{Massimiliano.Tamborrino@warwick.ac.uk}).}
\and Lyudmila Grigoryeva\thanks{Faculty of Mathematics and Statistics, University of St.~Gallen, Bodanstrasse~6, CH-9000 St.~Gallen,~Switzerland (\email{Lyudmila.Grigoryeva@unisg.ch}).} {Honorary Associate Professor, Department of Statistics, University of Warwick, Coventry, CV4 7AL, UK. (\email{Lyudmila.Grigoryeva@warwick.ac.uk})}
\and Massimiliano Tamborrino\footnotemark[2]}

\headers{nnGParareal}{G. Gattiglio, L. Grigoryeva, and M. Tamborrino}
\end{tcbverbatimwrite}
\input{tmp_\jobname_header.tex}

\ifpdf
\hypersetup{ pdftitle={Nearest Neighbors GParareal: Improving Scalability of Gaussian Processes for Parallel-in-Time Solvers} }
\fi

\begin{document}
\maketitle

\begin{tcbverbatimwrite}{tmp_\jobname_abstract.tex}
\begin{abstract}
With the advent of supercomputers, multi-processor environments and parallel-in-time (PinT) algorithms offer ways to solve initial value problems for ordinary and partial differential equations (ODEs and PDEs) over long time intervals, a task often unfeasible with sequential solvers within realistic time frames. A recent approach, GParareal, combines Gaussian Processes with traditional PinT methodology (Parareal) to achieve faster parallel speed-ups. The method is known to outperform Parareal for low-dimensional ODEs and a limited number of computer cores. Here, we present Nearest Neighbors GParareal (nnGParareal), a novel data-enriched PinT integration algorithm.  nnGParareal builds upon GParareal by improving its scalability properties for higher-dimensional systems and increased processor count. Through data reduction, the model complexity is reduced from cubic to log-linear in the sample size, yielding a fast and automated procedure to integrate initial value problems over long time intervals. First, we provide both an upper bound for the error and theoretical details on the speed-up benefits. Then, we empirically illustrate the superior performance of nnGParareal, compared to GParareal and Parareal, on nine different systems with unique features (e.g., stiff, chaotic, high-dimensional, or challenging-to-learn systems).
\end{abstract}


\end{tcbverbatimwrite}
\input{tmp_\jobname_abstract.tex}






\section{Introduction}
\label{sec:intro}

Sequential bottlenecks in time and the growing availability of high-performance computing have motivated the research and development of parallel-in-time (PinT) solvers for initial value problems (IVPs) for ordinary and partial differential equations (ODEs and PDEs), whose solutions using sequential time-stepping algorithms would be unattainable in realistic time frames. This is the case, for example, in plasma physics simulations, where traditional numerical schemes often reach saturation on modern supercomputers, leaving time parallelization as the only suitable option \cite{samaddar2019application}. Among others, molecular dynamics simulations present an additional challenge, as they involve averages over long trajectories of stochastic dynamics \cite{gorynina2022combining}.

Consider the following IVP, a system of $d \in \mathbb{N}$ ODEs (and similarly for PDEs)
\begin{equation}
\label{ODE}
\frac{d \boldsymbol{u}}{d t}=h(\boldsymbol{u}(t), t) \text { on } t \in\left[t_0, t_N\right] \text {, with } \boldsymbol{u}\left(t_0\right)=\boldsymbol{u}^0, \enspace N\in \mathbb{N},
\end{equation}
where $h: \mathbb{R}^d \times\left[t_0, t_N\right] \longrightarrow \mathbb{R}^d$ is a smooth multivariate function, $\boldsymbol{u}:\left[t_0, t_N\right] \longrightarrow \mathbb{R}^d$ is the time dependent vector solution, and $\boldsymbol{u}^0 \in \mathbb{R}^d$ are the initial values at $t_0$. Over the years, various PinT approaches have been proposed, differing in their perspectives on solving the problem (see, e.g., \cite{dutt2000spectral,lions2001resolution,maday2008parallelization} and \cite{ong2020applications} for a review). 

In this paper, we focus on improving the scalability performance of a particular PinT solver, GParareal \cite{pentland2023gparareal}, reducing its computational cost while maintaining the same accuracy. GParareal is built upon Parareal~\cite{lions2001resolution}, an iterative algorithm  that uses two solvers of different granularity: a fine one, denoted by $\f$, assumed accurate but slow, and a coarse one, denoted by $\g$, faster but less precise. First, the time span is divided into $N$ subintervals and the $\g$ solver is run sequentially on one core. Then, the computationally expensive numerical integrator $\f$ is run in parallel on $N$ cores, one for each time interval, starting from the initial condition provided by $\g$. All initial conditions are then iteratively
updated using an evaluation of $\g$ and a correction term based on the discrepancy $\f-\g$ between the two solvers computed at the previous iteration. Parareal stops once a desired tolerance is met, with the convergence of its solution to the true one obtained by running $\mathcal{F}$ serially guaranteed under certain mathematical conditions (see Section~\ref{sec:para}). The algorithm proceeds in a Markovian manner, relying only on the last iteration's data to propagate the algorithm further. This is potentially wasteful, as previously generated data can be leveraged within a Machine Learning framework. In particular, neural networks and their physics-informed versions (PINNs)~\cite{cuomo2022scientific} have been used to learn and approximate different components of the procedure, such as the fine solver $\f$,  reducing the running time, or the course solver $\g$, aiming to improve the course approximation, favoring a faster convergence (see, e.g., \cite{agboh2020parareal, ibrahim2023parareal,nguyen2023numerical,yalla2018parallel}). Recently, Parareal has also been used to speed-up the training of PINNs in \cite{meng2020ppinn}.

However, many of these studies precisely state neither the number of used processors, nor the computational cost of neural networks' training, as the latter is often assumed to be ``offline'' and thus not included in Parareal's runtimes.

Differently from Parareal, GParareal uses $d$ independent Gaussian Processes (GPs) to model the $d$ correction terms $(\f-\g)_s$, one per coordinate $s=1,\ldots, d$. Each GP is trained on a dataset consisting of all previously computed correction terms $\f-\g$, which are approximately $O(kN)$ at iteration $k$. By using all available data, as opposed to just those from the last iteration, GParareal achieves faster convergence than Parareal. Moreover, it is less susceptible  to poor starting conditions, and can reuse legacy data to provide an even faster convergence when applied to the same ODE/PDE with different initial conditions. 
However, although  GParareal needs fewer iterations than Parareal to converge, it has a higher computational cost per iteration, which impacts its overall running time. Indeed, the cubic cost of inverting the GP covariance matrix raises its computational complexity at iteration $k$ to $O((kN)^3)$.  As the algorithm proceeds, more data are accumulated until the model cost of the GP dominates or is comparable to that of running the fine solver, making GParareal no longer a viable/desirable algorithm. A further limitation comes from the fact that the training of $d$ independent univariate GPs in parallel requires $d$ cores, making it thus not suitable for high-dimensional systems of ODEs/PDEs.

In this work, we derive a new algorithm, Nearest Neighbors GParareal (nnGParareal), which reduces the computational cost of GParareal while maintaining or improving its performance, yielding better scalability properties in both $d$ and $N$. To achieve this goal, we replace the GPs within Parareal with the nearest neighbors GPs (nnGP) \cite{vecchia1988estimation}. We show that the same accuracy for learning the discrepancy $\f-\g$ can be obtained by training the GPs on a reduced dataset of size $m$, consisting of the nearest neighbors for the nnGP. We illustrate how we found $m$ to be independent of $N$ and $k$, with small values between $15-20$ proven sufficient. This yields a far more scalable algorithm at a computational complexity of $O(Nm^3+N\log(kN))$ at iteration $k$. We verify this empirically on both ODEs and PDEs.

The paper is structured as follows. In Sections~\ref{sec:para} and \ref{sec:gpara}, we recall Parareal and GParareal, respectively. The proposed new algorithm, nnGParareal, is introduced in Section~\ref{sec:nngp}. There, we describe its key features, derive its computational complexity, and provide an upper bound for its numerical error at a given iteration.  In Section~\ref{sec:num_exp}, we test nnGParareal on nine systems of ODEs and PDEs with unique features (e.g., challenging to learn from data, stiff, chaotic, non-autonomous, or high-dimensional systems), demonstrating the improved scalability, both in terms of computer cores $N$ and system dimension $d$. The discussion in Section~\ref{sec:discussion} concludes the paper.

\vspace{0.2cm}

\noindent {\bf Notation.} Column vectors are denoted by bold lowercase or uppercase symbols like $\boldsymbol{v}$ or $\boldsymbol{V}$. Given a vector $\boldsymbol{v} \in \mathbb{R}  ^n $, we denote its entries by $v_i$, 
with $i \in \left\{ 1, \dots, n
\right\} $, its Euclidean norm by $\| \boldsymbol{v}\|  $, and its infinity norm by $\| \boldsymbol{v}\|_{\infty}$. We denote by $\mathbb{R}^{n \times  m }$ the space of $\mathbb{R}$-valued $n\times m$ matrices ($m, n \in \mathbb{N} $) for which we use uppercase symbols. Given $A\in \mathbb{R}^{n \times  m }$, we write its components as $A _{ij} $, while $A_{(\cdot, j)}$ and $A_{(i, \cdot)}$ are used for the $j$th column and the $i$th row of $A$, respectively, $i \in \left\{ 1, \dots, n\right\} $, $j \in \left\{ 1, \dots m\right\} $. For any square matrix $A \in \mathbb{R}  ^{n \times n}$,   ${\rm det}(A)$ denotes its determinant. For any real matrix $A$, $A^\top$ denotes its transpose, while $\mathbb{I}_{n}$ is used for the identity matrix of dimension $n$.  We use calligraphic symbols, like $\mathcal{D}$, $\mathcal{U}$, $\mathcal{Y}$, for sets of vectors, with $\mathcal{F}$ and  $\mathcal{G}$ reserved for the fine and coarse solvers, respectively. Finally, wej write $a \vee b$ to denote the maximum between $a,b\in \mathbb{R}$.

\section{Parareal}
\label{sec:para}

Consider $N$ available computer cores and the time domain partitioned into $N$ sub-intervals of equal length. The goal of Parareal is to solve \eqref{ODE} on $[t_0,t_N]$ by solving {\it in parallel} the corresponding $N$ IVPs
\[
\frac{d \boldsymbol{u}_i}{d t}=h\left(\boldsymbol{u}_i\left(t \mid \boldsymbol{U}_i\right), t\right), \quad  t \in\left[t_i, t_{i+1}\right], \quad \boldsymbol{u}_i\left(t_i\right)=\boldsymbol{U}_i, \text{ for } i=0,\ldots,N-1,
\]
where $\boldsymbol{u}_i\left(t \mid \boldsymbol{U}_i\right)$ denotes the solution of the IVP \eqref{ODE} at time $t$, with initial condition $\boldsymbol{u}(t_i)=\boldsymbol{U}_i$.
We note that the only known initial value is $\boldsymbol{U}_0 =\boldsymbol{u}^0\in \mathbb{R}^d$ at time $t_0$. In contrast, the other initial conditions $\boldsymbol{U}_i$, $i=1,\ldots, N-1$, need to be determined such that they satisfy the continuity conditions $\boldsymbol{U}_i=\boldsymbol{u}_{i-1}\left(t_i|\boldsymbol{U}_{i-1}\right)$. Parareal aims to derive such initial conditions iteratively, running the accurate but computationally expensive solver $\f$ in parallel, and the less precise but computationally cheaper solver $\g$ sequentially. These solvers can be either two different solvers, or the same integrator with different time steps. 

We now describe in detail how Parareal works and exemplify it for a simple one-dimensional ODE in Figure~\ref{fig:para_evolution}.
Let us denote by $\boldsymbol{U}_i^k$ the Parareal approximation of $\boldsymbol{U}_i=\boldsymbol{u}_i(t_i)$ at iteration $k\geq 0$. At iteration $k=0$, Parareal initializes $\{\boldsymbol{U}_{i}^{0}\}_{i=1}^{N-1}$ using a \textit{sequential} run of the coarse solver $\g$ (the red line in the left plot of Figure~\ref{fig:para_evolution}) on one core, obtaining $\boldsymbol{U}_i^0=\g(\boldsymbol{U}_{i-1}^0)$, $i=1,\ldots, N-1$. At iteration $k\geq 1$, such approximations are then propagated \textit{in parallel} via $\f$ on $N$ cores, obtaining $\f(\boldsymbol{U}_{i-1}^{k-1})$, $i=1,\ldots, N$. Although more accurate than those computed via $\g$, these approximations are, in general, still different from the true values $\boldsymbol{U}_i$, since $\f$ is started from the (possibly) wrong initial conditions. To correct this, these values are updated \textit{sequentially}, using the so-called predictor-corrector rule
\begin{equation}
    \boldsymbol{U}_i^{k}=\mathscr{G}(\boldsymbol{U}_{i-1}^{k})+\mathscr{F}(\boldsymbol{U}_{i-1}^{k-1})-\mathscr{G}(\boldsymbol{U}_{i-1}^{k-1})=\mathscr{G}(\boldsymbol{U}_{i-1}^{k})+(\mathscr{F}-\mathscr{G})(\boldsymbol{U}_{i-1}^{k-1}),
\label{eq:update_rule}
\end{equation}
with $i=1, \ldots, N-1$, $k\geq 1$  \cite{gander2007analysis}. Here, the prediction $\mathscr{G}(\boldsymbol{U}_{i-1}^{k})$, run \textit{sequentially} at iteration $k$, is corrected by adding the discrepancy $\f-\g$ between the fine and the coarse solvers computed at the previous iteration $k-1$. Given some pre-defined accuracy level $\epsilon>0$, the Parareal solution \eqref{eq:update_rule} is then considered to have $\epsilon$-converged up to time $t_L\leq t_N$, if  solutions across consecutive iterations are close enough. More precisely, it holds that
    \begin{equation}
    ||\boldsymbol{U}^k_i-\boldsymbol{U}^{k-1}_i||_{\infty} < \epsilon,\quad \forall i\leq L\leq N-1.
    \label{eq:stop_crit}
\end{equation}
This is a standard stopping criterion for Parareal, but other criteria are possible (see, e.g., \cite{samaddar2019application,samaddar2010parallelization}). In this Parareal implementation, all converged Parareal approximations $\boldsymbol{U}_i^k$, $i\leq L$, are not iterated anymore to avoid unnecessary overhead \cite{elwasif2011dependency,garrido2006convergent}. Unconverged solution values $\boldsymbol{U}_i^k$, $i>L$, are instead updated during future iterations by first running $\f$ in parallel and then using the prediction-correction rule \eqref{eq:update_rule} until all Parareal solutions have converged at some iteration $K_{\rm Para}$, that is, \eqref{eq:stop_crit} is satisfied with $L=N-1$ and $k=K_{\rm Para}$. Note that when $K_{\rm Para}=N$, the Parareal solution trivially converges to that of the fine solver in a sequential manner, requiring an additional computational cost and time coming from the (unused) parallel implementation. 

\begin{figure}
     \centering
    \includegraphics[width=0.49\textwidth, trim={1.9cm 0cm 2.5cm 0cm},clip]{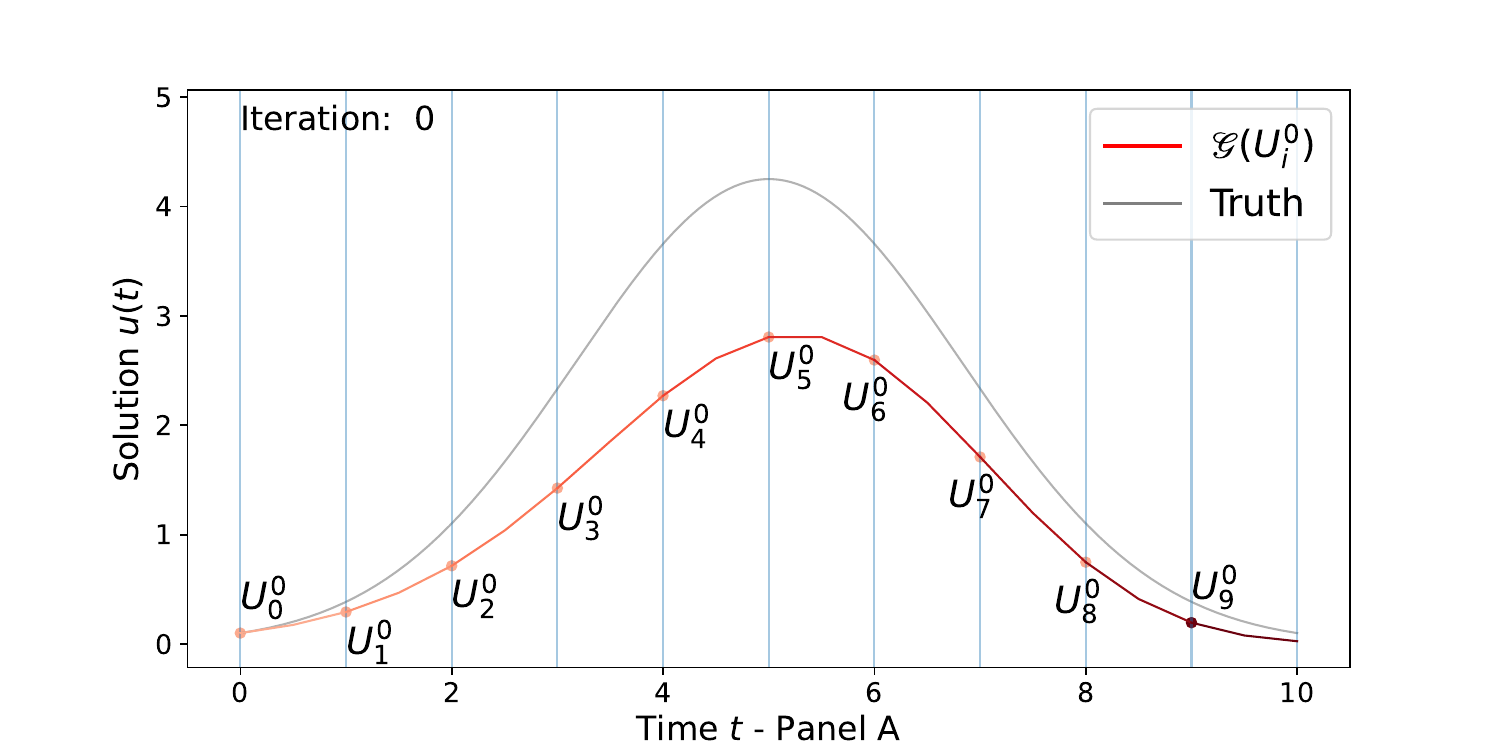}
    \includegraphics[width=0.49\textwidth, trim={2.5cm 0cm 2.2cm 0cm},clip]{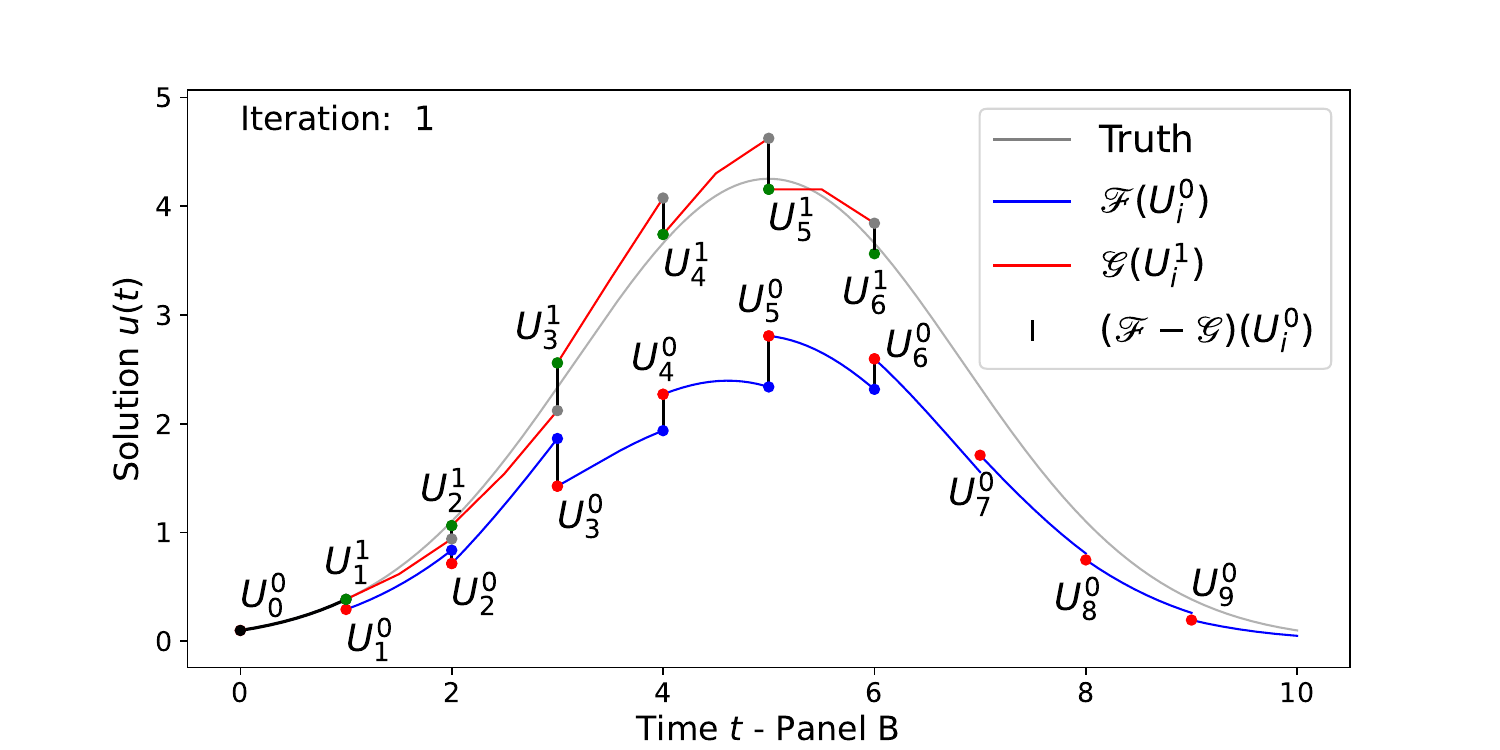}
    
    \caption{Parareal applied to an ODE in \eqref{ODE}, with $d=1$, $u(t_0)=0.1$, $t_0=0$, $t_N=10$, $N=10$. {\bf Panel A}: The gray line represents the (typically unavailable) true solution computed sequentially by $\f$. At iteration $k=0$, the coarse solver $\g$ is run sequentially (red line) to obtain initial starting conditions $U_i^0$, $i=1,\ldots, 9$. {\bf Panel B}: At iteration $k=1$, such values (red dots) are propagated in parallel via the fine solver $\f$ (blue lines). Then, $\g$ is run from $U_{i-1}^1$, $i=2,\ldots, 9$, for one time-interval (red lines), obtaining $\g(U_{i-1}^1)$ (grey dots), which is then added to $\f(U_{i-1}^0) -\g(U_{i-1}^0)$ (distance between the blue and red dots) to obtain $U_i^1$, $i=2,\ldots, 9$, (green dots) sequentially via the prediction-correction rule \eqref{eq:update_rule}. Note that \eqref{eq:update_rule} is shown only up to $i=6$ for clarity.
    }
     \label{fig:para_evolution}
\end{figure}

\section{GParareal}
\label{sec:gpara}
While Parareal leads to speed-up and more efficient use of resources with respect to running the fine solver sequentially, achieving convergence in $K_{\rm Para}\leq N$ iterations, $K_{\rm Para}$ can still be a sensible fraction of $N$, which leaves room for improvements (see Table~\ref{tab:nonaut_scaling}). GParareal was proposed for this \cite{pentland2023gparareal}. Instead of computing the predictor-corrector \eqref{eq:update_rule} at iteration $k$ using information from the previous iteration $k-1$, as done in Parareal,
GParareal uses data from the current iteration $k$, using the following expression
\begin{equation}
    \boldsymbol{U}_i^{k} = \mathscr{F}(\boldsymbol{U}_{i-1}^{k}) = (\f-\g+\g )(\boldsymbol{U}_{i-1}^{k}) =\g (\boldsymbol{U}_{i-1}^{k})+(\f-\g )(\boldsymbol{U}_{i-1}^{k}),
    \label{eq:gpara_update}
\end{equation}
with $i=1,\ldots, N-1$, and $k\geq 1$. 
Notice that stopping on the first equality would require a serial computation of $\f(\boldsymbol{U}_{i-1}^{k})$, defeating the goal of parallelization. Instead, Gaussian processes are used to predict the correction term $\f-\g$ using the known (updated) initial condition $\boldsymbol{U}_{i-1}^{k}$. When solving a $d$-dimensional system, GParareal uses the information about all the $d$ coordinates while producing predictions for each of the $d$ variables independently. 
This is done by using 
a dedicated one-dimensional GP per coordinate instead of a multi-output GP 
for $(\f-\g )(\boldsymbol{U}_{i-1}^{k}) \in \R^d$, which reduces the runtime and allows parallelization.

We now describe how the GPararel algorithm works. The first step is defining a GP prior for learning the correction function for each of the $d$ coordinates  as
\begin{equation}
(\f-\g)_s \sim GP(\mu_{\rm GP}^{(s)}, \mathcal{K}_{\rm GP}), \enspace s=1,\ldots, d,
    \label{eq:gp_prior}
\end{equation}
where $\mu_{\rm GP}^{(s)}: \R^d \longrightarrow{} \R$ is the mean and $\mathcal{K}_{\rm GP}: \R^d \times \R^d \longrightarrow \R$ is the variance kernel functions, respectively. Here, the mean function is taken to be zero, while for the variance, we consider the squared exponential kernel, defined for any $\boldsymbol{U}, \boldsymbol{U}'\in \mathbb{R}^d$  as
\begin{equation}
    \mathcal{K}_{\rm GP}(\boldsymbol{U},\boldsymbol{U}') = \sigma_{\rm o}^2 \exp(-\|\boldsymbol{U}-\boldsymbol{U}'\|^2 / \sigma_{\rm i}^2),
    \label{eq:se_kern}
\end{equation}
where $\sigma_{\rm o}^2$ and $\sigma_{\rm i}^2$ are the output and input length scales, respectively. At every Parareal iteration, we collect several evaluations of the correction function $\f-\g$
which are stored in a dataset. For every iteration $k>0$, denote 
the collection of all the inputs $\boldsymbol{U}_{i-1}^j\in\mathbb{R}^d$ and of all the outputs $(\f-\g)(\boldsymbol{U}_{i-1}^j)\in\mathbb{R}^d$, $i=1,\ldots, N$, $j=0,\ldots, k-1$, by $\mathcal{U}_k$ and $\mathcal{Y}_k$, respectively (we will use $U\in \mathbb{R}^{Nk\times d}$  and  $Y\in \mathbb{R}^{Nk\times d}$ to denote their matrix analogs, respectively
). The dataset $\mathcal{D}_k$ is composed of pairs of inputs and their corresponding outputs, that is
\begin{equation}\label{Dk}
\mathcal{D}_k:= \{(\boldsymbol{U}_{i-1}^{j}, (\f-\g)(\boldsymbol{U}_{i-1}^{j})), \enspace i=1,\ldots,N, \enspace j=0,\ldots,k-1\}.
\end{equation}
Conditioning the prior (\ref{eq:gp_prior}) using the observed data $\mathcal{D}_k$, the posterior distribution for each $s$th coordinate, $s=1,\ldots, d$, of $\boldsymbol{U}'\in\mathbb{R}^d$ is given by
\begin{equation}
    ((\f-\g)(\boldsymbol{U}'))_s|\mathcal{D}_k \sim \; GP(\mu_{\mathcal{D}_k}^{(s)}(\boldsymbol{U}'), \mathcal{K}_{\mathcal{D}_k}(\boldsymbol{U}',\boldsymbol{U}')),\label{eq:gp_posterior}
\end{equation}
with posterior mean $\mu_{\mathcal{D}_k}^{(s)}(\boldsymbol{U}') \in \mathbb{R}$ and variance $\mathcal{K}_{\mathcal{D}_k}(\boldsymbol{U}', \boldsymbol{U}') \in \mathbb{R}^+$ given by
\begin{eqnarray}
    \mu^{(s)}_{\mathcal{D}_k}(\boldsymbol{U}')&=& \mathcal{K}({U}, \boldsymbol{U}')^\top \left(\mathcal{K}({U},{U}) + \sigma_{\rm reg}^2 \mathbb{I}_{Nk}\right)^{-1} {Y}_{(\cdot,s)}, \label{eq:gp_posterior_m}\\
    \mathcal{K}_{\mathcal{D}_k}(\boldsymbol{U}', \boldsymbol{U}') &= &\mathcal{K}_{\rm GP}(\boldsymbol{U}', \boldsymbol{U}') - \mathcal{K}({U}, \boldsymbol{U}')^\top \left(\mathcal{K}({U},{U}) + \sigma_{\rm reg}^2 \mathbb{I}_{Nk}\right)^{-1} \mathcal{K}({U}, \boldsymbol{U}'), \label{eq:gp_posterior_v}
\end{eqnarray}
where $\mathcal{K}({U}, \boldsymbol{U}')\in \mathbb{R}^{Nk}$ is a vector of covariances between every input collected in ${U}$ and  $\boldsymbol{U}'$ defined as $(\mathcal{K}({U},\boldsymbol{U}'))_{r}=\mathcal{K}_{\rm GP}({{U}_{(r, \cdot)}}^\top, \boldsymbol{U}' )$, $r=1,\ldots, Nk$, and $\mathcal{K}({U},{U})\in \mathbb{R}^{Nk\times Nk}$ is the covariance matrix, with $(\mathcal{K}({U}, {U}))_{r,q} = \mathcal{K}_{\rm GP}({{U}_{(r,\cdot)}} ^\top, {{U}_{(q,\cdot)}}^\top )$, $r,q=1,\ldots, Nk$. Note that \eqref{eq:gp_posterior_v} slightly differs from the posterior variance in \cite{pentland2023gparareal}, as we added a regularization term $\sigma_{\rm reg}^2$, known as nugget, jitter, or regularization strength, to improve the condition number of the covariance matrix upon inversion (see Supplementary Material \ref{app:impl_det} for the detailed discussion).

An important step in the training of the GP is the optimization of the hyperparameters $\boldsymbol{\theta}:=(\sigma_{\rm i}^2, \sigma_{\rm o}^2, \sigma_{\rm reg}^2)$ for each of the $d$ components,
which is done at every iteration by numerically maximizing the marginal log-likelihood (see Supplementary Material \ref{app:impl_det} for implementation details):
\begin{equation}
   \log p({Y}_{(\cdot,s)}|{U}, \boldsymbol{\theta}) \propto - {{Y}_{(\cdot,s)}}^\top[\mathcal{K}({U},{U})+\sigma_{\rm reg}^2 \mathbb{I}_{Nk}]^{-1}{Y}_{(\cdot,s)} - \log {\rm det}(\mathcal{K}({U},{U})),
    \label{eq:gp_llik}
\end{equation}
where $\mathcal{K}(\cdot, \cdot)$ depends on $\boldsymbol{\theta}$ through the kernel $\mathcal{K}_{\rm GP}$ in \eqref{eq:se_kern}. 
Once the optimal hyperparameters $\boldsymbol{\theta}^*$ are determined, the $d$ scalar posterior distributions of the GP~\eqref{eq:gp_posterior} can be used to make a prediction at $\boldsymbol{U}'$.  

Going back to GParareal, the training of the GP happens once per iteration after new data are collected from the parallel run of $\f$. 
The (posterior) prediction~\eqref{eq:gp_posterior} is then computed sequentially, and the derived mean~\eqref{eq:gp_posterior_m} is used as a  point prediction for $\f-\g$ in the  GParareal predictor-corrector rule~\eqref{eq:gpara_update}, as in \cite{pentland2023gparareal}. An alternative (not considered in this work) could be to propagate the whole posterior distribution by also including its variance, as done in the stochastic version of Parareal proposed in \cite{pentland2022stochastic}, as it would provide uncertainty quantification for the final Parareal solution.

The flexibility and empirical performance of GPs come with the high computational cost required to invert the covariance matrix in \eqref{eq:gp_posterior_m} when making a prediction, and to evaluate the log-likelihood~\eqref{eq:gp_llik} during the hyper-parameter optimization. In particular, the cost of training the GP at iteration $k$ is $O((kN)^3)$, cubic in the number of time intervals and processors $N$.  This raises a trade-off. On the one hand, it is desirable to use all the available computational resources to speed up the algorithm, with modern clusters in Tier~2 HPC facilities offering $10^2$-$10^4$ processor cores. On the other hand, $N$ needs to be relatively low as, otherwise, the cost of the GP will negate any parallel gains. Indeed, as seen in Figure~\ref{fig:nonaut_scal_speedup}, a few hundred cores are sufficient to negatively impact the performance.

Computationally cheaper alternatives of matrix inversion exist and are widely used, for example, in parallel kernel ridge regression (KRR) \cite{rudi2017falkon}. The impact of replacing the matrix inverses with such operations has been extensively studied in the GP and KRR literature (see, e.g., \cite{meanti2020kernel} and references therein). In this work, we pursue a different approach, as described in the following section.

\begin{table}[t]
{\footnotesize
    \centering
    \begin{tabular}{c|c|c}%
         Model & Notation & $\fhat(\boldsymbol{U}_{i-1}^{k})=$%
         \\[0.1cm]
         \hline
         &&%
         \\%
         Parareal & 
         $\fhat_{\rm Para}$ & 
         $(\f-\g)(\boldsymbol{U}_{i-1}^{k-1})$%
         \\[0.4cm]
         \multirow{ 2}{*}{$m$-nn Parareal} & \multirow{ 2}{*}{$\fhat_{m\mhyphen\text{nn}}$}&$\widehat{(\f-\g)}(\boldsymbol{U}_{i-1}^{k})=\sum_{l=1}^m w_l \mathbf{y}^{(l\mhyphen\textrm{nn})}_{\boldsymbol{U}_{i-1}^{k}}$ as in \eqref{eq:fhat_mnn} and \eqref{eq:fhat_mnn:U}
         \\[0.4cm]
         \multirow{ 2}{*}{GParareal} & \multirow{ 2}{*}{$\fhat_{\rm GPara}$} &$\widehat{(\f-\g)}(\boldsymbol{U}_{i-1}^{k})=(\mu_{\mathcal{D}_{k}}^{(1)}(\boldsymbol{U}_{i-1}^{k}), \ldots, \mu_{\mathcal{D}_{k}}^{(d)}(\boldsymbol{U}_{i-1}^{k}))^\top$,%
         \\
         [0.15cm]
          &  & with $\mu_{\mathcal{D}_{k}}^{(s)}(\boldsymbol{U}_{i-1}^{k})$, $1\leq s\leq d$, as in (\ref{eq:gp_posterior_m}) and $\mathcal{D}_{k}$ in \eqref{Dk}%
          \\[0.5cm]
         \multirow{ 4}{*}{nnGParareal} & 
         \multirow{ 4}{*}{$\fhat_{{\rm nnGPara}}$} &$\widehat{(\f-\g)}(\boldsymbol{U}_{i-1}^{k})= (\mu_{\mathcal{D}_{i-1,k}}^{(1)}(\boldsymbol{U}_{i-1}^{k}), \ldots, \mu_{\mathcal{D}_{i-1,k}}^{(d)}(\boldsymbol{U}_{i-1}^{k}))^\top$,  
         \\[0.15cm]
        & 
        & where, for $s=1,\ldots, d$, $\mu_{\mathcal{D}_{i-1,k}}^{(s)}(\boldsymbol{U}_{i-1}^{k})$ is as in (\ref{eq:gp_posterior_m}), with   %
\\[0.2cm]
          &  &   $\mathcal{D}_{i-1,k}:= \{(\boldsymbol{U}^{(l\mhyphen\textrm{nn})}_{\boldsymbol{U}_{i-1}^k}, \mathbf{y}^{(l\mhyphen\textrm{nn})}_{\boldsymbol{U}_{i-1}^{k}}),\enspace l=1,\ldots,m \}\subset \mathcal{D}_k$\\\hline
    \end{tabular}
    }
    \caption{Definition of the correction term $\fhat$ of the predictor-corrector rule (\ref{eq:update_rule_generic}) for different Parareal variants.}
    \label{tab:fhat}
\end{table}
\section{Nearest Neighbor GParareal (nnGParareal)}
\label{sec:nngp}
We start this section by reformulating the Parareal and GParareal predictor-corrector rules, interpreting them as a learning task, and introduce our approach to data dimensionality reduction based on the choice of $m\in \mathbb{N}$ nearest neighbors (nns). In Subsections~\ref{sec:nngp_algo} and \ref{sec:subset_choice}, we describe the nnGParareal algorithm, and discuss alternative choices for selecting the reduced dataset beyond nns, respectively. In Subsection~\ref{sec:choose_m}, we explain how to choose the number $m$ of neighbors, while in Subsection \ref{sec:comp_complx} we provide the computational complexity and speed-up calculation of nnGParareal. We conclude with Subsection~\ref{sec:error_analysis}, where we derive an upper bound of the nnGParareal error.

\subsection{Rethinking the predictor-corrector rule}\label{Section4.1}
The predictor-corrector rule (\ref{eq:update_rule}) can be rewritten as 
\begin{equation}
    \boldsymbol{U}_i^k = \g(\boldsymbol{U}_{i-1}^{k}) + \fhat(\boldsymbol{U}_{i-1}^{k}),
    \label{eq:update_rule_generic}
\end{equation}
where $\fhat: \R^d \longrightarrow \R^d$, $\boldsymbol{U} \longmapsto \fhat(\boldsymbol{U})$, specifies how the correction or discrepancy function $\f-\g$ is computed or approximated based on some observation $\boldsymbol{U} \in \R^d$. Both Parareal \eqref{eq:update_rule} and GParareal \eqref{eq:gpara_update} update rules can be written in terms of $\fhat$, as shown in Table~\ref{tab:fhat}. A natural way to compare different approaches is through their (prediction) error $\|(\f-\g)(\boldsymbol{U}) - \fhat(\boldsymbol{U})\|$. We can expect lower prediction errors to yield faster convergence due to the increased accuracy in correcting the error committed by $\g$.

Interestingly, the prediction error at a test observation $\boldsymbol{U}' \in \R^d$ of the original Parareal predictor-corrector rule (\ref{eq:update_rule}), $\fhat_{\rm Para}$, at iteration $k$ is similar (in the sense of evolutions in the first few intervals) to that obtained by the $m$-nearest neighbor ($m$-nn) algorithm\footnote{Although commonly referred to as $k$-nn, here we use the letter $m$ to indicate the number of nns to avoid confusion with $k$, the iteration number.}, as discussed below. Such algorithm is a non-parametric supervised learning method commonly used as a baseline for regression tasks. For a given $\boldsymbol{U}'$, the algorithm determines the $m$ closest points $\boldsymbol{U}^{(l\mhyphen\textrm{nn})}_{\boldsymbol{U}'}$, $l=1,\ldots, m$, in the dataset $\mathcal{U}_k$ at iteration $k$, and then returns as prediction the weighted linear combination of the outputs $\f-\g$ of the $m$ nearest neighbors, with weights $w_l, l=1,\ldots, m$. That is, the $m$-nn prediction $\hat f_{m\mhyphen\textrm{nn}}$ is given by
\begin{equation}
\fhat_{m\mhyphen\text{nn}}(\boldsymbol{U}') = \sum_{l=1}^m w_l \mathbf{y}^{(l\mhyphen\textrm{nn})}_{\boldsymbol{U}'}, 
    \label{eq:fhat_mnn}
\end{equation}
with 
$\mathbf{y}^{(l\mhyphen\textrm{nn})}_{\boldsymbol{U}'}=(\f-\g)(\boldsymbol{U}^{(l\mhyphen\textrm{nn})}_{\boldsymbol{U}'})\in \mathcal{Y}_k$ being the correction term $\f-\g$ evaluated at the $l$th nn $\boldsymbol{U}^{(l\mhyphen\textrm{nn})}_{\boldsymbol{U}'}$ of $\boldsymbol{U}'$, defined as
\begin{align}
    \boldsymbol{U}^{(l\mhyphen\textrm{nn})}_{\boldsymbol{U}'}&\in_R\{\boldsymbol{U}\in\mathcal{U}_k: \mathcal{U}^{\rm norm}_{{\boldsymbol{U}',}(l)}=\boldsymbol{U}\},
   \label{eq:fhat_mnn:U}
\end{align}
where $\in_R$ means that, provided that the set has cardinality more than $1$, the $l$th nn is randomly chosen from it. Here, $\mathcal{U}^{\rm norm}_{{\boldsymbol{U}'},(l)}$ denotes the $l$th ordered statistics of the sample formed out of Euclidean distances between $\boldsymbol{U}'$ and each of the observations in $\mathcal{U}_k$, that is\footnote{Recall that the pair $(\mathcal{U}_k,\mathcal{Y}_k)$ define the dataset $\mathcal{D}_k$ at iteration $k$.}
\begin{equation*}\label{distance}
\mathcal{U}^{\rm norm}_{\boldsymbol{U}'}:=\{\|\boldsymbol{U}_{i-1}^{j}-\boldsymbol{U}'\|: \boldsymbol{U}_{i-1}^{j} \in \mathcal{U}_k, \enspace i=1,\ldots, N, j=0,\ldots, k-1\}.
\end{equation*}
A special case is when $m=1$, for which the prediction equals the output of the closest point. While the Euclidian distance is a natural choice, alternatives like Minkowski or Mahalanobis distances may also be used. However, our findings (not shown) indicate that the choice of distance has a marginal impact on performance. Therefore, we consider the Euclidean distance in this work.
\begin{figure}[t]
     \centering
    \includegraphics[width=1\textwidth]{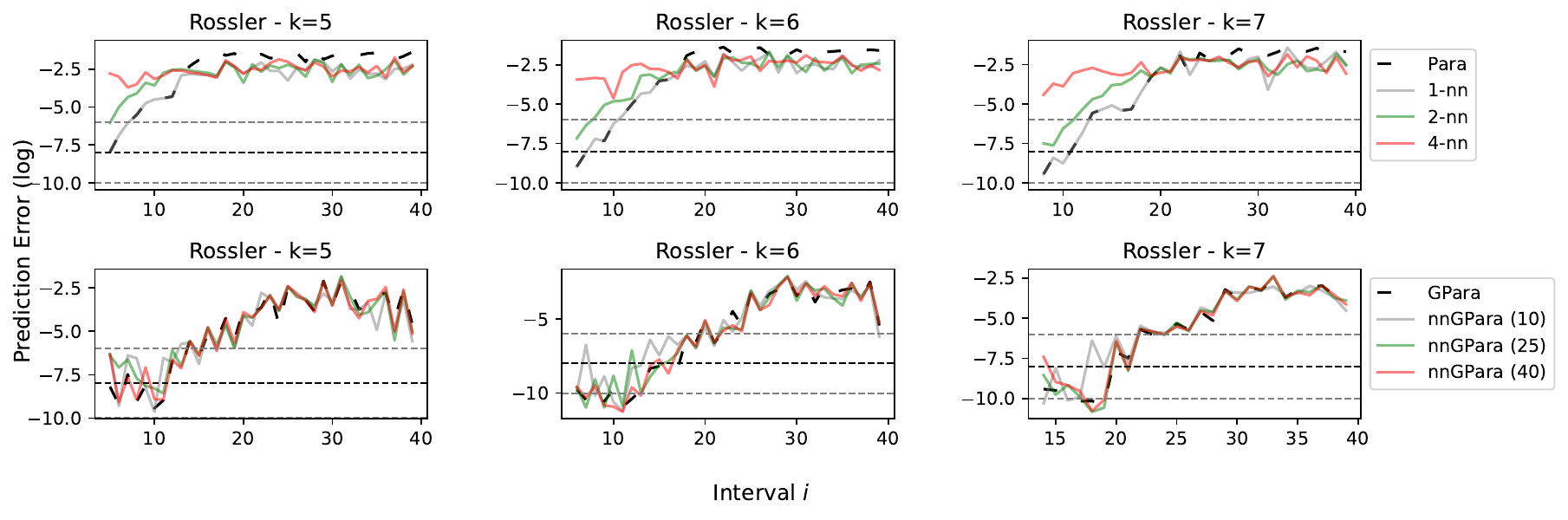}
    \caption{Comparison of the prediction errors of Parareal versus $m$-nn (top), and of GParareal versus nnGParareal (bottom) for iteration $k$. The $x$-axes start from the first unconverged interval $i$. Parareal and $1$-nn match on the first few intervals, with the worsening performance of $m$-nn for higher $m$. Instead, GParareal and nnGParareal show consistent alignment across $m$ values (see parentheses in the legend).}
     \label{fig:rossler_pred_err_both}
\end{figure}

The connection between $\fhat_{\rm Para}$ and $\fhat_{m\mhyphen\text{nn}}$ is illustrated in the top row of Figure~\ref{fig:rossler_pred_err_both}, where the $\fhat_{\rm Para}$ prediction error is compared to that of $\fhat_{m\mhyphen\text{nn}}$ across intervals $i$ on the $x$-axes for different values of $m$ and iterations $k$. Here, we choose uniform weights for the $m$-nn algorithm, i.e. $w_l=1/m, l=1,\ldots, m$, as commonly done in the literature. Note that in each subplot, the leftmost value of $i$ on the x-axis is the first unconverged interval. Parareal and $1$-nn behave equivalently for the first few intervals and then diverge. Surprisingly, %
the prediction error decreases/deteriorates in $m$ (with almost flat error curves for $m\in[10,15]$, results not shown), requiring thus more iterations than Parareal to convergence when used as predictor rule in \eqref{eq:update_rule_generic}, in an algorithm that we call $m$-nn Parareal.
Intuitively, this is because $\hat f_{m-{\rm nn}}$ is not optimal when using uniform weights, since closer inputs carry more information in a local smoothing task. We refer to Supplementary Material~A for an illustration of the data distribution for the $2$-dimensional Brussellator system. %

While increasing $m$ has a detrimental effect for $m$-nn Parareal compared to Parareal, this is not the case for $m$-nnGParareal due to the fact that uniform weights are replaced with ``optimal'' weights computed by the GP, as GPs have been extensively used as non-parametric smoothers in deterministic regression tasks (see, e.g.,~\cite{booker1999rigorous, huang2006sequential,mohammadi2016analytic,peng2014choice,ranjan2011computationally,taddy2009bayesian}). As shown in~\cite{vecchia1988estimation}, $m$-nnGPs, that is GPs trained on $m$-nns instead of the whole dataset, may 
maintain the same accuracy as the full GP by selecting a smaller but informative number of data points,  reducing the dimension of the associated covariance matrix, and, thus, the model training time. Figure~\ref{fig:rossler_pred_err_both} (bottom) compares the prediction errors of $\fhat_{\rm GPara}$ and $\fhat_{\rm nnGPara}$, both described in Table~\ref{tab:fhat}. As expected, they show a better alignment than Parareal and $m$-nn Parareal, with the prediction error of $m$-nnGP approaching that of the GP as $m$ increases, being already very close for $m=40$. This analogy is also observed in the two algorithms, both converging in a similar number of iterations.

\subsection{The nnGParareal algorithm}
\label{sec:nngp_algo}

The nnGParareal algorithm consists of the following steps:
\medskip
\begin{enumerate}[leftmargin=*]%
    \item Rescale %
    the ODE/PDE system via change of variables such that each coordinate takes values in $[-1,1]$. 
    \medskip
    \item Iteration $k=0$: take $\boldsymbol{U}_{0}^0=\boldsymbol{u}^0$, and initialize $\boldsymbol{U}_{i}^0$, $i=1,\ldots,N-1$, using a \textit{sequential} run of the coarse solver $\g$, that is $\boldsymbol{U}_{i}^0 = \g(\boldsymbol{U}_{i-1}^0)$. 
    Set $\mathcal{D}_0=\emptyset$. 
    \medskip
    \item Iterations $k \geq 1$: First, compute 
    $\f(\boldsymbol{U}_{i-1}^{k-1})$, $i=1,\ldots,N$, in parallel. 
    Update the dataset by adding the newly computed observations $\mathcal{D}_k = \mathcal{D}_{k-1} \cup \{(\boldsymbol{U}_{i-1}^{k-1}, (\f-\g)(\boldsymbol{U}_{i-1}^{k-1})), i=1,\ldots,N\}$. 
    Then, for each unconverged solution in the interval $i=L+1,\ldots, N-1$, update the initial conditions $\boldsymbol{U}_i^{k}$ using the predictor-corrector rule~\eqref{eq:update_rule_generic} with $\fhat=\fhat_{\rm nnGPara}$. 
In particular, first find the $m$ closest neighbors to $\boldsymbol{U}_{i-1}^{k}$ in $\mathcal{U}_{k}$, i.e.  $\mathcal{U}_{i-1,k}:=\{\boldsymbol{U}^{(l\mhyphen\textrm{nn})}_{\boldsymbol{U}_{i-1}^{k}}\}_{l=1}^m \subset \mathcal{U}_{k}$, and the associated correction terms $\mathcal{Y}_{i-1,k}:=\{\mathbf{y}^{(l\mhyphen\textrm{nn})}_{\boldsymbol{U}_{i-1}^{k}}\}_{l=1}^m \subset \mathcal{Y}_{k}$; the corresponding pairs of elements of these sets form $\mathcal{D}_{i-1,k}\subset \mathcal{D}_k$, as in Table \ref{tab:fhat}.
   Concatenate the elements of $\mathcal{U}_{i-1,k}$ and of $\mathcal{Y}_{i-1,k}$ to obtain ${{U}}\in \mathbb{R}^{m\times d}$  and ${{Y}}\in \mathbb{R}^{m\times d}$, respectively. Train one $m$-nnGP per system coordinate $s$ using \eqref{eq:gp_posterior_m} and \eqref{eq:gp_posterior_v}, to form the predictions $\mu^{(s)}_{\mathcal{D}_{i-1,k}}(\boldsymbol{U}_{i-1}^{k})$, $s = 1,\ldots, d$, as follows:
    \begin{enumerate}
        \item Estimate the optimal nnGP hyperparamters $\boldsymbol{\theta}^*$ by numerically maximizing the marginal log-likelihood \eqref{eq:gp_llik} from a random initial value $\boldsymbol{\theta}_0$.
        \item Repeat step (a) for several ($n_{\rm start}$) random restarts $\boldsymbol{\theta}_0$. Choose as hyperparameter $\boldsymbol{\theta}^*$  the one achieving the largest value of \eqref{eq:gp_llik}, and use it to compute the prediction $\mu^{(s)}_{\mathcal{D}_{i-1,k}}(\boldsymbol{U}_{i-1}^{k})$ with (\ref{eq:gp_posterior_m}).
    \end{enumerate}
    Form the prediction $\fhat_{\rm nnGPara}(\boldsymbol{U}_{i-1}^{k}) = (\mu^{(1)}_{\mathcal{D}_{i-1,k}}(\boldsymbol{U}_{i-1}^{k}), \ldots, \mu^{(d)}_{\mathcal{D}_{i-1,k}}(\boldsymbol{U}_{i-1}^{k}))^\top$.
        \medskip
    \item Repeat steps 3-4 until convergence, that is until \eqref{eq:stop_crit} is satisfied with $L=N-1$.
\end{enumerate}
\medskip

The rescaling in Step 1 has two advantages: first, normalization is known to improve the performance of learning models whenever the inputs have different scales. Second, the $\epsilon$-covergence criterion in (\ref{eq:stop_crit}) depends on the scale of the data for every system of interest, while the normalization allows to compare the nnGParareal performance across different ODE/PDE systems. Normalization is simple to apply in the context of ODEs/PDEs, as it reduces to a change of variables (see~\cite{langtangen2016scaling}, Section~2.1.3 for more details). 
In Step 4, we use random restarts to explore the loss landscape and avoid local optima. Note also that the operations (a) and (b) in Step 4  can be parallelized, as we placed  independent scalar nnGPs for each coordinate. 

The proposed nnGParareal has theoretical advantages over GParareal beyond the computational cost. The fact that hyperparameters $\boldsymbol{\theta}$ are optimized for every prediction (over the unconverged intervals) instead of once per iteration lends itself to a better adaptation of nnGP to local structures, where nonstationary behavior patterns are observed, especially in a spatiotemporal setting~\cite{gramacy2015local}.

\subsection{Choice of the reduced dataset}
\label{sec:subset_choice}
At each iteration $k$ and unconverged interval $i=L+1,\ldots, N-1$, nnGParareal uses a GP trained on a subset of $m$ data points from the full dataset $\mathcal{D}_k$, where the $m$ points are the $m$ nns. Although some intuition has been provided to motivate this decision in Section~\ref{Section4.1}, there exist other possibilities. When predicting $\boldsymbol{U}_i^k$, one may for example consider some intuitive \lq\lq rule-of-thumbs\rq\rq, selecting $m$ data points among observations from previous iterations, i.e., $\boldsymbol{U}_i^{k-m},\ldots, \boldsymbol{U}_i^{k-1}$, $k\geq m$ (selecting thus \lq\lq by columns\rq\rq), or from the preceding intervals, i.e., $\boldsymbol{U}_{i-m}^k,\ldots, \boldsymbol{U}_{i-1}^k$, $i\geq m$ (selecting thus \lq\lq by rows\rq\rq). Alternatively, one may incorporate into the GP kernel $\mathcal{K}_{\rm GP}$ the time information carried by interval $i$ and iteration $k$. For example, similar in spirit to~\cite{datta2016nonseparable}, the dataset $\mathcal{D}_k$ can be augmented with the tuple $(i,k)$ 
prior to selecting $m$ data points (possibly also nns) from it.
Supplementary Material~\ref{app:nngp_time} provides a comparison of the performance of these and other alternative approaches for the selection of the $m$ data points for a given subset. We show that they are either equivalent to nns or suboptimal in terms of higher number of iterations to converge or computational time. We also discuss some relevant literature.

\subsection{Choice of \texorpdfstring{$m$}{m}}%
\label{sec:choose_m}
There seems to be a consensus in the nnGP literature that the specific number of nns used has a negligible effect on performance, as long as a reasonable number of them is chosen \cite{datta2016hierarchical,datta2016nearest,gramacy2016lagp,gramacy2015local,stein2004approximating}, even as small as $10-15$~\cite{datta2016hierarchical}. This is also our finding, as observed in Figure \ref{fig:nngp_m_distr}, where we report the histograms of the numbers of iterations needed by nnGParareal to converge for several ODE systems obtained for several values of $m$ and random seeds per $m$ value. The shapes of empirical distributions indicate the consistent performance of the algorithm regardless of the value of $m$ and the specific execution. %
Further numerical investigations are reported in Supplementary Material~\ref{app:choose_m}, showing marginal improvements with increasing $m$, with a largely unaffected speed-up.

\begin{figure}
    \centering    
    \includegraphics[width=1\linewidth]{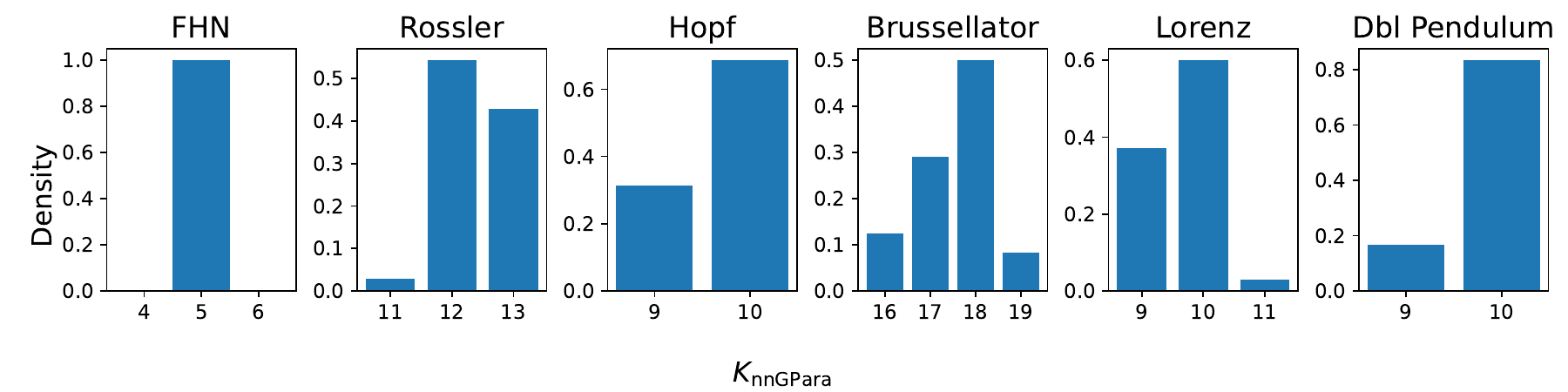}
    \caption{Histograms of the number of iterations $K_\textrm{nnGPara}$ needed by nnGParareal to converge for several ODE systems obtained using $7$ values of $m\in [10,20]$, and $5$ random seeds per each $m$ value.}
    \label{fig:nngp_m_distr}
\end{figure}

\subsection{Computational complexity}
\label{sec:comp_complx}
Let $T_\g$ and $T_\f$ be the time it takes to run $\g$ and $\f$ over one interval $[t_i,t_{i+1}]$, respectively.
Here, we compute the parallel speed-up $S_{\rm alg}\in \mathbb{R}^+$ of using the Parareal, GParareal, and nnGParareal algorithms as opposed to running $\f$ serially across $N$ intervals. The speed-up describes the relative performance of two systems processing the same problem and is computed as the ratio of the serial runtime $T_\textrm{Serial}:=N T_\f$ over the parallel runtime $T_{\rm alg}$, that is $S_{\rm alg}:=T_\textrm{Serial}/T_{\rm alg}$. $S_{\rm alg}$ can be approximated theoretically, as done below, or computed empirically (from measured runtimes including overheads), as done in Section \ref{sec:num_exp}. We refer to the latter as $\hat S_{\rm alg}$.
Let $\tmdl(k)$ be the total cost of evaluating $\fhat_\textrm{model}$ during iteration $k$, which includes both model training and predicting. If an algorithm converges in $K_{\rm alg}$ iterations, its worst-case runtime, excluding communication overheads, can be approximated by 
\begin{equation}
\label{Tparallel}\hspace{-.15cm}T_{\rm alg}\hspace{-.1cm}\approx \hspace{-.1cm}N T_{\g} + \hspace{-.1cm}\sum_{k=1}^{K_{\rm alg}} (T_{\f} + (N-k)T_{\g}+\tmdl(k))\hspace{-.1cm}=\hspace{-.1cm}K_{\rm alg} T_{\f} + (K_{\rm alg}+1)(N-\frac{K_{\rm alg}}{2})T_{\g}+\tmdl,
\end{equation}
with $\tmdl= \sum_{k=1}^{K_{\rm alg}}\tmdl(k)$. Here, the first term corresponds to the initialization cost of running the coarse solver, while the sum is the worst-case cost per iteration, assuming that only one interval converges per iteration, requiring thus $N-k$ runs of $\g$ at iteration $k$.
The speed-up is then easily computed as
\begin{equation}
    S_{\rm alg} \approx  \left( \frac{K_{\rm alg}}{N} + (K_{\rm alg}+1)\left(1-\frac{K_{\rm alg}}{2N}\right)\frac{T_{\g}}{T_{\f}}+\frac{\tmdl}{N T_\f}\right)^{-1}\leq \left(\frac{K_{\rm alg}}{N}\right)^{-1}=:S_{\rm alg}^*.
    \label{eq:speed}
\end{equation} 
From this relation, it is clear that the speed-up $S_{\rm alg}$ is maximized when $K_{\rm alg}$, the ratio $T_{\g}/T_{\f}$ and $\tmdl/T_\textrm{Serial}$ are as small as possible, with the first indicating how quickly the scheme converges, the second determining the amount of serial cost which cannot be parallelized and the third the impact of the model cost. The upper bound $S^*_{\rm alg}$ is computed ignoring all overheads (and the cost of running $\g$), model cost, and including only the unavoidable cost of running the fine solver $K_{\rm alg}$ times.

For Parareal, $\tmdl=T_{\rm Para} \in O(1)$, since $\fhat_{\rm Para}$ is a simple look-up operation. This would be the fastest parallel algorithm to converge if the three algorithms were converging at the same iteration $K$, which is typically not the case. For the GP, $\tmdl(k)=T_{\rm GP}(k)$ is the wallclock time needed to run the $d$ independent GPs in parallel at iteration $k$, which %
includes fitting the model, optimizing the hyperparameters, and making predictions. This quantity is challenging to estimate beforehand, as it depends on the size of the covariance matrix to invert and the iterations required to optimize the log-likelihood  when computing the optimal hyperparameters. Nonetheless, since the size of the dataset $\mathcal{D}_k$ for one coordinate is of order $O(kN)$ by iteration $k$, and $(N-k)$ predictions are made, $T_{\rm GP}(k)$ scales as 
\[
\underbrace{(\tfrac{d}{N} \vee 1) \; O\left(d (Nk)^2 + (Nk)^3\right)}_{\text{Training}} + \underbrace{(\tfrac{d}{N} \vee 1) \;(N-k) \; O(dNk)}_{\text{Post. mean prediction}},
\]
where  $(d/N\vee 1)$ is implied by the fact that the $d$ GPs are independent and can be parallelized over the $N$ cores, $d (Nk)^2$ corresponds to computing the kernel matrix $\mathcal{K}({X},{X})$, $(Nk)^3$ for its inversion, and $O(dNk)$ is the cost of one posterior mean evaluation. Overall, we have
\[
T_{\rm GP} =\sum_{k=1}^{K_{\rm GPara}}T_{\rm GP}(k)\in O\left((d/N \vee 1) (d \vee Nk) K_{\rm GPara}^{3 }N^{2} \right).
\]
If $K_{\rm Para}=K_{\rm GPara}$, the total cost of the GP needs to be negligible compared to $T_\textrm{Serial}$ to achieve the same speed-up $S_{\rm Para}$. For $S_{\rm GPara}$ to be as large as possible, we also need $T_{\rm GP}$ to be insignificant with respect to $K_{\rm GPara}T_{\f}$, the leading term in \eqref{Tparallel}.

Unlike GParareal, the runtime $T_{{\rm nnGPara}}$ 
of nnGParareal can be accurately estimated as the reduced subset is of constant size $m$. 
Let $n_{\rm start}$ be the number of random restarts performed during the optimization procedure, $n_{\rm reg}$ the numbers of tested jitters values $\sigma^2_{\rm reg}$ 
(see also Supplementary Material~\ref{app:impl_det}), %
and $T_{{\rm GP}_m}$ the cost of training, hyperparameter optimization, and prediction for a GP based on a dataset of size $m$. Then, the model cost of nnGParareal at iteration $k$ is
\[
T_{\rm nnGP}(k) = (N-k) \left((dn_{\rm reg}  n_{\rm start})  /N \vee 1 \right) T_{{\rm GP}_m}=(N-k)T_{\rm m-nnGP},
\]
where the rescaling term $(N-k)$ follows from the need to retrain the nnGP at every unconverged interval at iteration $k$, while $T_{\rm m-nnGP}$ is the cost of using the model to make a single prediction, including training. %
Unlike GParareal, $T_{\rm nnGP}(k)$ decreases as $k$ increases, see Figure \ref{fig:tom_lab_scal_gp_int_runtime}. Similarly to GParareal, the quantity $d n_{\rm reg} n_{\rm start}$ can be massively parallelized over the $N$ cores. Moreover, $T_{\rm m-nnGP}$ is almost identical across intervals and iterations. Hence, to estimate $T_{\rm nnGPara}$, it is sufficient to estimate $T_{{\rm GP}_m}$, which is straightforward given the small sample size required. %
Although the advantage of using nnGParareal may not immediately clear from the above formula, $T_{\rm m-nnGP}$ is sensibly cheaper than $T_{\rm GP}$, as
\[
T_{\rm m-nnGP}\in O\left(\left((dn_{\rm reg}  n_{\rm start})  /N \vee 1 \right)  ((d \vee m)m^2 + \log(kN))\right).
\]
 The term $(d \vee m)m^2$ comes from the cost of evaluating and inverting a covariance matrix of size $m$, while $\log(kN)$ provides the cost of determining the $m$ nearest neighbors. Among several approaches developed for the latter, we opt for kd-trees~\cite{friedman1977algorithm} for structured storage of the observations and fast retrieval of the closest points. Searching the $m$ nns has an average complexity of $O(\log N)$. Updating the tree structure for the whole iteration $k$ has a total cost of $O(N \log(kN))$, since the insertion of a tree element is in $O(\log n)$, with the number $n$ of observations currently in the tree, and the cost for finding $1$ or $m$ neighbors is the same (see~\cite{friedman1977algorithm} or~\cite{skrodzki2019kd} for a summary of kd-tree operations cost). Hence, $T_{\rm nnGP}= \sum_{k=1}^{K_{{\rm nnGPara}}} T_{\rm nnGP}(k) $ becomes
\[
T_{\rm nnGP}  \in  O(\left((dn_{\rm reg}  n_{\rm start})  /N \vee 1 \right)K_{{\rm nnGPara}} N ((d \vee m)m^2 + \log(K_{{\rm nnGPara}}N)),
\]
which, for small values of $m$ (e.g., $15-20$), is much lower than that of $T_{\rm GP}$. %

\subsection{Error Analysis}
\label{sec:error_analysis}
In this section, we provide an upper bound for the error between the exact and the nnGParareal solutions at iteration $k$ and interval $i$. A similar bound has been provided for GParareal~\cite[Theorem 3.5]{pentland2023gparareal}, which relies on the posterior mean consistency of GPs. There, the GP error is controlled by the fill distance, defined as 
\[
h_{\mathcal{D}} = \sup_{\boldsymbol{U} \in \mathcal{U}} \inf_{\boldsymbol{U}_j \in \mathcal{D}} \|\boldsymbol{U} - \boldsymbol{U}_j\|.
\]
for a given dataset $\mathcal{D}=\{(\mathcal{U},\mathcal{Y})\}$ and $\mathcal{U} \subset \R^d$. This provides a measure of how well the observations cover the sample space. Given that the GP is trained once and then used to potentially predict at any point $\boldsymbol{U} \in \mathcal{U}$, this notion is meaningful to quantify how the prediction error depends on the available observations. However, in our setting, it is not very helpful. Indeed, since the model is trained to make a prediction based on nns, a better approach commonly used in the literature is to quantify the ``density'' of observations $\boldsymbol{U}_j$ near the test point $\boldsymbol{U}' \in \mathcal{U}$:
\[
d_\delta(\boldsymbol{U}') = \max_{\boldsymbol{U} \in B_\delta(\boldsymbol{U}')} \min_{\boldsymbol{U}_j \in \mathcal{D}_{i,k} } \|\boldsymbol{U} - \boldsymbol{U}_j\|,
\]
for some fixed $\delta >0$, where $B_\delta(\boldsymbol{U}')=\{\boldsymbol{U} \in \mathcal{U}: \|\boldsymbol{U}-\boldsymbol{U}'\| \leq \delta \}$ is a ball of radius $\delta$ centered around the test observation $\boldsymbol{U}'$. Then, the error of the nnGP for the squared exponential kernel \eqref{eq:se_kern} can be bounded using $d_\delta(\boldsymbol{U}')$~\cite{wu1993local}, as shown below. Following a strategy similar to the Parareal error bound~\cite{gander2008nonlinear} yields the desired result. We begin by recalling the necessary assumptions and results.

Set $\Delta t = t_{i+1}-t_i$ and define the error of on the $i$th interval at the $k$th \mbox{iteration as}
\[
E_{i}^{k}:=\|\boldsymbol{U}_{i}-\boldsymbol{U}_{i}^{k}\|,
\]
with $E_0^k=0, k\geq 0$.
\begin{assumption}[Exact fine solver $\f$]
\label{ass:1}
The fine solver $\f$ \mbox{solves \eqref{ODE} exactly:}
\[
\boldsymbol{U}_i=\f(\boldsymbol{U}_{i-1}), \quad i=1,\ldots,N.
\]
\end{assumption}

\begin{assumption}[One-step coarse solver $\g$]
\label{ass:2}
$\g$ is a one-step numerical solver with uniform local truncation error $O(\Delta t^{p+1})$ for $p \geq 1$, such that 
\[
\f(\boldsymbol{U}) - \g(\boldsymbol{U}) =  c^{(p+1)}(\boldsymbol{U}) \Delta t^{p+1} +  c^{(p+2)}(\boldsymbol{U}) \Delta t^{p+2} + \ldots,
\]
where $\boldsymbol{U} \in \R^d$ and the functions $c^{(j)}: \mathbb{R}^d \longrightarrow \mathbb{R}^d$, $j=p+1,p+2,...$ are $\mathcal{C}^\infty$, i.e., infinitely differentiable functions. Hence, it follows that
\begin{equation}
    \|(\f - \g)(\boldsymbol{U})-(\f-\g)(\boldsymbol{U}')\| \leq C_1 \Delta t^{p+1} \|\boldsymbol{U}-\boldsymbol{U}'\|,
     \label{eq:ass_2}
\end{equation}
where $\boldsymbol{U}' \in \R^d$ and $C_1 >0$ is the Lipschitz constant for $c^{(p+1)}(\boldsymbol{U})$.%
\end{assumption}

\begin{assumption}[$\g$ is Lipschitz]
\label{ass:3}
    $\g$ satisfies the Lipschitz condition
    \begin{equation}
        \|\g(\boldsymbol{U})-\g(\boldsymbol{U}')\| \leq (1+C_2 \Delta t) \|\boldsymbol{U}-\boldsymbol{U}'\|,
        \label{eq:ass_3}
    \end{equation}
    for $\boldsymbol{U},\boldsymbol{U}' \in \R^d$ and for some $C_2 >0$.
\end{assumption}

\begin{assumption}[$\boldsymbol{U}_{i}^{k-1} \in \mathcal{D}_{i,k}$]
\label{ass:4}
    $\mathcal{D}_{i,k}$ is constructed such that $\boldsymbol{U}_{i}^{k-1} \in \mathcal{D}_{i,k}$.
\end{assumption}
This last assumption allows us to bound the nnGP error in terms of the distance between the test observation $\boldsymbol{U}_i^k$ and the approximation for the same initial condition obtained during the previous iteration, $\boldsymbol{U}_i^{k-1}$, as they are both in $\mathcal{D}_{i,k}$. As seen in Figure~\ref{fig:brus_dataset_vis_para_both}, this is a reasonable assumption when $\mathcal{D}_{i,k}$ is chosen using the nns, as the approximation obtained at the previous iteration tends to be the closest observation to the test one. 

\begin{definition}[RKHS]
    We define the reproducing kernel Hilbert space (RKHS) $H_\mathcal{K}(\mathcal{U})$ corresponding to the kernel $\mathcal{K}(\cdot, \cdot)$ introduced in Section~\ref{sec:gpara} as a Hilbert space of functions $f: \mathcal{U} \rightarrow \R$ equipped with an inner product $\langle\cdot, \cdot \rangle_{H_{\mathcal{K}}(\mathcal{U})}$.
\end{definition}
\noindent We refer the reader to, e.g., \cite{berlinet2011reproducing} and ~\cite{muandet2017kernel}, for a  comprehensive discussion on RKHS.

\begin{theorem}[Adapted from~\cite{wu1993local}, Theorem 5]
\label{th:gp_constistency}
Let the kernel $\mathcal{K}$ be as in (\ref{eq:se_kern}),
$f=\f-\g$, and, $\widehat{f}=\widehat{f}_{ {\rm nnGPara}}$, to shorten the notation. Assume that $f \in H_{\mathcal{K}}(\mathcal{U})$ and let $\delta>0$ be given. Then, there exists $C>0$ such that
\[
\|\simplef(\boldsymbol{U}')-\widehat{f}(\boldsymbol{U}')\| \leq C d_\delta(\boldsymbol{U}') \|\simplef\|_{H_{\mathcal{K}}(\mathcal{U})},
\]
for all test observations $\boldsymbol{U}' \in \mathcal{U}$, where $\|\simplef\|_{H_{\mathcal{K}}(\mathcal{U})}^2 = \langle\simplef,\simplef\rangle_{H_{\mathcal{K}}(\mathcal{U})}$.
\end{theorem}
\noindent Observe that the assumption $\simplef \in H_{\mathcal{K}}(\mathcal{U})$ implies that the function we wish to approximate using the GP is in the RKHS induced by the kernel. This might not be the case if $\simplef$ is not sufficiently smooth, and a more suitable kernel would be required (see~\cite{kanagawa2018gaussian} for a discussion). Further, note that Theorem \ref{th:gp_constistency} was originally developed for local radial basis function (RBF) interpolation. While there are strong connections between Gaussian process regression (GPR) and RBF interpolation, finding the corresponding covariance function for a given basis function is not always possible. Nonetheless, GPR with the squared exponential kernel is equivalent to RBF interpolation with Gaussian basis, a special case of the result in~\cite{wu1993local}.

\begin{theorem}[nnGParareal error bound]
Suppose $\f$ and $\g$ satisfy Assumptions~\ref{ass:1}-\ref{ass:3}, $\mathcal{D}_{i,k}$ satisfies Assumption~\ref{ass:4}, and that the conditions of Theorem \ref{th:gp_constistency} hold. Then, the error of the nnGParareal solution to an autonomous ODE \eqref{ODE} with \\ $h(\boldsymbol{u}(t),t) = h(\boldsymbol{u}(t))$ satisfies
    \[
       E_i^k \leq \gamma \eta^k \sum_{l=0}^{n-k-1}\binom{l+k}{l} \left( \alpha + \beta +\eta \right) ^{l}, \enspace k+1\leq i\leq N
    \]
    with $E_0^k=0, \forall k\geq 0,$ and
     $\alpha=C_1 \Delta t^{p+1}$, $\beta=1+C_2 \Delta t$, $\gamma=C_3 \Delta t^{p+1}$, and $\eta = 2C \|\simplef\|_{H_{\mathcal{K}}(\mathcal{U})}$.
    
\end{theorem}

\begin{proof}
    Using the predictor-corrector rule for nnGParareal and Assumption~\ref{ass:1}, we write
     \[
    \begin{aligned}
        E_{i}^{k+1} 
        &= \|\boldsymbol{U}_i-\boldsymbol{U}_i^{k+1}\|= \|\boldsymbol{U}_i - \g(\boldsymbol{U}_{i-1}^{k+1}) - \widehat{f}(\boldsymbol{U}_{i-1}^{k+1})\|\\
        &= \|\f(\boldsymbol{U}_{i-1}) - \g(\boldsymbol{U}_{i-1}^{k+1}) - \widehat{f}(\boldsymbol{U}_{i-1}^{k+1})\|\\
        &= \|\f(\boldsymbol{U}_{i-1}) - \g(\boldsymbol{U}_{i-1}^{k+1}) - \widehat{f}(\boldsymbol{U}_{i-1}^{k+1}) \pm \simplef(\boldsymbol{U}_{i-1}^{k+1}) \pm \g(\boldsymbol{U}_{i-1})\|,
    \end{aligned}
    \]
    where $\pm$ denotes adding and subtracting the corresponding term.
    By using the triangle inequality and the definition of $\simplef$, we obtain
    \begin{align}
        E_{i}^{k+1} 
        \leq &\| \g(\boldsymbol{U}_{i-1}) - \g(\boldsymbol{U}_{i-1}^{k+1})\| + \|(\f-\g)(\boldsymbol{U}_{i-1}) - (\f-\g)(\boldsymbol{U}_{i-1}^{k+1})\| \nonumber \\
        &+ \|\simplef(\boldsymbol{U}_{i-1}^{k+1}) -  \widehat{f}(\boldsymbol{U}_{i-1}^{k+1})\|\leq  \beta E_{i-1}^{k+1} + \alpha E_{i-1}^{k+1} + \|\simplef(\boldsymbol{U}_{i-1}^{k+1}) -  \widehat{f}(\boldsymbol{U}_{i-1}^{k+1})\|, \label{error bound}
    \end{align}
    where $\alpha=C_1 \Delta t^{p+1}$, $\beta=1+C_2 \Delta t$. In the last inequality, we used 
   \eqref{eq:ass_3} for the first term and \eqref{eq:ass_2} for the second one. Further, by Theorem~\ref{th:gp_constistency} we write
    \begin{align*}
        \|\simplef(\boldsymbol{U}_{i-1}^{k+1}) -  \widehat{f}(\boldsymbol{U}_{i-1}^{k+1})\| 
        &\leq C \|\simplef\|_{H_{\mathcal{K}}(\mathcal{U})} d_\delta(\boldsymbol{U}_{i-1}^{k+1})\\
        &=  C \|\simplef\|_{H_{\mathcal{K}}(\mathcal{U})} \max_{\boldsymbol{U} \in B_\delta(\boldsymbol{U}_{i-1}^{k+1})} \min_{\boldsymbol{U}_j \in \mathcal{D}_{i-1,k+1} } \|\boldsymbol{U} - \boldsymbol{U}_j\|  \\
        &\leq C \|\simplef\|_{H_{\mathcal{K}}(\mathcal{U})} \max_{\boldsymbol{U} \in B_\delta(\boldsymbol{U}_{i-1}^{k+1})} \|\boldsymbol{U} - \boldsymbol{U}_{i-1}^{k}\|,
        \end{align*}
       where the last step follows from  Assumption~\ref{ass:4}.
        Using again the triangle inequality (adding and subtracting $U_{i-1}^{k+1}$) and the definition of $B_\delta$, we get
         \begin{align*}
         \|\simplef(\boldsymbol{U}_{i-1}^{k+1}) -  \widehat{f}(\boldsymbol{U}_{i-1}^{k+1})\| 
        &\leq C \|\simplef\|_{H_{\mathcal{K}}(\mathcal{U})} \max_{\boldsymbol{U} \in B_\delta(\boldsymbol{U}_{i-1}^{k+1})} \|\boldsymbol{U}_{i-1}^{k+1} - \boldsymbol{U} \| + \|\boldsymbol{U}_{i-1}^{k+1} - \boldsymbol{U}_{i-1}^{k} \|\\
        &\leq C \|\simplef\|_{H_{\mathcal{K}}(\mathcal{U})}  (\delta + \|\boldsymbol{U}_{i-1}^{k+1} - \boldsymbol{U}_{i-1}^{k} \| )\\
        & \leq  2C \|\simplef\|_{H_{\mathcal{K}}(\mathcal{U})}   \|\boldsymbol{U}_{i-1}^{k+1} - \boldsymbol{U}_{i-1}^{k} \|\\
        &\leq \eta(\|\boldsymbol{U}_{i-1}-\boldsymbol{U}_{i-1}^{k}\| + \|\boldsymbol{U}_{i-1}-\boldsymbol{U}_{i-1}^{k+1}\|)= \eta(E_{i-1}^k + E_{i-1}^{k+1}),
        \end{align*}
        where the third inequality follows from choosing $\delta \leq \|\boldsymbol{U}_{i-1}^{k+1} - \boldsymbol{U}_{i-1}^{k} \| $, and the forth one from setting $\eta =2C \|\simplef\|_{H_{\mathcal{K}}(\mathcal{U})}$ and applying another triangle inequality. Finally, using this expression in \eqref{error bound}, we obtain the following recursion
    \begin{equation*}\label{rec}
    E_i^{k+1} \leq \beta E_{i-1}^{k+1} + \alpha E_{i-1}^{k+1} + \eta(E_{i-1}^k + E_{i-1}^{k+1})=(\alpha+\beta+\eta)E_{i-1}^{k+1} +\eta E_{i-1}^{k}.
    \end{equation*}
    Notice that it holds that  
    \begin{align*}
        E_i^0&=\|\boldsymbol{U}_i-\boldsymbol{U}_i^0 + \g(\boldsymbol{U}_{i-1}) - \g(\boldsymbol{U}_{i-1})\|\leq \|\g(\boldsymbol{U}_{i-1})-\g(\boldsymbol{U}_{i-1}^0)\|+\|\boldsymbol{U}_i-\g(\boldsymbol{U}_{i-1})\|\\
      &=\beta E^0_{i-1}+\gamma\leq (\alpha+\beta+\eta)E_{i-1}^0+\gamma, \quad i\in\mathbb{N},
    \end{align*}
    where Assumptions~\ref{ass:1}-\ref{ass:3} are used and $\gamma=C_3 \Delta t^{p+1}$. This recursive inequality can be solved using the result from Lemma B.1 in \cite{pentland2022stochastic}, with their constants $\Lambda=0, A=\eta, B=(\alpha+\beta+\eta), D=\gamma, n=i$ and adjusting for the initial error being $E_i^0$ instead of $E_i^1$.
\end{proof}

\section{Numerical experiments}
\label{sec:num_exp}

In this section, we compare the performance of Parareal, GParareal, and nnGParareal on a set of different systems. We provide a detailed discussion of two low-dimensional ODEs, the non-linear non-autonomous Hopf bifurcation system and the Thomas labyrinth (known to be hard to learn from data~\cite{yang2023learning}), and two PDEs, the viscous Burgers' and the two-dimensional FitzHugh-Nagumo, which we reformulate as high-dimensional ODEs. Additional results for these systems and the other five ODEs (FitzHugh-Nagumo, R\"ossler, Brusselator, double pendulum, and Lorenz) are reported in the Supplementary Material. The algorithms are compared by looking at both $K_{\rm alg}$, the number of iterations required to converge, and $S_{\rm alg}$, the speed-up. It is clear from \eqref{eq:speed} that while minimizing $K_{\rm alg}$ is ideal, as the most expensive component of the algorithms is assumed to be $\f$, maximizing the speed-up is equally important, as the model cost may negatively affect the total wallclock time, slowing down an algorithm otherwise converging in fewer iterations. Although comparing the algorithms based on their prediction error of $\f-\g$ may be an interesting alternative, this would require the serial calculation of the fine solver and analyzing one plot per iteration, which is why we do not pursue it here, other than in Figure \ref{fig:rossler_pred_err_both} before.
\subsection{Comparing the algorithms based on \texorpdfstring{$K_{\rm alg}$}{K alg}}
Since Parareal takes only a few seconds/minutes to solve the seven considered ODE systems with $N \in [32,50]$, one would not need to use GParareal and nnGParareal. Nevertheless, we do that to compare their performances with respect to the required number of iterations to converge. Remarkably, nnGParareal requires the lowest number across the ODE systems even for a small number of nns (see Table \ref{tab:gp_vs_nngp7}, where $m=15$).
Hence, not only replacing the GP with an nnGP within Parareal is not detrimental, but it may even lead to faster convergence than GParareal. Moreover, even when $K_\textrm{GPara}=K_{ {\rm nnGPara}}$ (as for the FHN and the double pendulum), the number of observations used by the GP to make a prediction $\widehat f$ during the last iteration before convergence is much higher than that of the nnGP, as observed in the bottom part of Table~\ref{tab:gp_vs_nngp7}. Hence, in these simple settings, nnGParareal achieves a sensible data reduction while preserving a high accuracy. Remarkably, this is also the case in more challenging settings than those presented in the next sections.   %
\begin{table}[t]
{
\footnotesize
    \centering
    \begin{tabular}{lccccccc}
System& FHN & R\"ossler& Hopf& Bruss. & Lorenz  & Dbl Pend. &Thomas Lab.\\
\hline
$N$& $40$ & $40$& $32$& $32$ & $50$  & $32$ &$32$\\
\hline
 \multicolumn{8}{c}{Number of iterations for each algorithm until convergence with a normalized accuracy $\epsilon=5e^{-7}$} \\ \hline
Parareal  &  11  &  18  &  19  &  19  &  15  &  15 &30\\
 GParareal  &  \textbf{5}  &  13  &  10  &  20  &  11  &  \textbf{10} &25\\
nnGParareal&  \textbf{5}  &  \textbf{12}  &  \textbf{9}  &  \textbf{17}  &  \textbf{9}  &  \textbf{10} & \textbf{24}\\ \hline\\  
 \multicolumn{8}{c}{Number of observations used for training $\widehat{f}$ at iteration $K_{\rm alg}$} \\
 \hline
 Parareal& 1&1&1&1&1&1&1 \\
 GParareal&  187  &  338  &  249  &  274  &  383  &  224 & 472 \\
nnGParareal&  15&15&18&15&15&15&18 \\\hline
    \end{tabular}
    }
    \caption{Top: Number of iterations required by Parareal, GParareal, and nnGParareal to converge for the systems in  Supplementary Material~\ref{app:models} and Sections~\ref{tab:nonaut_scaling} and \ref{sec:exp_tomlab}. Bottom: Number of observations required by each algorithm to train $\widehat f$ at the converging iteration $K_{\rm alg}$. The results of nnGParareal are the same as for  nnGParareal with Time (see Section~\ref{sec:subset_choice} and Supplementary Material~\ref{app:nngp_time}).} %
    \label{tab:gp_vs_nngp7}
\end{table}

\subsection{Non-linear Hopf bifurcation model}
\label{sec:exp_nonaut}
The non-linear Hopf bifurcation model is a two-dimensional non-autonomous ODE~\cite{seydel2009practical}, which is given by
\begin{equation*}
    \frac{du_1}{dt} = -u_2+u_1\left(\frac{t}{T}-u_1^2 - u_2^2\right),\quad 
  \frac{du_2}{dt} = u_1+u_2\left(\frac{t}{T}-u_1^2 - u_2^2\right).
  \label{eq:non_aut}
\end{equation*}
  To obtain an autonomous system, we rewrite it as a three-dimensional model, adding time as an additional coordinate. We integrate it over $t \in [-20,500]$, taking $\boldsymbol{u}_0=(0.1, 0.1, 500)$ as the initial condition and an increasing number of time intervals, namely $N \in\{32,64,128,256, 512\}$, using $m=18$. Runge-Kutta 1 (RK1) with $2048$ steps is chosen as the coarse solver $\g$, while RK8 with $1.74e^9$ steps serves as the fine solver $\f$, making it costly to be run sequentially. %
Ideally, one would use $N$ processors for this task. However, the  HPC facilities available to us do not allow such fine-grained control. Instead, the corresponding number of cores is $47$, $94$, $141$, $282$, or $517$, which is properly considered in the speed-up computations.

The runtime breakdown for Parareal, GParareal, and nnGParareal is shown in Table~\ref{tab:nonaut_scaling}. %
The last three columns report (i) the overall wallclock time of training, hyperparameter selection, and prediction for each corresponding model, (ii) the total wallclock time of running the algorithm, including the costs of $\f$ and $\g$, the model cost, bookkeeping, and other minor operations, and (iii) the improvement in execution speed-up compared to the sequential run of the fine solver, as computed in Subsection~\ref{sec:comp_complx}, respectively.
\begin{table}[t!]
    \centering
   {\small 
\begin{tabular}{lcccccc}
    \multicolumn{7}{c}{Non-linear Hopf bifurcation model for  $N=32$}\\\\
    \hline
    Algorithm & $K_{\rm alg}$ & $T_\g$ & $T_\f$ & $\tmdl$ & $T_{\rm alg}$ & $\hat S_{\rm alg}$\\
    \hline
    Fine & $-$ & $-$ & $-$ & $-$ & 3.51e+04 & 1\\
    Parareal & 19 & 3.70e-02 & 1.09e+03 & 1.96e-03 & 2.08e+04 & 1.69\\
    GParareal & 10 & 6.24e-02 & 1.09e+03 & 6.40e+00 & 1.09e+04 & 3.21\\
    nnGParareal & 9 & 1.53e-02 & 1.09e+03 & 4.71e+00 & 9.80e+03 & \textbf{3.58}\\
    \hline\\
    \multicolumn{7}{c}{$N=64$}\\
    \hline
    Algorithm & $K_{\rm alg}$ & $T_\g$ & $T_\f$ & $\tmdl$ & $T_{\rm alg}$ & $\hat S_{\rm alg}$\\
    \hline
    Fine & $-$ & $-$ & $-$ & $-$ & 3.51e+04 & 1\\
    Parareal & 30 & 3.42e-02 & 5.50e+02 & 6.37e-03 & 1.65e+04 & 2.13\\
    GParareal & 14 & 5.31e-02 & 5.52e+02 & 5.55e+01 & 7.78e+03 & 4.52\\
    nnGParareal & 11 & 3.38e-02 & 5.50e+02 & 1.05e+01 & 6.06e+03 & \textbf{5.80}\\
    \hline\\
    \multicolumn{7}{c}{$N=128$}\\
    \hline
    Algorithm & $K_{\rm alg}$ & $T_\g$ & $T_\f$ & $\tmdl$ & $T_{\rm alg}$ & $\hat S_{\rm alg}$\\
    \hline
    Fine & $-$ & $-$ & $-$ & $-$ & 3.51e+04 & 1\\
    Parareal & 54 & 5.29e-02 & 2.72e+02 & 2.36e-02 & 1.47e+04 & 2.40\\
    GParareal & 16 & 8.49e-02 & 2.73e+02 & 4.31e+02 & 4.79e+03 & 7.33\\
    nnGParareal & 13 & 6.68e-02 & 2.72e+02 & 3.28e+01 & 3.56e+03 & \textbf{9.86}\\
    \hline\\
    \multicolumn{7}{c}{$N=256$}\\
    \hline
    Algorithm & $K_{\rm alg}$ & $T_\g$ & $T_\f$ & $\tmdl$ & $T_{\rm alg}$ & $\hat S_{\rm alg}$\\
    \hline
    Fine & $-$ & $-$ & $-$ & $-$ & 3.51e+04 & 1\\
    Parareal & 97 & 8.88e-02 & 1.96e+02 & 7.63e-02 & 1.90e+04 & 1.85\\
    GParareal & 18 & 1.36e-01 & 1.37e+02 & 3.24e+03 & 5.72e+03 & 6.15\\
    nnGParareal & 16 & 1.19e-01 & 1.36e+02 & 1.08e+02 & 2.28e+03 & \textbf{15.42}\\
    \hline\\
    \multicolumn{7}{c}{$N=512$}\\
    \hline
    Algorithm & $K_{\rm alg}$ & $T_\g$ & $T_\f$ & $\tmdl$ & $T_{\rm alg}$ & $\hat S_{\rm alg}$\\
    \hline
    Fine & $-$ & $-$ & $-$ & $-$ & 3.51e+04 & 1\\
    Parareal & 149 & 1.55e-01 & 6.80e+01 & 1.97e-01 & 1.02e+04 & 3.46\\
    GParareal & 19 & 2.41e-01 & 6.95e+01 & 1.88e+04 & 2.01e+04 & 1.75\\
    nnGParareal & 19 & 2.50e-01 & 7.06e+01 & 2.68e+02 & 1.62e+03 & \textbf{21.75}\\\hline
    \end{tabular}}  %
    \caption{Simulation study on the empirical scalability and speed-up of Parareal, GParareal, and nnGParareal (with $m=18$) for the non-linear non-autonomous Hopf bifurcation system with $d=3$ and $t_N=500$, when $\g$ is RK1, and $\f$ is RK8. $T_\f$ and $T_\g$ refer to the runtimes (in seconds) per iteration of the fine and coarse solvers, respectively. $K_{\rm alg}$ denotes the number of iterations taken to converge. $\tmdl$ includes the running time of training, hyperparameter selection, and prediction for the corresponding model employed. $T_{\rm alg}$ stands for the runtime of the algorithm, while the speed-up $\hat S_{\rm alg}$ reported in the last column is the empirical speed-up of the algorithm, computed with respect to the cost of the sequential fine solver, as described in Subsection~\ref{sec:comp_complx}.}
        \label{tab:nonaut_scaling}
    \end{table}

As expected, all algorithms require more iterations to converge as $N$ doubles, even though the increase of $K_{\rm GPara}$ and $ K_{ {\rm nnGPara}}$ is not substantial. The Parareal speed-up $S_{\rm Para}$ increases in $N$, as its model cost is negligible (requiring only the evaluation of $\fhat_{\rm Para}$, a simple look-up operation) and $T_\f$ halves as a direct effect of parallelization. This is not the case for GParareal, whose model cost is much higher than $T_\f$ for $N$ as large as $256$, slowing down its speed-up at the point of becoming slower than Parareal for $N=512$. On the contrary, nnGParareal converges sooner than the other algorithms for each $N$. Moreover,  it consistently leads to the highest speed-up, which increases in $N$, due to a much smaller model cost than GParareal, thanks to learning $\f-\g$ only on $m$ nns instead of the entire dataset.
The behavior of the empirical and theoretical speed-up of the three algorithms as a function of the number of cores is reported in Figure~\ref{fig:nonaut_scal_speedup}. %
Note the wide gap between the upper bound and theoretical speed-up for GParareal due to the model cost, and how closely nnGParareal matches both the theoretical and the ideal upper bound $K_{ {\rm nnGPara}}/N$ for up to $N=256$. 

Finally, we investigate the impact of the coarse solver accuracy (measured as the number of coarse solvers steps) on the Parareal and nnGParareal speed-ups, for $N=32$ and $1.74e^8$ steps of the fine solver. The results (see Figure \ref{Gextra} in Supplementary Material \ref{extraHopf}) show that  nnGParareal can use a coarser solver than Parareal to obtain a similar speed-up, with the latter showing non-convergent behavior for smaller numbers of coarse solver steps.
 \begin{figure}
     \centering     \includegraphics[width=0.8\linewidth]{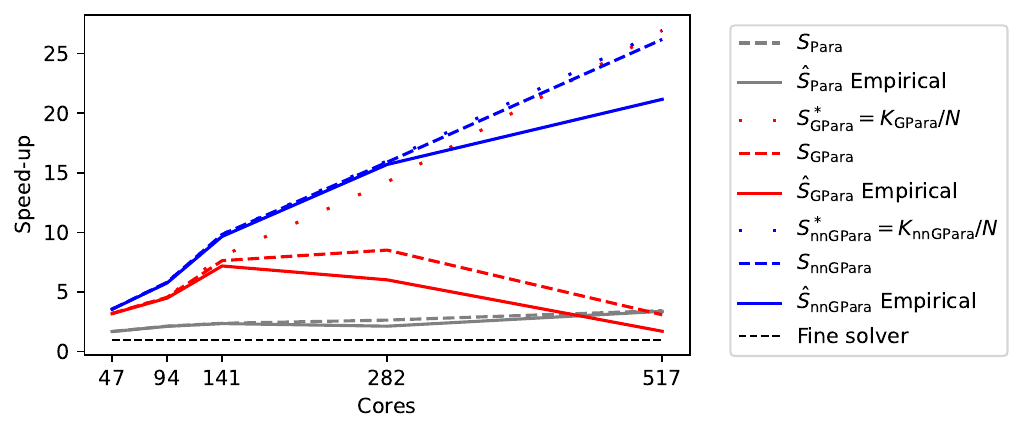}
     \caption{Speed-ups for the Hopf bifurcation system as a function of the number of cores: upper bound $S_{\rm alg}^*$ (dotted lines), theoretical $S_{\rm alg}$ (dashed lines) and empirical $\hat S_{\rm alg}$ (solid lines) values for alg = Para, GPara and nnGPara, standing for Parareal, GParareal, and nnGParareal, respectively.}
     \label{fig:nonaut_scal_speedup}
 \end{figure}
\subsection{Thomas labyrinth}
\label{sec:exp_tomlab}
After showcasing the good empirical performance of nnGParareal for the Hopf bifurcation system, we challenge it on the chaotic Thomas labyrinth system~\cite{thomas1999deterministic}. This is one of the five most difficult systems (out of the $131$ known chaotic systems listed in~\cite{gilpin2021chaos}) to learn from data using sparse kernel flows, when judging based on the symmetric mean absolute percentage error criterion~\cite{yang2023learning}.
This is a three-dimensional system representative of a large class of auto-catalytic models that occur frequently in chemical reactions~\cite{rasmussen1990coreworld}, ecology
\cite{deneubourg1989collective}, and evolution~\cite{kauffman1993origins}. It is described by the following equations
\begin{equation}
        \dfrac{dx}{dt} = b \sin y - ax, \quad
        \dfrac{dy}{dt} = b \sin z - ay,  \quad
        \dfrac{dz}{dt} = b \sin x - az, 
\label{eq:tomlab}
\end{equation}
where $a=0.5$ and $b=10$. We integrate it over $t \in [0, 10]$ for $N=32$ and $N= 64$, $t \in [0, 40]$ for $N=128$, and $t \in [0, 100]$ for $N=256$ and $N= 512$ intervals. Following~\cite{gilpin2021chaos}, we take $\boldsymbol{u}_0=(4.6722764,5.2437205e^{-10},-6.4444208e^{-10})$ as initial condition, such that the system exhibits chaotic dynamics. Further, we use RK1 with $10N$ steps for the coarse solver $\g$ and RK4 with $1e^9$ steps for the fine solver $\f$. We use the normalized version (see Step 1 of nnGParareal) with a normalized accuracy $\epsilon=5e^{-7}$. Also in this case, nnGParareal achieves the highest speed-up for all $N$ other than $N=64$, when $S_{\rm GPara}=3.79$ and $S_{{\rm nnGPara}}=3.65$ (we refer the reader to Table~\ref{tab:tomlab} in Supplementary Material~\ref{app:tomlab_scal_table}). Note also that the speed-ups for different $N$ are not comparable anymore, unless solving for the same $t_N$. %

When designing this case study, we limited the maximum runtime  for each algorithm to $48$ hours, with GParareal failing to
converge for $N=256$ and $N=512$. As before, GParareal suffers from drastic performance losses for high numbers of intervals. Moreover, the situation worsens here due to the high number of iterations required to converge, as the complexity is $O(k^4)$.
Nevertheless, one may still be interested in checking how close to convergence GParareal is. This can be observed in %
Figure~\ref{fig:tom_lab_scal_gp_int_runtime}, where we show both the percentage of converged intervals (red lines) and the training time per iteration (blue lines) for GParareal (top figures) and nnGParareal (bottom figures). For $N=256$ and $N=512$, less than $40\%$ intervals have converged for GParareal, with a training time of $2.5$ and $5$ hours, respectively. On the contrary, nnGParareal converges in $12$ and $7$ minutes, respectively, with training times less than $15$ seconds.

\begin{figure}
\centering\includegraphics[width=0.58\textwidth, , trim={0.2cm 7.5cm 0.2cm 0.0cm},clip]{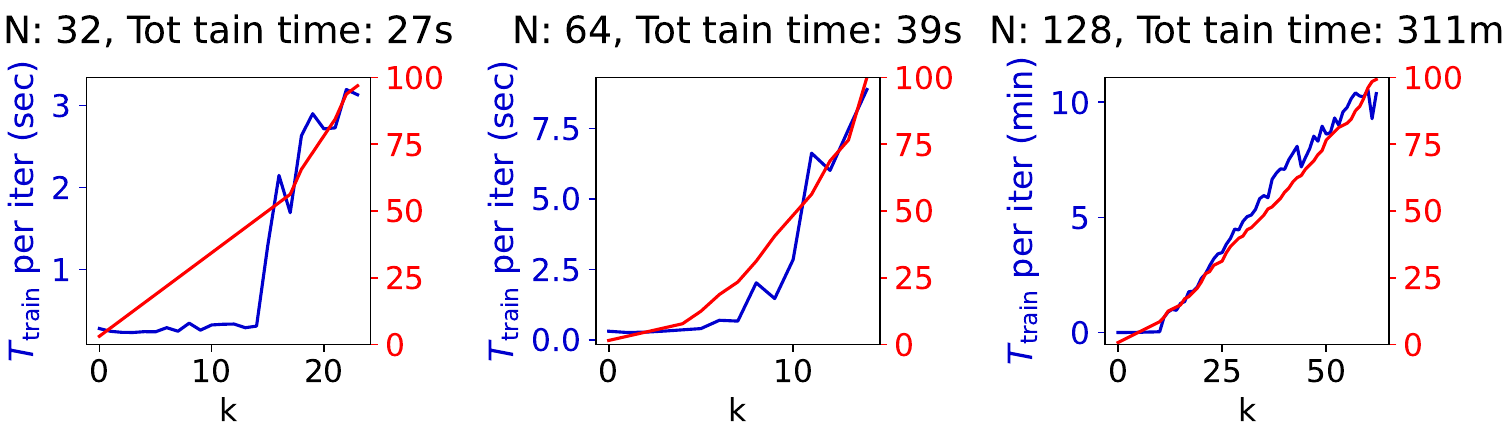}
    \includegraphics[width=0.41\textwidth, , trim={4.3cm 0.2cm 3.9cm 7.7cm},clip]{figures/tom_lab_scal_gp_int_runtime_upd.pdf}
    \includegraphics[width=0.58\textwidth, , trim={0.2cm 7.5cm 0.2cm 0.0cm},clip]{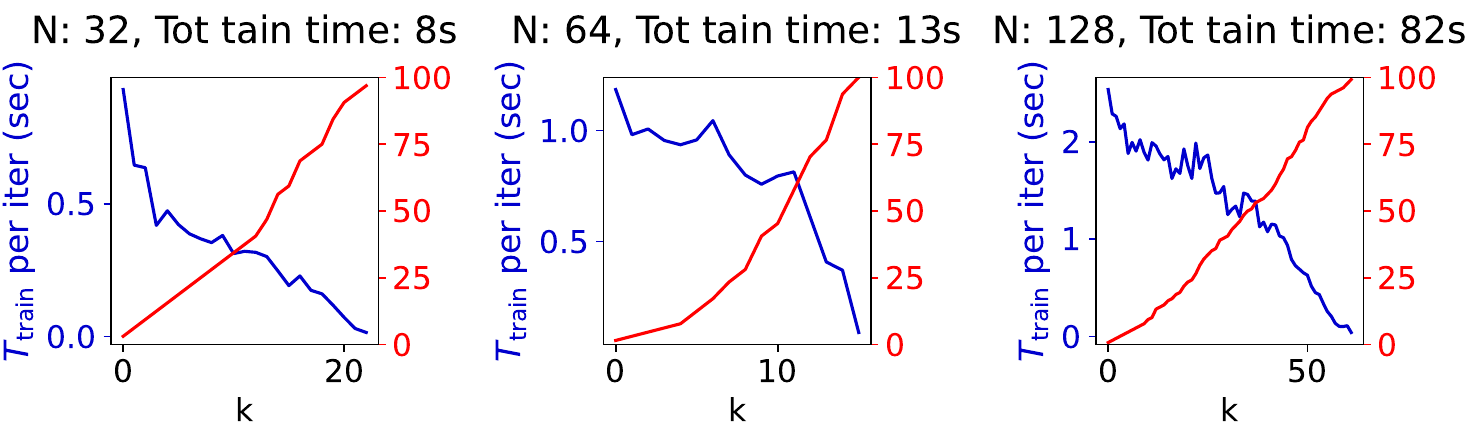}
    \includegraphics[width=0.41\textwidth, , trim={4.3cm 0.2cm 3.9cm 7.7cm},clip]{figures/tom_lab_scal_nngp_int_runtime_upd.pdf}
    \caption{GParareal's (top) and nnGParareal's (bottom) percentages of converged intervals (red) and training time per iteration (blue) for Thomas labyrinth~\eqref{eq:tomlab}. The latter is increasing for GParareal and decreasing for nnGParareal. The aggregated model cost across $k$ is shown in the title. GParareal failed to converge for $N=256$ and $N=512$ within the computational time budget of $48$ hours.}
    \label{fig:tom_lab_scal_gp_int_runtime}
\end{figure}

\subsection{Viscous Burgers' equation}
\label{sec:exp_burg}
We now focus on non-linear PDEs, starting with the viscous Burgers' equation with Dirichlet boundary conditions. This convection-diffusion model is a one-dimensional system defined as
\begin{equation}
    v_t = \nu v_{xx} - vv_x, \qquad (x,t)\in [-L,L] \times(t_0,t_N],
    \label{eq:burg}
\end{equation}
with initial condition $v(x, t_0)=v_0(x)$, $x\in [-L,L]$, and boundary conditions \\$v(-L,t)=v(L,t)$, $v_x(-L,t)=v_x(L,t)$, $t \in [t_0, t_N]$. Here, $\nu$ is the diffusion coefficient. We discretize the spatial domain using finite difference~\cite{fornberg1988generation} and $d+1$ equally spaced points $x_{j+1}=x_j+\Delta x$, where $\Delta x = 2L/d$ and $j=0,\ldots,d$. This results in a $d$-dimensional ODE system.  In the numerical experiments, we choose $N=d=128$, $L=1$, $\nu=1/100$, $v_0(x)=0.5(\cos (\frac{9}{2}\pi x)+1)$, $m=18, t_0=0$, and consider two values for the time horizon, namely $t_N=5$ and $t_N=5.9$. %
We use RK1 with $4N$ steps for the coarse solver $\g$ and RK8 with $5.12e^6$ steps for the fine solver $\f$, and the normalized accuracy value  $\epsilon=5e^{-7}$. 
By extending the time horizon without increasing the course solver's discretization points, we make the problem more challenging, as the information coming from $\g$ is less accurate. We choose $t_N=5.9$, as larger values of $t_N$ make all algorithms fail due to $\g$ being too poor. This is a more challenging setting than for the Thomas labyrinth, where the number of coarse solver steps is adjusted proportionally to $N$ for larger $t_N$ and $N$. Table~\ref{tab:burges} shows the simulation results. 

\begin{table}[t]
    \centering
    {\small
\begin{tabular}{lcccccc}
\multicolumn{7}{c}{Viscous Burgers' PDE with $t_N=5$}\\\\
\hline
Algorithm & $K_{\rm alg}$ & $T_\g$ & $T_\f$ & $\tmdl$ & $T_{\rm alg}$ & $\hat S_{\rm alg}$\\
\hline
Fine & $-$ & $-$ & $-$ & $-$ & 4.75e+04 & 1\\
Parareal & 10 & 1.09e-01 & 3.82e+02 & 7.59e-03 & 3.82e+03 & 12.43\\
GParareal & 6 & 1.44e-01 & 3.77e+02 & 3.72e+03 & 5.98e+03 & 7.94\\
nnGParareal & 9 & 1.41e-01 & 3.73e+02 & 4.32e+02 & 3.79e+03 & \textbf{12.54}\\
\hline\\
\multicolumn{7}{c}{$t_N=5.9$}\\
\hline
Algorithm & $K_{\rm alg}$ & $T_\g$ & $T_\f$ & $\tmdl$ & $T_{\rm alg}$ & $\hat S_{\rm alg}$\\
\hline
Fine & $-$ & $-$ & $-$ & $-$ & 4.76e+04 & 1\\
Parareal & 90 & 4.62e-02 & 3.74e+02 & 3.66e-02 & 3.37e+04 & 1.41\\
GParareal & 8 & 1.30e-01 & 3.75e+02 & 9.38e+03 & 1.24e+04 & 3.84\\
nnGParareal & 14 & 1.20e-01 & 3.74e+02 & 6.40e+02 & 5.88e+03 & \textbf{8.09}\\
\hline
\end{tabular}}
\caption{Performance of Parareal, GParareal and nnGParareal (with $m=18$) on viscous Burgers' PDE \eqref{eq:burg}, rewritten and solved as a $d$-dimensional ODE, with $N=d=128$, with RK1 for $\g$ and  RK8 for $\f$. Time is shown in seconds. $T_\f$ and $T_\g$ refer to the runtime per iteration of the fine and coarse solvers, respectively. $K_{\rm alg}$ denotes the number of iterations to converge. $\tmdl$ includes the running time of training, hyperparameter selection, and prediction for the corresponding model employed across all iterations. $T_{\rm alg}$ and $\hat S_{\rm alg}$ are the runtime and the empirical parallel speed-up of the algorithms, respectively.}%
\label{tab:burges}
\end{table}
 Unlike for ODEs, nnGParareal requires more iterations than GParareal to converge when considering PDEs. This is also the case for the two-dimensional FitzHugh-Nagumo PDE in~Section~\ref{sec:exp_fhn}, suggesting a loss of information when considering the ($m$) nnGP instead of the GP trained on the entire dataset. Higher values of $m$  helps in reducing $K_{{\rm nnGPara}}$, lowering it down to 7 for $t_N=5$ and 12 for $t_N=5.9$ for $m$ as large as 30 (see Figure \ref{fig:Burges_perf_across_m_K} in Supplementary Material~\ref{app:burg_perf_across_m}). At the same time though, the model cost increases too, so the highest speed-up may not be necessarily achieved by a larger $m$ and a lower $K_{{\rm nnGPara}}$ (see Figure \ref{fig:Burges_perf_across_m_speed} in Supplementary Material~\ref{app:burg_perf_across_m}). In general, even if $K_{{\rm nnGPara}}>K_{\rm GPara}$ for all considered $m$ values, the nnGParareal speed-up is always higher than that of the other algorithms, particularly when considering $t_N=5.9$, since (a) nnGParareal converges faster than Parareal for $t_N=5$ and much faster for $t_N=5.9$; (b) the high model cost of GParareal negatively offsets its gains in terms of $K_{\rm alg}$. As previously observed, Gparareal is an effective alternative to Parareal for moderate values of $N$, while the proposed nnGParareal makes it competitive and advantageous for large $N$ as well. 
 Remarkably, both GParareal and nnGParareal are less sensible to the sub-optimal coarse solver chosen for $t_N=5.9$ than Parareal, with $K_{\rm GPara}$ and $K_{{\rm nnGPara}}$ rising slightly compared to $K_{\rm Para}$. This cushioning effect is particularly desired when in doubt about the accuracy of the selected $\g$. 

We conclude this section by investigating whether the performance observed for Parareal and nnGParareal may result from asymmetries in the finer temporal/spatial resolutions, as an unbalanced discretization may \lq\lq artificially\rq\rq enhance speed-up gains. To assess whether this is the case, we perform two additional simulations. First, we increase the spatial discretization resolution to $d = 400, 815, 1128$, while simultaneously reducing the temporal resolution of the fine solver $\mathcal{F}$ to $10^4, 10^5,$ and $6\times10 ^5$ steps, respectively. In all tested scenarios, nnGParareal consistently achieves greater speed-ups than Parareal (results not shown), mirroring the findings summarized in Table \ref{tab:burges}. Although this adjustment partially reduces the  differences in orders of magnitude of the temporal and spatial discretization errors, it still does not quantify whether the two errors are effectively balanced.
Direct quantification and balancing of spatial and temporal discretization errors (error equilibration) are inherently challenging tasks in complex high-dimensional PDE problems, like those considered here. Without access to analytical solutions, error estimation typically demands highly accurate numerical approximations in both time and space, making it computationally infeasible.
To overcome this limitation, we conduct an additional numerical experiment using the heat equation, a PDE with an explicitly known analytical solution,  enabling the derivation of balanced discretization errors. For this experiment, we choose the initial condition $u_0(x)=\sin(2\pi x)$, and evaluate the performance of both Parareal and nnGParareal in a setting specifically designed to ensure balanced discretization errors, benchmarking their results against a sequential exact solver as baseline. This experiment, detailed in Supplementary Material \ref{app:heat}, demonstrates parallel speed-ups of $8.85$ for Parareal and $23.50$ for nnGParareal.

\subsection{FitzHugh-Nagumo PDE}
\label{sec:exp_fhn}
In this section, we solve the two-dimensional non-linear FitzHugh-Nagumo PDE model~\cite{ambrosio2009propagation}, an extension of the corresponding ODE system described in Supplementary Material~\ref{sys:fhn}. This PDE system models a set of cells constituted by a small nucleus of pacemakers near the origin, immersed among an assembly of excitable cells. %
We discretize both spatial dimensions using finite difference and $\Tilde{d}$ equally spaced points, yielding an ODE with $d=2\Tilde{d}^2$ dimensions. We then test the scalability properties of the algorithms, increasing the dimension $d$. In particular, we consider $\Tilde{d}=10,12,14,16$, and thus $d=200, 288, 392, 512$ for $N=512$, $m=20$, RK8 with $10^8$ steps for $\f$ and different $t_N$, coarse solvers and number of $\g$ steps. %
We refer the reader to Supplementary Material~\ref{sys:fhn_pde} for the model equations and the simulation setup.

As for the viscous Burgers' PDE, if it terminates within the allocated 48 hours time budget, GParareal requires the smallest number of iterations to converge, followed by nnGParareal (see Figure~\ref{fig:fhn_pde_speeedup}
, left panel). Nevertheless, except for $d=200$, when $\hat S_{\rm Para}>\hat S_{{\rm nnGPara}}$, nnGParareal outperforms the other two algorithms in terms of speed-up, as shown in Figure \ref{fig:fhn_pde_speeedup} 
(a detailed breakdown of all running times and training costs is reported in Table~\ref{tab:fhn_pde} in Supplementary Material~\ref{app:FHN_PDE_scal}). The gap between theoretical and empirical speed-ups for nnGParareal can be closed using a more efficient parallel implementation (i.e. one where the overhead from parallelization is lower), while that, albeit small, between upper bound and theoretical speed-up by decreasing the model training cost (i.e. by using a cheaper model).

 \begin{figure}[t]
     \centering
\begin{minipage}{.35\textwidth}
   {\scriptsize 
\begin{tabular}{lcc}

Algorithm & d=200 & d=288\\
\hline
Parareal & 25 & 67 \\
GParareal & 8 & 7  \\
nnGParareal & 12 & 10  \\\\

Algorithm &  d=392 & d=512 \\
\hline
Parareal  & 44 & 79 \\
GParareal & $>11$ & $>9$ \\
nnGParareal  & 5 & 6 \\
\end{tabular}}
\end{minipage}
     \begin{minipage}{.64\textwidth}
\includegraphics[width=\linewidth]{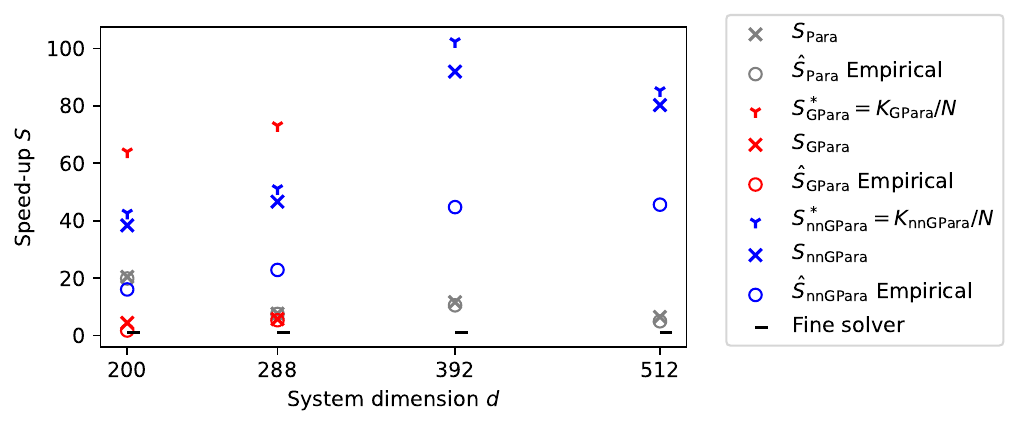}
     \end{minipage}
  \caption{FitzHugh-Nagumo PDE model. Left: Number of iterations $K$ needed to converge for Parareal, GParareal, and nnGParareal as a function of the dimension $d$. Right: observed, ideal, and theoretical  speed-ups of the three algorithms, with the latter computed according to the formulas in Subsection~\ref{sec:comp_complx}. For $N=256$ and $N= 512$, GParareal failed to converge within the computational time budget.}
\label{fig:fhn_pde_speeedup}
 \end{figure}


\section{Discussion}

\label{sec:discussion}

In this paper, we introduce nnGParareal, a new PinT scheme for the solution of high-dimensional and challenging systems of ODEs (or PDEs). The algorithm is built upon GParareal, a PinT method combining Parareal and Gaussian processes, where the latter are used as an emulator to learn the prediction accuracy $\f-\g$ between the fine $\f$ and coarse $\g$ solvers. 
By using the entire dataset instead of data coming from the last iterations, GParareal was shown to converge in fewer iterations than Parareal (even when Parareal failed due to numerical instability of the coarse solver) and to be able to use legacy solutions (e.g. solutions from prior runs of the IVP for different initial conditions) to further accelerate its convergence. However, GParareal suffers from the curse of dimensionality in both the number of data points $kN$ (where $k$ is the iteration and $N$ the number of time intervals) and the model dimension $d$.

The idea of nnGParareal is to replace the GP trained on the entire dataset with one trained on $m$ data points. Among several approaches, we numerically illustrate how choosing $m$ as the nearest neighbors leads to the highest  prediction accuracy and the smaller number of iterations needed to converge. By embedding the derived nnGP within Parareal, we obtain an algorithm which mostly preserves or even improves the accuracy of GParareal (measured in terms of number of iterations needed to converge), thanks to the partial relaxation of the second-order stationarity assumption made by the GPs \cite{datta2016hierarchical}, while drastically reducing its computational cost and wallclock time, thus improving the speed-up with respect to running $\f$ sequentially. Indeed, the time complexity of nnGParareal is now only $O(dm^3N+dN\log (kN))$ at iteration $k$ instead of $O(d(kN)^3)$ of GParareal. This allows nnGParareal to tackle GParareal's scalability issues, consistently yielding the highest speed-up compared to GParareal and Parareal across several different systems of ODEs and PDEs. This makes our algorithm better suited to the modern requirements and capabilities of HPC, thus expanding the temporal scale over which we can solve ODEs and PDEs with time parallelization. Nevertheless, the performance of nnGParareal might be further improved by leaving the Gaussian processes framework and selecting a learned model capable of training on all $ d$ coordinates simultaneously, with a cheaper hyperparameter optimization routine, as recently proposed in \cite{GGT2024} via random weights neural networks.\\
\indent Finally, it is worth recalling that both GParareal and nnGParareal use only the underlying GP and nnGP posterior means for prediction, ignoring the underlying uncertainty in the GP and nnGP posteriors. Propagating the full uncertainty will yield a probabilistic version of GParareal and nnGParareal, respectively, as discussed in \cite{pentland2023gparareal}.  We leave this for future work.

\section*{Acknowledgements}
\label{sec:acknowledgements}
GG is funded by the Warwick Centre of Doctoral Training in Mathematics and Statistics. GG thanks the hospitality of the University of St. Gallen; it is during a visit to this institution that part of the results in this paper were obtained. Calculations were performed using the Warwick University HPC facilities on Dell PowerEdge C6420 compute nodes each with 2 x Intel Xeon Platinum 826 (Cascade Lake) 2.9 GHz 24-core processors, with 48 cores per node and 192 GB DDR4-2933 RAM per node. Python code accompanying this manuscript is made available at \href{https://github.com/Parallel-in-Time-Differential-Equations/Nearest-Neighbors-GParareal}{https://github.com/Parallel-in-Time-Differential-Equations/Nearest-Neighbors-GParareal}.

\bibliographystyle{siamplain}
\bibliography{references}


\appendix
%

\section{Models}
\label{app:models}

\subsection{FitzHugh–Nagumo}
\label{sys:fhn}
The deterministic FitzHugh-Nagumo (FHN) is a model for an animal nerve axon \cite{nagumo1962active}. It is a commonly used example of a relatively simple system to learn, does not exhibit chaotic behavior, and is described by the following equations
\[
\frac{{d} u_1}{{d} t}=c\left(u_1-\frac{u_1^3}{3}+u_2\right), \quad \frac{{d} u_2}{{d} t}=-\frac{1}{c}\left(u_1-a+b u_2\right),
\]
with $a = 0.2$, $b=0.2$, and $c=3$. Here, $u_1$ is the axon membrane voltage, with $u_2$ representing a linear recovery variable. We integrate over $t \in [0,40]$ using $N=40$ intervals, taking $\boldsymbol{u}_0=(-1,1)$ as initial condition. We use Runge-Kutta~2  (RK2) with 160 steps for the coarse solver $\g$, and RK4 with $1.6e^5$ steps for the fine solver $\f$. This is the same setting as in \cite{pentland2023gparareal}, which allows almost direct comparison, although we use the normalized version of the system with $\epsilon=5e^{-7}$. 

Figure~\ref{fig:fhn-conv-prec} shows the performance of different Parareal variants. The left plot displays the number of intervals converged per iteration $k$ and the total number of Parareal iterations. %
The right plot shows the pointwise solution error of each algorithm with respect to the solution obtained by the fine solver. Upon convergence, Parareal guarantees that the error at each initial condition is below $\epsilon$. However, there is no such constraint for the points in between, as the Parareal scheme does not address these. The error naturally increases for chaotic systems due to the diverging nature of infinitesimally different trajectories, as seen in Figure~\ref{fig:rossler_res}.
\begin{figure}[ht]
     \centering
     \begin{subfigure}[b]{0.5\textwidth}
        \centering
        \includegraphics[width=1\textwidth]{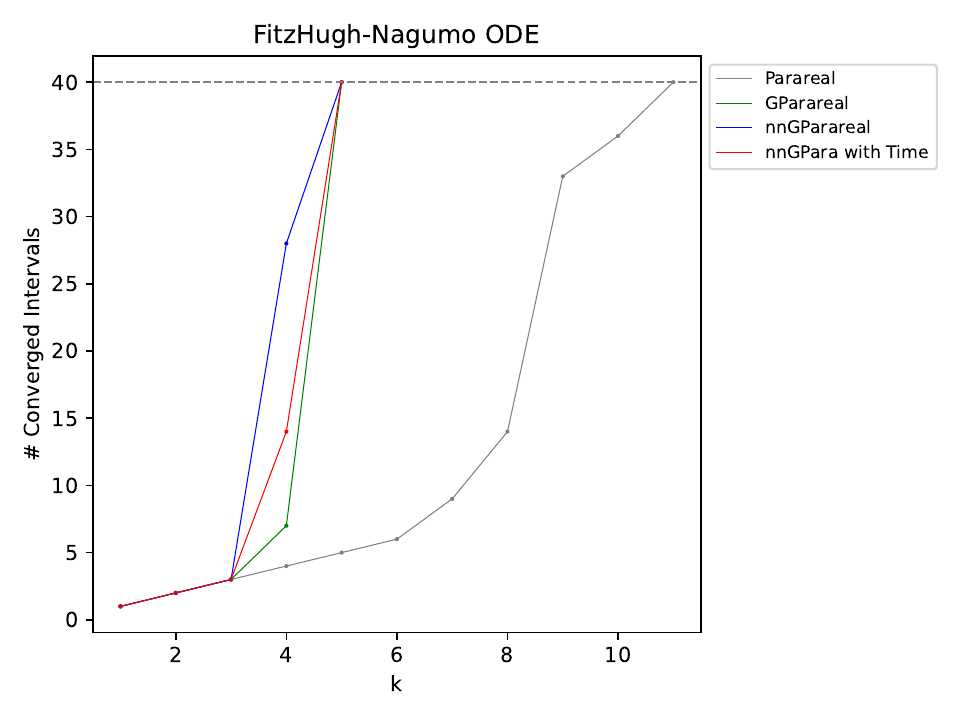}
     \end{subfigure}%
     \hfill
     \begin{subfigure}[b]{0.5\textwidth}
        \centering
        \includegraphics[width=1\textwidth]{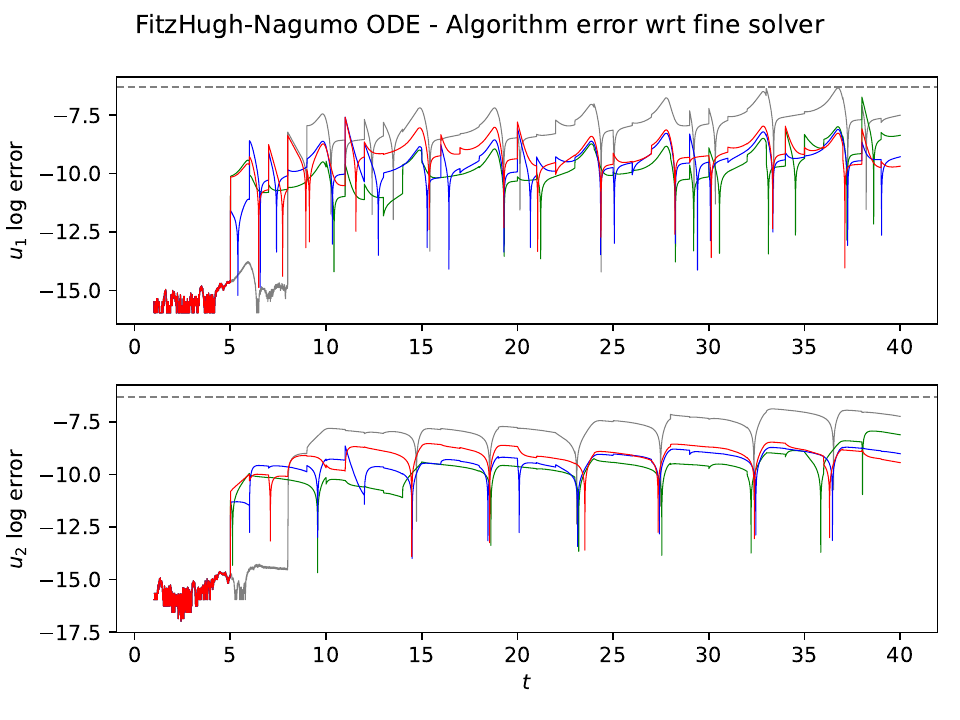}
     \end{subfigure}%
     \caption{FitzHugh-Nagumo model. Left: number of converged intervals per iteration $k$. Right: accuracy of different Parareal algorithms compared to the fine solver.}
     \label{fig:fhn-conv-prec}
\end{figure}

\subsection{R\"ossler}
\label{sys:rossler}
The R\"ossler system is a model for turbulence \cite{rossler1976equation}, described by the following equations
\[
\frac{{d} u_1}{{d} t}=-u_2-u_3, \quad \frac{{d} u_2}{{d} t}=u_1+a u_2, \quad \frac{{d} u_3}{{d} t}=b+u_3\left(u_1-c\right),
\]
which, for $a=0.2$, $b=0.2$, and $c=5.7$,  exhibits chaotic behavior. This configuration is commonly used throughout the literature. We integrate over $t \in [0,340]$ using $N=40$ intervals, taking $\boldsymbol{u}_0=(0,-6.78,0.02)$ as initial condition. We use RK1 with $9e^4$ steps for the coarse solver $\g$, and RK4 with $4.5e^8$ steps for the fine solver $\f$. This is the same setting as \cite{pentland2023gparareal}, and in our empirical illustrations, we use the normalized version of this system with $\epsilon=5e^{-7}$.
\begin{figure}[ht!]
     \centering
     \begin{subfigure}[b]{0.5\textwidth}
        \centering
        \includegraphics[width=1\textwidth]{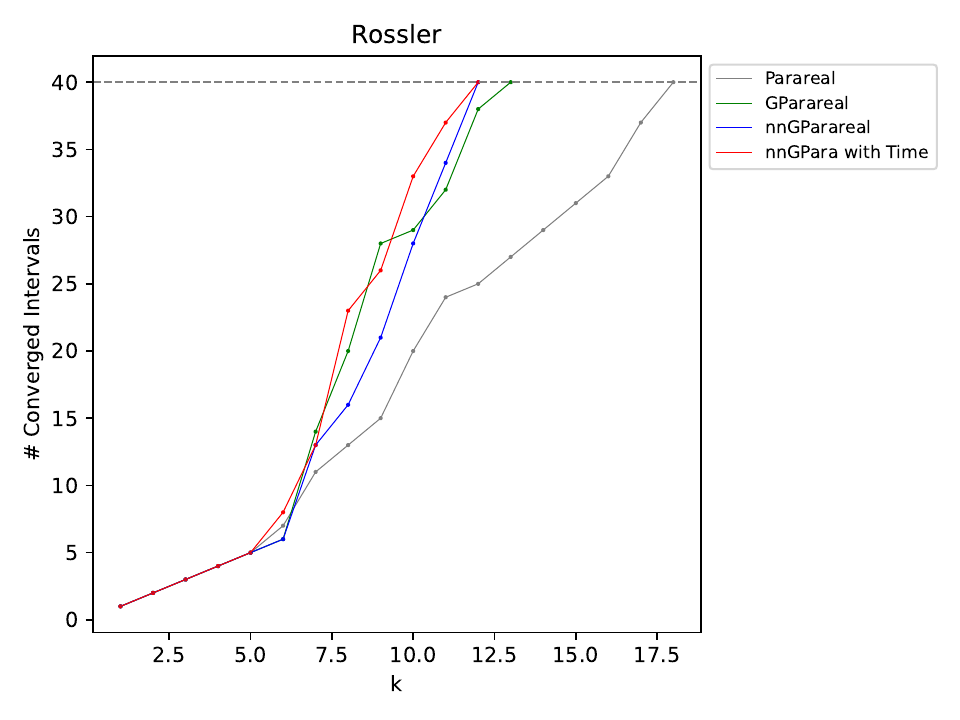}
     \end{subfigure}%
     \hfill
     \begin{subfigure}[b]{0.5\textwidth}
        \centering
        \includegraphics[width=1\textwidth]{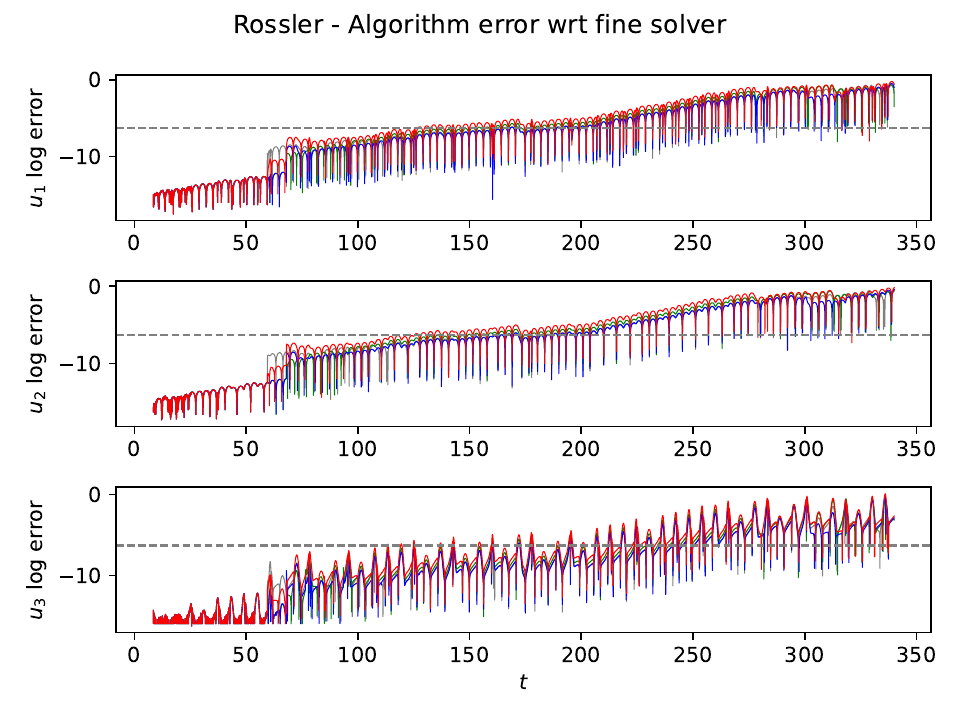}
     \end{subfigure}%
     \caption{R\"ossler system. Left: number of converged intervals per iteration $k$. Right: accuracy of different Parareal algorithms models compared to that of the fine solver.}
     \label{fig:rossler_res}
\end{figure}

\begin{figure}[ht]
     \centering
     \begin{subfigure}[b]{0.5\textwidth}
        \centering
        \includegraphics[width=1\textwidth]{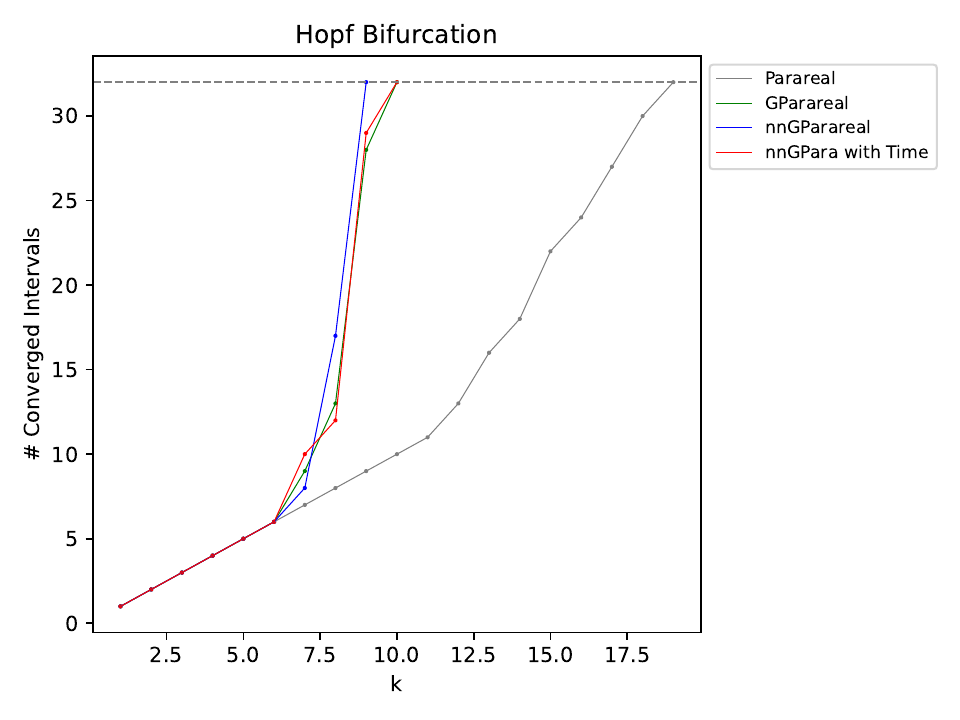}
     \end{subfigure}%
     \hfill
     \begin{subfigure}[b]{0.5\textwidth}
        \centering
        \includegraphics[width=1\textwidth]{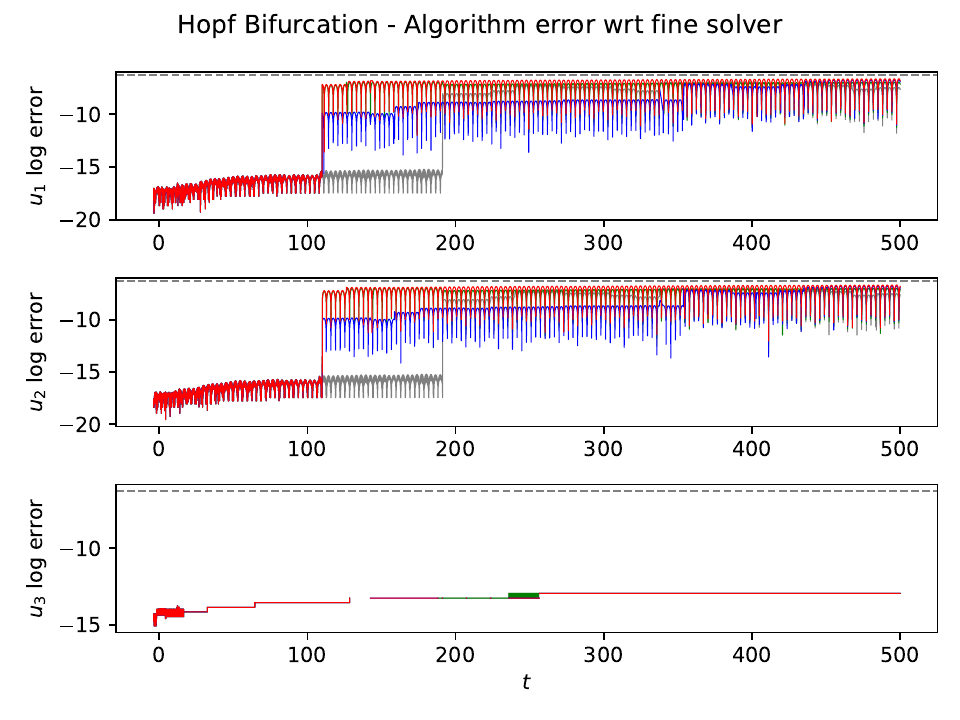}
     \end{subfigure}%
     \caption{Non-linear Hopf bifurcation model. Left: number of converged intervals per iteration $k$. Right: accuracy of different Parareal models compared to the fine solver.}
\end{figure}

\subsection{Brusselator}
\label{sys:brus}
The Brusselator models an autocatalytic chemical reaction \cite{lefever1971chemical}. It is a stiff, non-linear ODE, defined as
\begin{equation*}
\frac{d u_1}{d t}=a+u_1^2 u_2-(b+1) u_1, \quad 
\frac{d u_2}{d t}=b u_1-u_1^2 u_2,
\end{equation*}
where we set $a=1$ and $b=3$. We integrate over $t \in [0, 100]$ using $N=32$ intervals, taking $\boldsymbol{u}_0=(1, 3.7)$ as initial condition. We use RK4 with $2.5e^2$ steps for the coarse solver $\g$, and RK4 with $2.5e^4$ steps for the fine solver $\f$. We use the normalized version and set a normalized $\epsilon=5e^{-7}$.
\begin{figure}[ht]
     \centering
     \begin{subfigure}[b]{0.5\textwidth}
        \centering
        \includegraphics[width=1\textwidth]{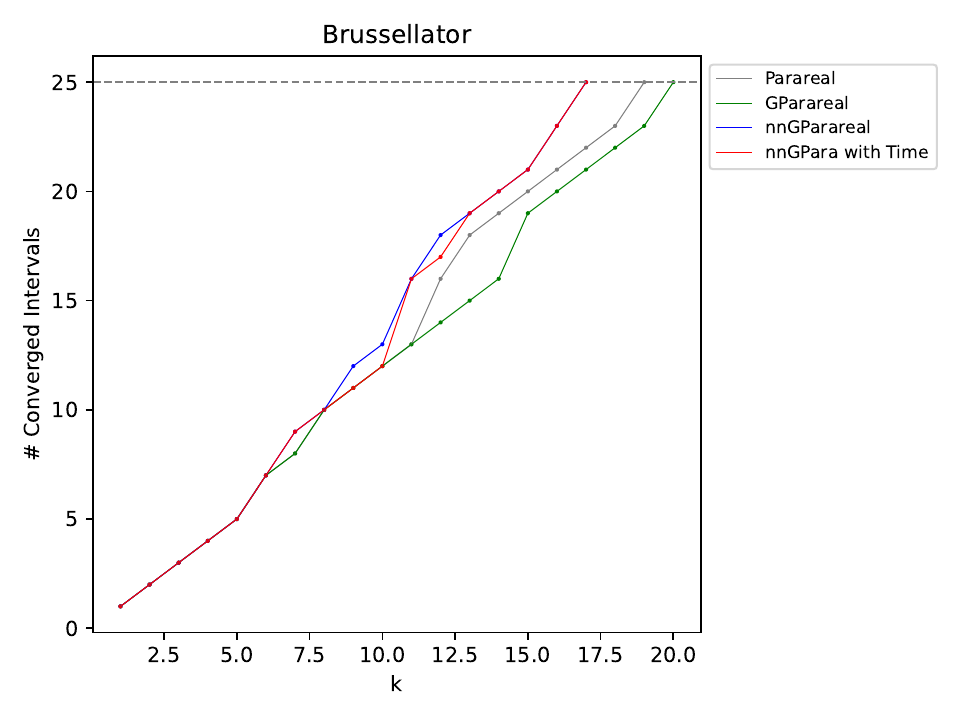}
     \end{subfigure}%
     \hfill
     \begin{subfigure}[b]{0.5\textwidth}
        \centering
        \includegraphics[width=1\textwidth]{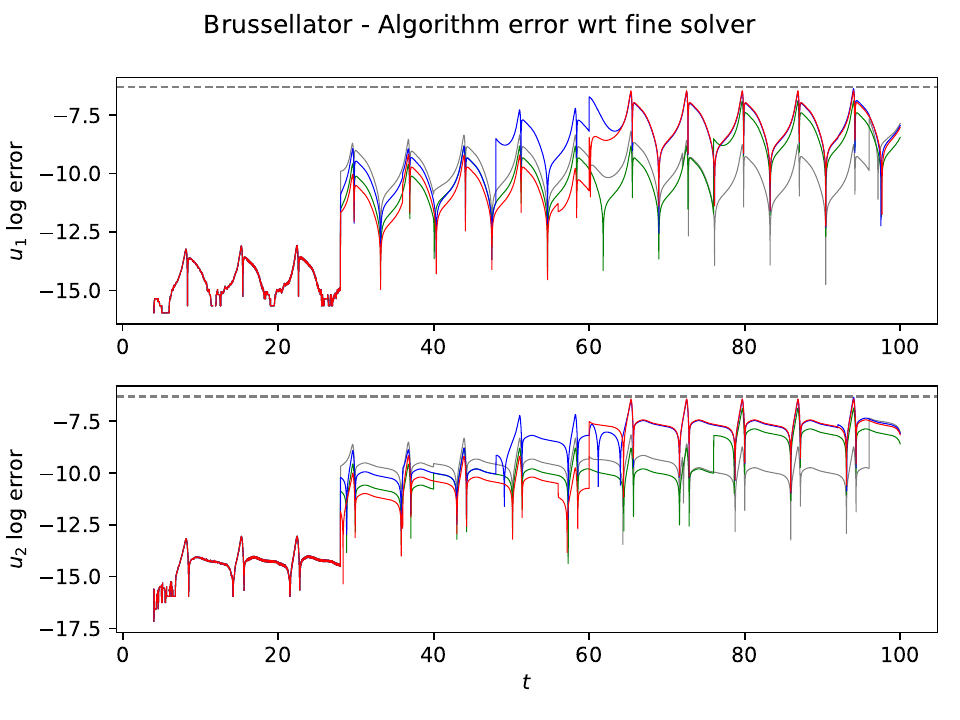}
     \end{subfigure}%
     \caption{Brusselator system. Left: number of converged intervals per iteration $k$. Right: accuracy of different Parareal models compared to the fine solver.}
\end{figure}

\subsection{Double pendulum}
\label{sys:pend}
This is a model for a double pendulum, adapted from \cite{danby1997computer}. It consists of a simple pendulum of mass $m$ and rod length $\ell$ connected to another simple pendulum of equal mass $m$ and rod length $\ell$, acting under gravity $g$. The model is defined by the following equations
\[
\begin{aligned}
\frac{{d} u_1}{{~d} t} & =u_3, \\
\frac{{d} u_2}{{~d} t} & =u_4, \\
\frac{{d} u_3}{{~d} t} & =\frac{-u_3^2 f_1\left(u_1, u_2\right)-u_4^2 \sin \left(u_1-u_2\right)-2 \sin \left(u_1\right)+\cos \left(u_1-u_2\right) \sin \left(u_2\right)}{f_2\left(u_1, u_2\right)}, \\
\frac{{d} u_4}{{~d} t} & =\frac{2 u_3^2 \sin \left(u_1-u_2\right)+u_4^2 f_1\left(u_1, u_2\right)+2 \cos \left(u_1-u_2\right) \sin \left(u_1\right)-2 \sin \left(u_2\right)}{f_2\left(u_1, u_2\right)},
\end{aligned}
\]
with $
f_1\left(u_1, u_2\right)=\sin \left(u_1-u_2\right) \cos \left(u_1-u_2\right)$ and $
f_2\left(u_1, u_2\right)=2-\cos ^2\left(u_1-u_2\right)
$ and 
where $m, \ell$, and $g$ have been scaled out of the system by letting $\ell=g$. The variables $u_1$ and $u_2$ measure the angles between each pendulum and the vertical axis, while $u_3$ and $u_4$ measure the corresponding angular velocities. The system exhibits chaotic behavior and is commonly used in the literature. Based on the initial condition, it can be difficult to learn. We integrate over $t \in [0, 80]$ using $N=32$ intervals, taking $\boldsymbol{u}_0=(-0.5,0,0,0)$ as initial condition. We use RK1 with $3104$ steps for the coarse solver $\g$, and RK8 with $2.17e^5$ steps for the fine solver $\f$. This is a similar setting as \cite[Figure 4.10]{pentland2023gparareal}, although, like above, we use the normalized version and set a normalized $\epsilon=5e^{-7}$.
\begin{figure}[ht]
     \centering
     \begin{subfigure}[b]{0.5\textwidth}
        \centering
        \includegraphics[width=1\textwidth]{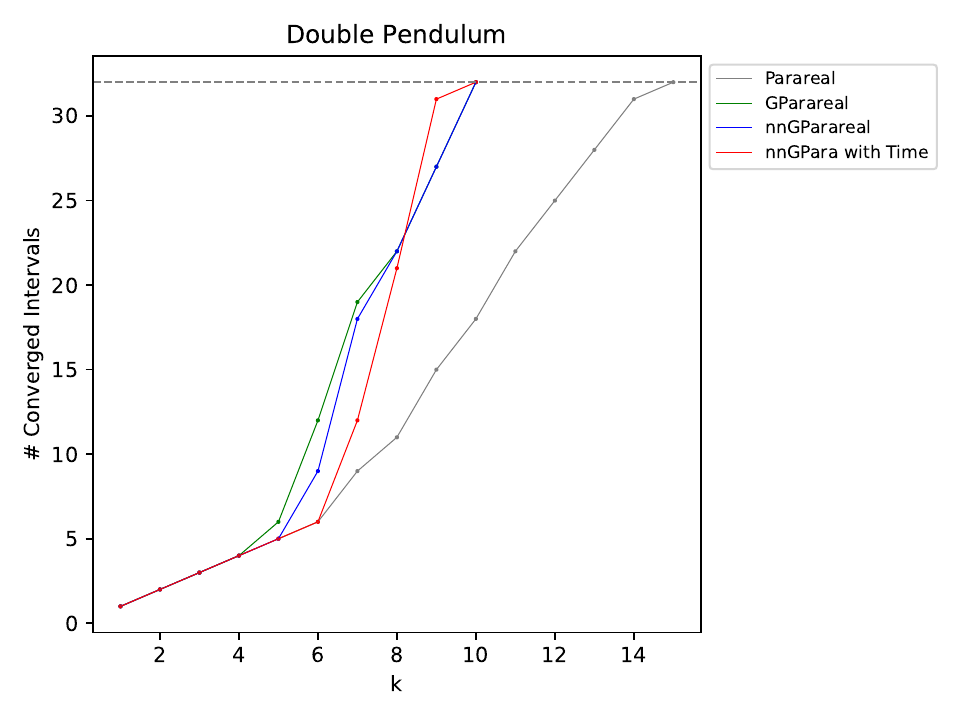}
     \end{subfigure}%
     \hfill
     \begin{subfigure}[b]{0.5\textwidth}
        \centering
        \includegraphics[width=1\textwidth]{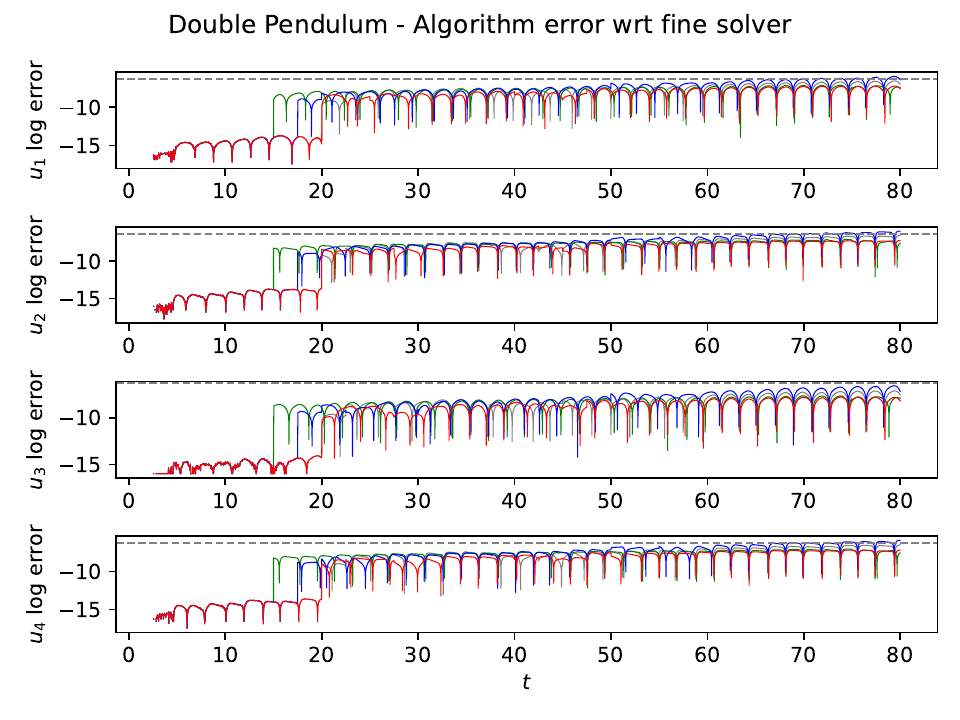}
     \end{subfigure}%
     \caption{Double pendulum. Left: number of converged intervals per iteration $k$. Right: accuracy of different Parareal models compared to the fine solver.}
\end{figure}

\subsection{Lorenz}
\label{sys:lorenz}
The Lorenz system is a simplified model used in different applications, e.g. for weather prediction \cite{lorenz1963deterministic}. It is given by
\[
\begin{aligned}
& \frac{d u_1}{d t}=\gamma_1\left(u_2-u_1\right), \quad \frac{d u_2}{d t}=\gamma_2 u_1-u_1 u_3-u_2, \quad
\frac{d u_3}{d t}=u_1 u_2-\gamma_3 u_3,
\end{aligned}
                        \]
 and exhibits chaotic behavior with $(\gamma_1, \gamma_2,\gamma_3) = (10,28,8/3)$. We consider $t \in [0, 18]$ and $N=50$ intervals, taking $\boldsymbol{u}_0=(-15, -15, 20)$ as initial condition. We use RK4 with $3e^2$ steps for the coarse solver $\g$, and RK4 with $2.25e^4$ steps for the fine solver $\f$. We use the normalized version and set a normalized $\epsilon=5e^{-7}$.
\begin{figure}[ht]
     \centering
     \begin{subfigure}[b]{0.5\textwidth}
        \centering
        \includegraphics[width=1\textwidth]{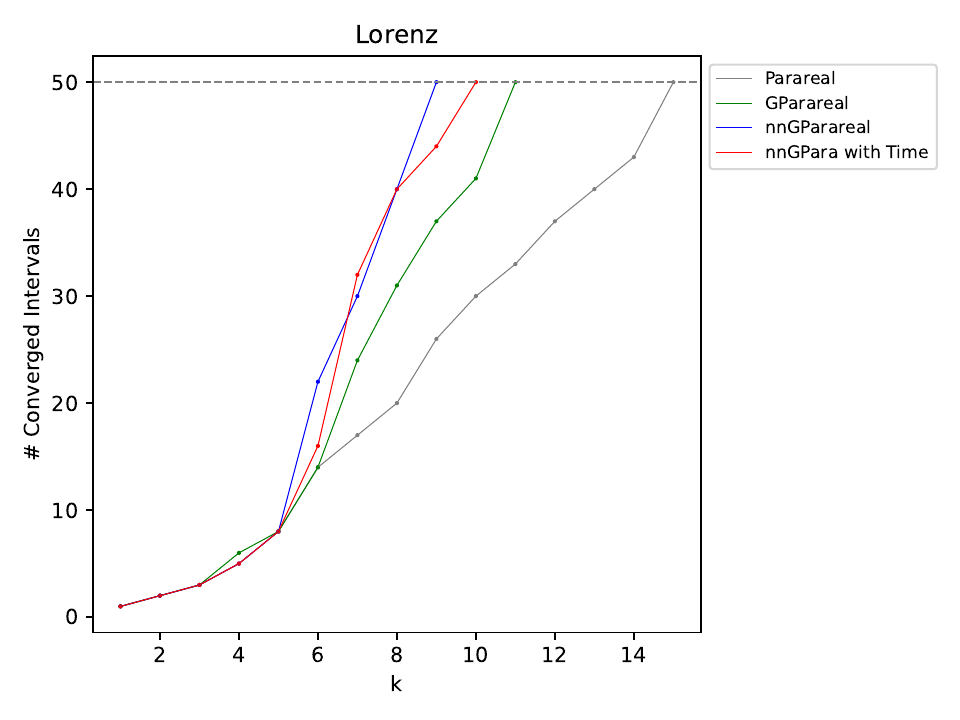}
     \end{subfigure}%
     \hfill
     \begin{subfigure}[b]{0.5\textwidth}
        \centering
        \includegraphics[width=1\textwidth]{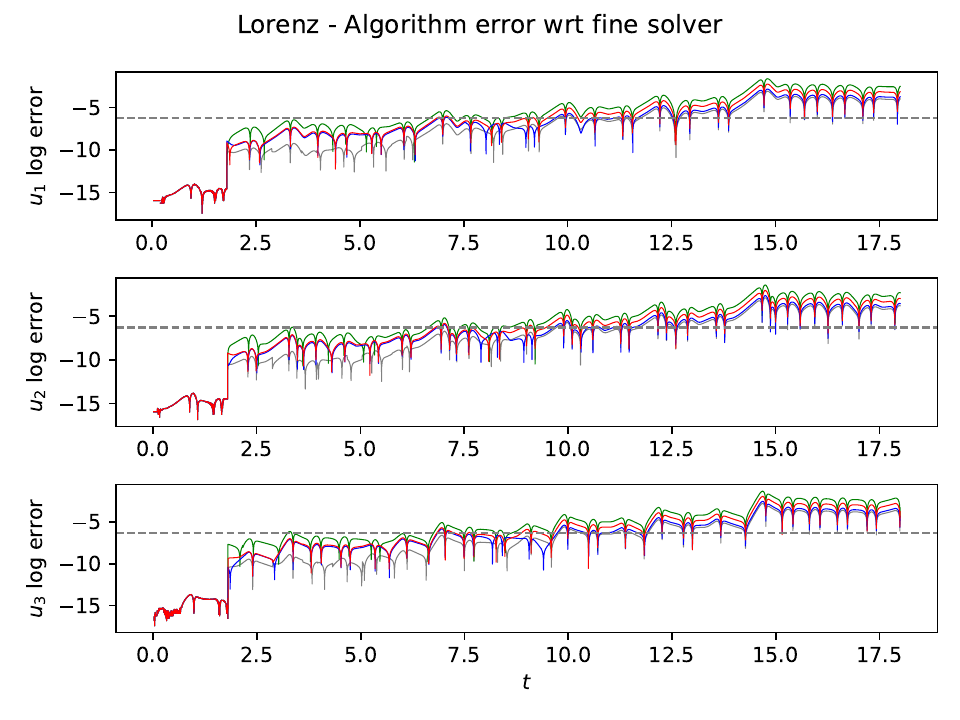}
     \end{subfigure}%
     \caption{Lorenz system. Left: number of converged intervals per iteration $k$. Right: accuracy of different Parareal models compared to the fine solver.}
\end{figure}

\subsection{FitzHugh-Nagumo PDE}
\label{sys:fhn_pde}
The two-dimensional, non-linear FHN PDE model \cite{ambrosio2009propagation} is an extension of the ODE system in Section~\ref{sys:fhn}. It represents a set of cells constituted by a small nucleus of pacemakers near the origin immersed among an assembly of excitable cells. The simpler FHN ODE system only considers one cell and its corresponding spike generation behavior. It is defined as
\begin{equation}
\begin{aligned}
    & v_t=a \nabla^2 v+v-v^3-w-c, \quad (x,t)\in (-L,L)^2 \times(t_0,t_N] \\
    & w_t=\tau\left(b \nabla^2 w+v-w\right),
\end{aligned}
    \label{eq:fhn_pde}
\end{equation}
for $v(x,t)$, $w(x,t)$ and Laplacian $\nabla^2 v=v_{xx}+v_{yy}$. We take initial conditions
\[
v(x, t_0)=v_0(x),\quad w(x, t_0)=w_0(x), \quad  x\in [-L,L],
\]
and boundary conditions
\[
\begin{aligned}
& v((x,-L), t)=v((x, L), t) \\
& v((-L, y), t)=v((L, y), t) \\
& v_y((x,-L), t)=v_y((x, L), t) \\
& v_x((-L, y), t)=v_x((L, y), t), \quad t \in [t_0, t_N].
\end{aligned}
\]
The boundary conditions for $w$ are taking the same. We discretize across spatial dimensions using finite difference and $\widetilde{d}$ equally spaced points, yielding an ODE with $d=2\widetilde{d}^2$ dimensions.  In the numerical experiments, we consider four values for $\widetilde{d}=10,12,14,16$, corresponding to $d=200, 288, 392, 512$. We set $N=512$, $L=1$, $t_0=0$, and take $v_0(x), w(0)$ randomly sampled from $[0,1]^d$ as initial condition. We use RK8 with $10^8$ steps for the fine solver $\f$. We use the normalized version with a normalized $\epsilon=5e^{-7}$. The time span and coarse solvers depend on $\widetilde{d}$, Table~\ref{tab:sim_setup_fhn_pde} describes their relation. This is to provide a realistic experiment where the user would need to adjust the coarse solver based on $t_N-t_0$.

\begin{table}
    \centering
    \begin{tabular}{cccc}
    \hline
    $d$&$\g$&$\g$ steps&$t_N$\\
        \hline
         $200$&RK2  & $3N$  & $t_N=150$ \\
         $288$&RK2  & $12N$  & $t_N=550$\\
         $392$&RK2  & $25N$  & $t_N=950$ \\
         $512$&RK4  & $25N$  & $t_N=1100$ \\
         \hline
    \end{tabular}
    \caption{Simulation setup for the two-dimensional FHN PDE. Adjusting the coarse solver based on the time horizon $t_N$ makes the simulation more realistic.}
    \label{tab:sim_setup_fhn_pde}
\end{table}

\section{\texorpdfstring{$m$}{m}-nn Parareal}
\label{app:choose_m}
In the left panel of Figure~\ref{fig:brus_dataset_vis_para_both}, we visualize the dataset $\mathcal{D}_6$ for the $2$-dimensional Brussellator collected up to iteration 6, together with a test observation $\bs{U}'= \bs{U}_{30}^6$ used during the predictor-corrector rule ($\ref{eq:update_rule_generic}$). We wish to compute $\fhat(\bs{U}_{30}^6)$ using the dataset $\mathcal{D}_6$. Note how only a few points are close to $\bs{U}_{30}^6$ in Euclidean distance. This can be better seen in the right plot, where we show $|\bs{U}-\bs{U}_{30}^6|$, the absolute coordinate-wise difference from every observation $\bs{U} \in \mathcal{U}_6$ and $\bs{U}_{30}^{6}$. The log transformation on each axis shows that a few points are very close, while the majority is far away. These nearest neighbors tend to be previous approximations of the same initial condition at previous iterations. In this example, $\bs{U}_{30}^{k}, k < 6$. Intuitively, $\fhat_{m-\text{nn}}$ is not optimal when using uniform weights since closer inputs carry more information in a local smoothing task. However, it is hard to say whether taking them to be inversely proportional to the distance will be optimal in all situations. The best approach is to learn the weights.

\begin{figure}[ht!]
     \centering
    \includegraphics[width=1\textwidth]{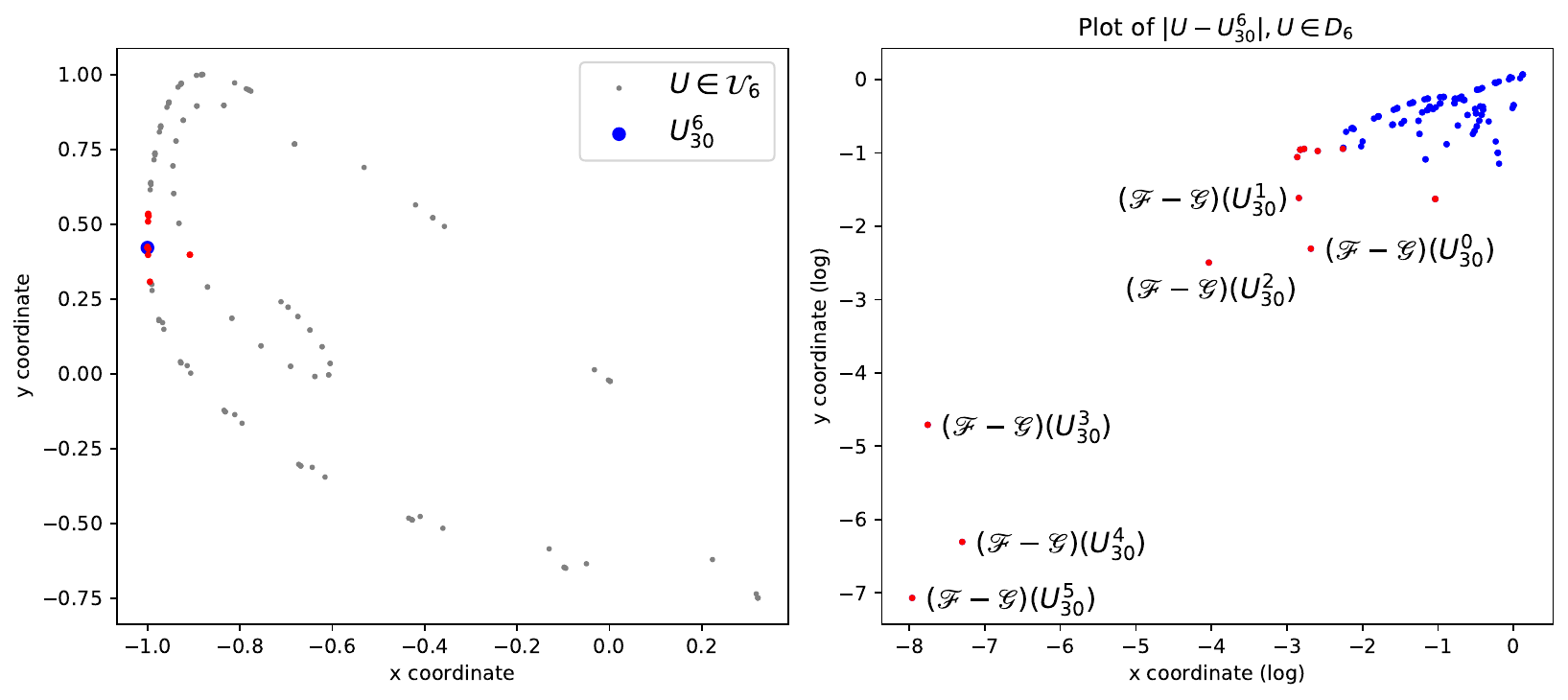}
    \caption{Visualization of the dataset for the two-dimensional Brusselator system (\ref{sys:brus}) Left: a scatterplot of the observations $\bs{U}\in\mathcal{U}_6$ accumulated by iteration 6 (gray), with the test observation $\bs{U}_{30}^{6}$ (blue); $m=15$ nearest neighbors to the test observations are marked in red. Note how far most of the points are from $\bs{U}_{30}^{6}$. Right: the absolute (log) distance coordinate-wise between $\bs{U}\in \mathcal{U}_6$ and $\bs{U}_{30}^{6}$.}
     \label{fig:brus_dataset_vis_para_both}
\end{figure}

\section{Choice of the reduced subset for the GP}
\label{app:nngp_time}
In this section, we compare the performance of alternative
approaches for selecting $m$ datapoints from the full dataset $\mathcal{D}_k$ at iteration $k$, arguing that they are either equivalent or suboptimal than the nearest
neighbors (nns). To do so, we divide them into rules-of-thumb and learned approaches, starting from the former.

An important element that has been left out of the algorithm until now is time. The initial conditions are naturally arranged temporally, where for a given iteration $k$, $\bs{U}_i^k$ depends on the value of $\bs{U}_{i-1}^k$. This gives rise to a time-ordered sequence $\bs{U}_0^k,...,\bs{U}_{N-1}^k$. Moreover, there is also a temporal evolution over Parareal iterations as the initial conditions converge, yielding the sequence $\bs{U}_i^0,...,\bs{U}_i^K$ for every interval $i$. We can picture this as a $K \times N$ matrix where each cell $(k, i)$ contains the initial condition $\bs{U}_{i}^{k-1}$. In short, the first $k$ rows contain the inputs $\mathcal{U}_k$ of the dataset $\mathcal{D}_k$. Then, we can refer to the sequence $\bs{U}_0^k,...,\bs{U}_{N-1}^k$ as a \textit{row} along the matrix, and to the sequence $\bs{U}_i^0,...,\bs{U}_i^K$ as a \textit{column}. When making a prediction $\fhat(\bs{U}_{i-1}^k)$, all we have available are the top $k$ rows, and we want to understand the best way to use the observations. GParareal's answer is to use all the data. Instead, nnGParareal ignores the temporal arrangement (rows and columns) and uses the closest cells (initial conditions) \textit{in Euclidean distance}. Here, we consider other strategies for building $\mathcal{D}_{i-1,k}$, the reduced dataset of size $m$:
\medskip
\begin{itemize}[leftmargin=*]
    \item \textit{Col + rnd}. Let $\mathcal{D}_{i-1,k}$ be formed of the \textit{column} $i-1$ up to row $m$, $\bs{U}_{i-1}^0,...,\bs{U}_{i-1}^m$, together with their corresponding outputs. If $k < m$, fill the remaining $m-k$ observations by randomly sampling from $\mathcal{D}_k$.
    \item \textit{Col only}. Similar to the above, but without sampling random observations if $k < m$. Note that this sets $m=k$, so the dataset size varies across iterations.
    \item \textit{Row + col}. Expand radially from the test observation $\bs{U}_{i-1}^k$ by striking a trade-off between previous iterations (column entries) and nearby intervals (row entries).
    \item \textit{Row-major}. Give priority to nearby intervals, thus expanding horizontally across columns first.
    \item \textit{Column-major}. Give priority to previous iterations, thus expanding vertically across rows first.
\end{itemize}
In Table~\ref{tab:nngp_time_heur}, we compare the performance of the rules-of-thumb above with the nns across several systems. As anticipated, the nn strategy achieves the best performance across all systems. 
\begin{table}
    \centering
    \small
\begin{tabular}{lcccccc}
\hline
Subset choice & FHN  & R\"ossler  & Hopf  & Brusselator  & Lorenz  & Double Pendulum  \\
\hline
 nn &  \textbf{5}  &  \textbf{12}  &  \textbf{10}  &  \textbf{17}  &  \textbf{10}  &  \textbf{10}  \\
 Col + rnd &  8  &  14  &  \textbf{10}  &  19  &  13  &  11  \\
 Col only &  8  &  17  &  \textbf{10}  &  $-$  &  13  &  15  \\
 Row + col &  8  &  17  &  \textbf{10}  &  $-$  &  13  &  12  \\
 Row-major &  10  &  21  &  30  &  $-$  &  12  &  13  \\
 Column-major &  7 &  16 &  \textbf{10} &  20 &  13 &  13 \\\hline
\end{tabular}
\caption{Simulation results for a range of heuristic choices of data subset for nnGParareal. The entries are the number of iterations to convergence. The dash indicates that the algorithm failed to converge. Where possible, $m$ has been set equal to $11$. The nearest neighbors (nns) consistently achieve the best performance.}
\label{tab:nngp_time_heur}
\end{table}

Of all the strategies mentioned above, none explicitly considered the spatial distance. On the other hand, the nns do not consider time. Like before, when selecting the optimal weighting scheme for the neighbors, we employ learning to bridge the gap between the two approaches. It might not be easy, a priori, to know how much importance to give to each factor between rows, columns, and spatial information. However, we can use the data and the GP kernel $\mathcal{K}_{\rm GP}(\cdot, \cdot)$ to learn the importance of each. Note that the kernel in \eqref{eq:se_kern} already ranks observations based on their Euclidean distance, so it is enough to generalize it to have also access to temporal information. 

Drawing inspiration from \cite{datta2016nonseparable}, we incorporate the time information carried by interval $i$ and iteration $k$ into our kernel. We enrich the dataset $\mathcal{D}_k$, denoted as $\widetilde{\mathcal{D}_k}$, by including a tuple $(i,k)$ indicating which cell the observation corresponds to in the matrix. By assigning independent length scales to space $\sigma_s^2$, intervals $\sigma_i^2$ (rows), and iterations $\sigma_k^2$ (columns), we can determine the relative importance of each in a structured, automated way by tuning them according to the marginal log-likelihood function. For this, let $\widetilde{\bs{U}}_i^k=(\bs{U}_i^k,i, k)$ be an enriched observation. We employ the following compositional kernel:
\begin{equation*}
\begin{aligned}
    \mathcal{K}_{\rm GP} (\widetilde{\bs{U}}_i^k, \widetilde{\bs{U}}_{i'}^{k'}  ) &= \mathcal{K}_{\rm GP}\left ((\bs{U}_i^k, i, k), (\bs{U}_{i'}^{k'}, i', k')\right ) \\
    &= {\sigma_{\rm o}^2} \mathcal{K}_{\rm GP}^1 (\widetilde{\bs{U}}_i^k, \widetilde{\bs{U}}_{i'}^{k'} ) \mathcal{K}_{\rm GP}^2 (\widetilde{\bs{U}}_i^k, \widetilde{\bs{U}}_{i'}^{k'}  ) \mathcal{K}_{\rm GP}^3 (\widetilde{\bs{U}}_i^k, \widetilde{\bs{U}}_{i'}^{k'}  ), 
\end{aligned}
    \label{eq:kernel_time}
\end{equation*}

with
\begin{align*}
    \mathcal{K}_{\rm GP}^1 (\widetilde{\bs{U}}_i^k, \widetilde{\bs{U}}_{i'}^{k'}  ) &= \exp\left(-\|\bs{U}_i^k-\bs{U}_{i'}^{k'}\|^2/ \sigma_s^2\right),\\
\mathcal{K}_{\rm GP}^2(\widetilde{\bs{U}}_i^k, \widetilde{\bs{U}}_{i'}^{k'}  ) &= \exp(-(i-i')^2/\sigma_i^2),\\
\mathcal{K}_{\rm GP}^3 (\widetilde{\bs{U}}_i^k, \widetilde{\bs{U}}_{i'}^{k'}  ) &= \exp(-  (k-k')^2/\sigma_k^2),
\end{align*}
where $\sigma_{\rm o}^2$ is the output length scale, as in \eqref{eq:se_kern}. Collectively, $\bs{\theta}=(\sigma_{\rm o}^2, \sigma_s^2, \sigma_i^2, \sigma_k^2, \sigma_{\rm reg}^2)$ form the hyperparameters of the GP augmented to include time information. When combined with Parareal and nn training, we refer to this as \texttt{nnGParareal with Time}.  

We aim to change the choice of the $m$ neighbors to account for time. Since the kernel estimates the correlation between observations, it seems natural to build $\widetilde{D}_{i-1,k}$ for a prediction $\fhat(\widetilde{\bs{U}}_{i-1}^k)$ out of the $m$ pairs $(\widetilde{\bs{U}}, (\f-\g)(\widetilde{\bs{U}}))$ which have the highest kernel score $\mathcal{K}_{\rm GP}(\widetilde{\bs{U}}, \widetilde{\bs{U}}_{i-1}^k)$, $\widetilde{\bs{U}} \in  \widetilde{\mathcal{U}}_k$. This was already identified by \cite{vecchia1988estimation} as a natural strategy for subset selection and is also what is used in\cite{datta2016nonseparable}. However, as pointed out in the same paper, the ranking operation requires knowledge of the optimal hyperparameters $\bs{\theta}^*$ to evaluate the kernel $\mathcal{K}_{\rm GP}(\widetilde{\bs{U}}, \widetilde{\bs{U}}_{i-1}^k)$. In general, $\bs{\theta}^*$ is not available unless a data subset $\widetilde{\mathcal{D}}_{i,k}$ is proposed first, and the marginal log-likelihood is minimized for $\bs{\theta}$. This suggests a computationally expensive iterative procedure. Propose $\widetilde{\mathcal{D}}_{i-1,k}^0$ at random and estimate $\bs{\theta}_0^*$ conditional on $\widetilde{\mathcal{D}}_{i-1,k}^0$. Using $\bs{\theta}^*_0$, rank the observation in $\widetilde{D}_k$ to obtain a new subset $\widetilde{\mathcal{D}}_{i-1,k}^1$. Estimate the new optimal hyperparameters $\bs{\theta}_1^*$ based on $\widetilde{\mathcal{D}}_{i-1,k}^1$. Repeat until the data subset has stabilized or for an arbitrary number of iterations. Since each iteration requires optimizing $\boldsymbol{\theta}$, which is usually expensive, such a strategy has historically been discarded in favor of more straightforward rules. For example, \cite{datta2016hierarchical} show empirically that simply choosing the $m$ nearest neighbors gives excellent performance compared to the full GP. Furthermore, \cite{datta2016nearest} states: “As nearest neighbors correspond to points with highest spatial correlation, this choice produces excellent approximations." We observed in practice that two or three iterations of the procedure are sufficient for the subsets $\widetilde{D}_{i-1,k}^{\cdot}$ to stabilize in most cases. Moreover, the converged subset consisted of the $m$ closest spatial neighbors in over 80\% of the cases, as seen Figure~\ref{fig:nngp_time_lorenz_final_subset} for the Lorenz system. Finally, to verify that the nnGParareal with Time does not provide further performance gains over nnGParareal, we compare the two over a range of systems in Table~\ref{tab:gp_vs_nngp7_with_time}. Since they show equivalent performance, we prefer the simpler nn variant, which is much cheaper to run.

\begin{table}[t]
{
\footnotesize
    \centering
    \begin{tabular}{lcccccc}
System& FHN & R\"ossler& Hopf& Bruss. & Lorenz  & Dbl Pendulum \\
\hline
$N$& $40$ & $40$& $32$& $32$ & $50$  & $32$ \\
\hline
 \multicolumn{7}{c}{Number of iterations for each algorithm until convergence with a normalized accuracy $\epsilon=5e^{-7}$} \\ \hline
Parareal  &  11  &  18  &  19  &  19  &  15  &  15 \\
 GParareal  &  \textbf{5}  &  13  &  10  &  20  &  11  &  \textbf{10} \\
 nnGParareal&  \textbf{5}  &  \textbf{12}  &  \textbf{9}  &  \textbf{17}  &  \textbf{9}  &  \textbf{10} \\ \hline
nnGParareal with Time &  \textbf{5} &  \textbf{12} &  \textbf{9} &  \textbf{17} &  \textbf{9} &  \textbf{10}\\\hline
    \end{tabular}
    }
    \caption{Number of iterations required by Parareal, GParareal, nnGParareal and nnGParareal with Time to converge for the systems described in  Supplementary Material~\ref{app:models} and Sections \ref{tab:nonaut_scaling} and \ref{sec:exp_tomlab}.} %
    \label{tab:gp_vs_nngp7_with_time}
\end{table}
\begin{figure}[ht]
     \centering
     \includegraphics[width=0.6\linewidth]{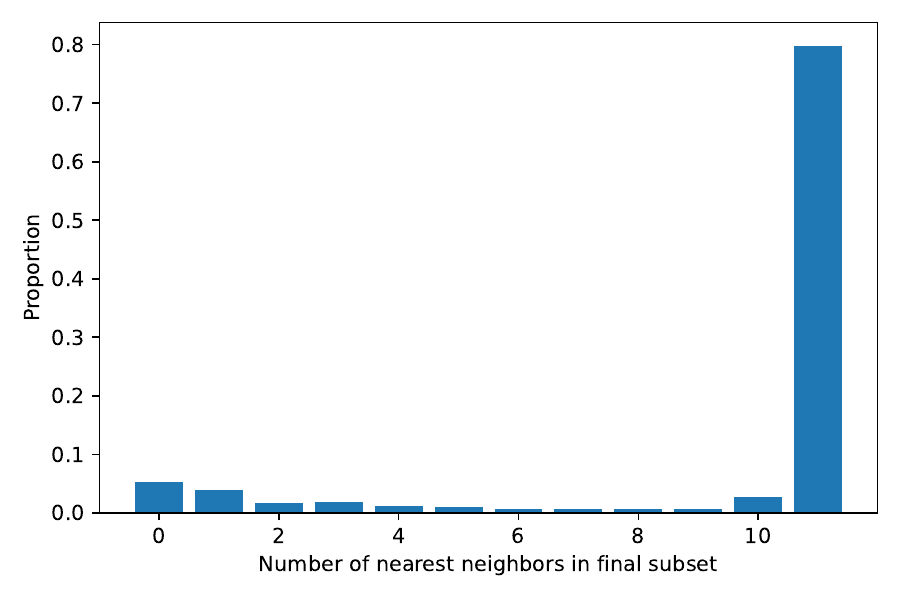}
     \caption{Histogram of the number of observations in the converged subset $\widetilde{D}_{i-1,k}^{last}$ which are nearest neighbors. The histogram has been computed aggregating over $i, k$ and multiple independent random samples of  $\widetilde{D}_{i-1,k}^{0}$.}
     \label{fig:nngp_time_lorenz_final_subset}
 \end{figure}
We conclude this section with a brief overview of the relevant literature. \cite{emery2009kriging} shows that the optimal neighbor sets are not necessarily nested as $m$ increases, which they are when constructed using a distance. \cite{stein2004approximating} propose to use a combination of points spatially near and far from the test observation $\bs{U}_i^k$. However, the authors note their approach becoming effective for relatively large values of $m$, around $m=35$. These results were later criticized by \cite{datta2016hierarchical}, who conducted extensive simulations, failing to find evidence of the benefit of adding far away points. Finally, we mention the approach of \cite{gramacy2015local}, implemented in the R package laGP \cite{gramacy2015local}, which iteratively builds the data subset by adding the observation leading to the maximum reduction in posterior variance. Through a series of approximations, they incrementally build $\mathcal{D}_{i-1,k}$ at $O(m^3)$ cost, comparable to that of the nnGP algorithm. However, they acknowledge a longer runtime than nnGP, which could be used to invert a larger subset $\mathcal{D}_{i-1,k}$. Although theoretically elegant, we found the approximate variance reduction estimate of LaGP to be too unstable with a close-to-singular correlation matrix, as in our case. One could try to stabilize it with proper tuning of the GP noise parameter $\sigma_{\rm reg}^2$. However, we did not explore this approach further due to its high computational cost.

\begin{figure}
     \centering
\includegraphics[width=0.6\linewidth]{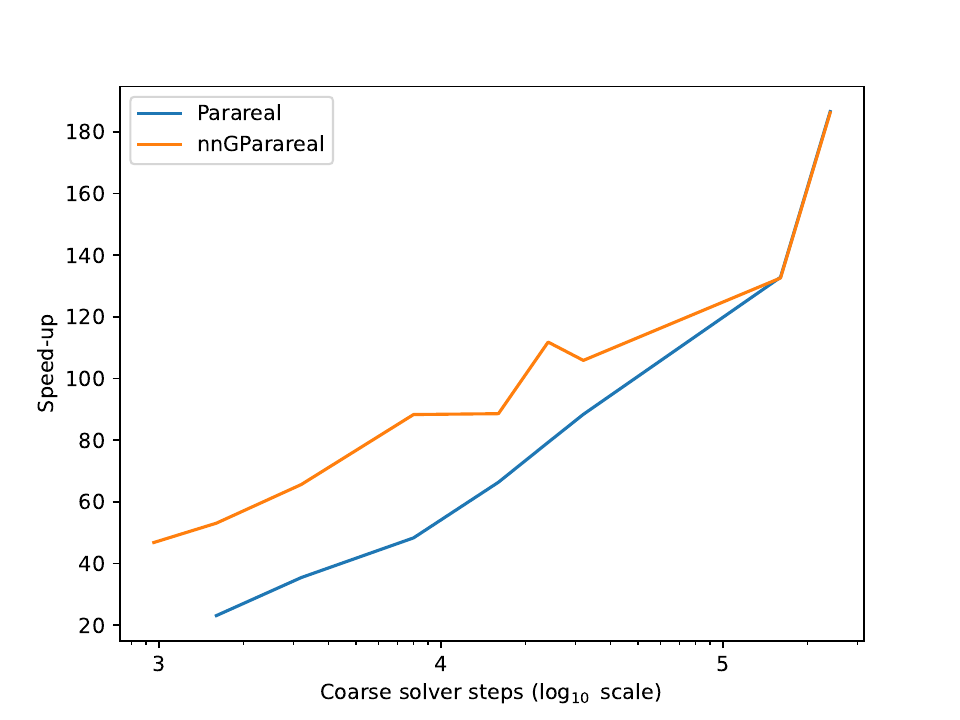}
     \caption{Effect of increased coarse solver precision (i.e., number of steps) on both Parareal and nnGParareal speed-up, for the Hopf bifurcation system with $N=32$ and $1.74e^8$ fine solver steps. 
     Parareal fails to converge for $10^3$ steps (missing data in the plot), while nnGParareal converges with a notable speed-up.}
     \label{Gextra}
 \end{figure}
\section{Hopf bifurcation's performance across \texorpdfstring{$N_\mathcal{G}$}{NG}}

\label{extraHopf}
Figure \ref{Gextra} shows the speed-ups of Parareal and nnGParareal as a function of the number of corse solver (RK1) steps for the Hopf bifurcation system with $N=32$ and $1.74 e^8$ steps of the fine solver (RK8). While the speed-ups eventually become similar for more than $10^5$ coarse steps, overall nnGParareal allows to use much cheaper coarse solver precision than Parareal. Indeed, the latter requires many more coarse steps than nnGParareal to obtain a similar speed-up (e.g., $10^4$ vs $10^3$ for a speed-up of 50), and does not converge for $10^3$ coarse solver steps.

\section{Thomas Labyrinth Scalability Table}
Table~\ref{tab:tomlab} reports more complete experimental results for the scalability study run on Thomas labyrinth in Section~\ref{sec:exp_tomlab}.
\label{app:tomlab_scal_table}

\begin{table}[ht]
    \centering
{\scriptsize
   
\begin{tabular}{lccccccc}

    \multicolumn{7}{c}{\footnotesize Thomas Labyrinth, with $N=32, d=3$, and $ t_N=10$}\\\\
    \hline
    Algorithm & $K_{\rm alg}$ & $T_\g$ & $T_\f$ & $\tmdl$ & $T_{\rm alg}$ & $\hat S_{\rm alg}$ & $S_{\rm alg}^*$\\
    \hline
    Fine & $-$ & $-$ & $-$ & $-$ & 3.18e+04 & 1&1\\
    Parareal & 30 & 1.82e-02 & 1.05e+03 & 2.47e-03 & 3.16e+04 & 1.01&1.07\\
    GParareal & 25 & 2.17e-02 & 1.06e+03 & 2.67e+01 & 2.65e+04 & 1.20&1.28\\
    nnGParareal & 24 & 2.17e-02 & 1.04e+03 & 7.50e+00 & 2.49e+04 & \textbf{1.27}&\textbf{1.33}\\
    \hline\\
    \multicolumn{7}{c}{\footnotesize $N=64, d=3$, and $ t_N=10$}\\
    \hline
    Algorithm & $K_{\rm alg}$ & $T_\g$ & $T_\f$ & $\tmdl$ & $T_{\rm alg}$ & $\hat S_{\rm alg}$ & $S_{\rm alg}^*$\\
    \hline
    Fine & $-$ & $-$ & $-$ & $-$ & 3.18e+04 & 1&1\\
    Parareal & 20 & 3.57e-02 & 5.34e+02 & 3.98e-03 & 1.07e+04 & 2.98&3.20\\
    GParareal & 15 & 4.58e-02 & 5.56e+02 & 3.86e+01 & 8.38e+03 & \textbf{3.79}&\textbf{4.26}\\
    nnGParareal & 16 & 4.29e-02 & 5.42e+02 & 1.26e+01 & 8.70e+03 & 3.65&4.00\\
    \hline\\
    
    \multicolumn{7}{c}{\footnotesize $N=128, d=3$, and $ t_N=40$}\\
    \hline
    Algorithm & $K_{\rm alg}$ & $T_\g$ & $T_\f$ & $\tmdl$ & $T_{\rm alg}$ &  $\hat S_{\rm alg}$ & $S_{\rm alg}^*$\\
    \hline
    Fine & $-$ & $-$ & $-$ & $-$ & 3.18e+04 & 1&1\\
    Parareal & 100 & 3.83e-02 & 2.70e+02 & 3.11e-02 & 2.70e+04 & 1.18&1.28\\
    GParareal & 64 & 4.78e-02 & 2.77e+02 & 1.86e+04 & 3.63e+04 & 0.87&1.5\\
     nnGParareal & 63 & 4.65e-02 & 2.74e+02 & 8.19e+01 & 1.73e+04 & \textbf{1.83}&\textbf{2.03}\\
    \hline\\
    
    \multicolumn{7}{c}{\footnotesize $N=256, d=3$, and $ t_N=100$}\\
    \hline
    Algorithm & $K_{\rm alg}$ & $T_\g$ & $T_\f$ & $\tmdl$ & $T_{\rm alg}$ & $\hat S_{\rm alg}$ & $S_{\rm alg}^*$\\
    \hline
    Fine & $-$ & $-$ & $-$ & $-$ & 3.18e+04 & 1&1\\
    Parareal & 256 & 6.85e-02 & 1.35e+02 & 1.50e-01 & 3.46e+04 & 0.92&1.00\\
    GParareal & $>$ 60 & 1.33e-01 & 1.41e+02 & $>$ 1.64e+05 & $>$ 1.73e+05 & $<$ 0.18&\\
    nnGParareal & 159 & 8.22e-02 & 1.56e+02 & 7.36e+02 & 2.56e+04 & \textbf{1.24}&\textbf{1.60}\\
    \hline\\

    \multicolumn{7}{c}{\footnotesize $N=512, d=3$, and $ t_N=100$}\\
    \hline
    Algorithm & $K_{\rm alg}$ & $T_\g$ & $T_\f$ & $\tmdl$ & $T_{\rm alg}$ & $\hat S_{\rm alg}$ & $S_{\rm alg}^*$\\
    \hline
    Fine & $-$ & $-$ & $-$ & $-$ & 3.18e+04 & 1&1\\
    Parareal & 180 & 1.33e-01 & 6.97e+01 & 2.14e-01 & 1.26e+04 & 2.53&2.84\\
    GParareal & $> 25$ & 3.17e-01 & 8.86e+01 & $> 1.71e+05$ & $> 1.73e+05$ & $< 0.18$&\\
    nnGParareal & 69 & 1.69e-01 & 8.02e+01 & 4.31e+02 & 5.98e+03 & \textbf{5.32}&\textbf{7.17}\\
    \hline
    \end{tabular}\\   
    }
    \caption{Scalability study for Thomas labyrinth (\ref{eq:tomlab}). Time is reported in seconds. $T_\f$ and $T_\g$ refer to the runtime per iteration of the fine and coarse solvers, respectively. $K_{\rm alg}$ denotes the number of iterations to converge. $\tmdl$ includes the running time of training, hyper-parameter selection, and prediction for the corresponding model. $T_{\rm alg}$ is the runtime of the algorithm. $\hat S_{\rm alg}$ and $S_{\rm alg}^*$ refer to the empirical and maximum theoretical speed-up, respectively. For nnGParareal, we choose $m=18$ nns. For $N=256, 512$, GParareal failed to converge within the computational time budget.}
    \label{tab:tomlab}
    \end{table}

\section{Scalability results for the FHN PDE model}
Table~\ref{tab:fhn_pde} reports more complete experimental results for the scalability study run on the FitzHugh-Nagumo PDE model in Section~\ref{sec:exp_fhn}.
\label{app:FHN_PDE_scal}
 \begin{table}[ht!]
    \centering
    \scriptsize
\begin{tabular}{lccccccc}
    \multicolumn{8}{c}{\footnotesize FitzHugh-Nagumo PDE, with $N=512, d=200$, and $ t_N=150$}\\\\
    \hline
    Algorithm & $K_{\rm alg}$ & $T_{\g}$ & $T_{\f}$ & $\tmdl$ & $T_{\rm alg}$ & $\hat S_{\rm alg}$ & $S_{\rm alg}^*$\\
    \hline
    Fine & $-$ & $-$ & $-$ & $-$ & 1.02e+05 & 1 & 1\\
    Parareal & 25 & 3.09e-01 & 2.05e+02 & 9.20e-02 & 5.14e+03 & \textbf{19.80}& 20.48\\
    GParareal & 8 & 3.74e-01 & 2.13e+02 & 5.63e+04 & 5.80e+04 & 1.75& \textbf{64}\\
    nnGParareal & 12 & 4.66e-01 & 2.02e+02 & 3.87e+03 & 6.31e+03 & 16.12& 42.66\\
    \hline\\
    
    \multicolumn{8}{c}{\footnotesize $N=512, d=288$, and $ t_N=550$}\\
    \hline
    Algorithm & $K_{\rm alg}$ & $T_{\g}$ & $T_{\f}$ & $\tmdl$ & $T_{\rm alg}$ & $\hat S_{\rm alg}$ & $S_{\rm alg}^*$\\
    \hline
    Fine & $-$ & $-$ & $-$ & $-$ & 1.79e+05 & 1 & 1\\
    Parareal & 67 & 3.42e-01 & 3.55e+02 & 3.38e-01 & 2.39e+04 & 7.50& 7.64\\
    GParareal & 7 & 4.59e-01 & 3.57e+02 & 3.10e+04 & 3.35e+04 & 5.34& \textbf{73.14} \\
    nnGParareal & 10 & 5.53e-01 & 3.53e+02 & 4.26e+03 & 7.80e+03 & \textbf{22.95}& 51.2\\
    \hline\\
    
    \multicolumn{8}{c}{\footnotesize $N=512,d=392$, and $t_N=950$}\\
    \hline
    Algorithm & $K_{\rm alg}$ & $T_{\g}$ & $T_{\f}$ & $\tmdl$ & $T_{\rm alg}$ & $\hat S_{\rm alg}$ & $S_{\rm alg}^*$\\
    \hline
    Fine & $-$ & $-$ & $-$ & $-$ & 2.87e+05 & 1 & 1\\
    Parareal & 44 & 5.41e-01 & 6.19e+02 & 2.21e-01 & 2.73e+04 & 10.51 & 11.63\\
    GParareal &  $>$ 11 & 6.32e-01 & 6.29e+02 & $>$ 1.37e+05 & $>$ 1.73e+05 & $<$ 0.18& $>$ 46.54\\
    nnGParareal & 5 & 8.25e-01 & 6.61e+02 & 3.08e+03 & 6.39e+03 & \textbf{44.81}& \textbf{102.4}\\
    \hline\\
    
    \multicolumn{8}{c}{\footnotesize $N=512,d=512$, and $t_N=1100$}\\
    \hline
    Algorithm & $K_{\rm alg}$ & $T_{\g}$ & $T_{\f}$ & $\tmdl$ & $T_{\rm alg}$ & $\hat S_{\rm alg}$ & $S_{\rm alg}^*$\\
    \hline
    Fine & $-$ & $-$ & $-$ & $-$ & 8.17e+05 & 1 & 1\\
    Parareal & 79 & 1.27e+00 & 2.09e+03 & 4.27e-01 & 1.65e+05 & 4.95& 6.48\\
    GParareal &  $>$ 9 & 1.01e+00 & 2.06e+03 & $>$ 1.25e+05 & $>$ 1.73e+05 & $<$ 0.18& $>$56.88\\
    nnGParareal & 6 & 1.85e+00 & 2.07e+03 & 5.39e+03 & 1.78e+04 & \textbf{45.80}& \textbf{85.33}\\
    \hline
    \end{tabular}\\    
    \caption{ Simulation study on the empirical scalability in $d$ of GParareal and nnGParareal for the FitzHugh-Nagumo PDE model (\ref{eq:fhn_pde}). $T_\f$ and $T_\g$ refer to the runtime per iteration of the fine and coarse solvers, respectively. $K_{\rm alg}$ denotes the number of iterations taken to converge. $\tmdl$ includes the running time of training, hyper-parameter selection, and prediction for the corresponding model employed. $T_{\rm alg}$ is the runtime of the algorithm. $\hat S_{\rm alg}$ and $S_{\rm alg}^*$ refer to the empirical and maximum theoretical speed-up, respectively. For $d=392, 512$, GParareal failed to converge within the computational time budget. For nnGParareal, we choose $m=20$ nns. Time is reported in seconds.}
        \label{tab:fhn_pde}
    \end{table}

\section{Burgers' performance across \texorpdfstring{$m$}{m}}
\label{app:burg_perf_across_m}
This section reports additional experimental results for the Burgers' equation in Section~\ref{sec:exp_burg}. In particular, we study the impact of different dataset sizes $m$ on nnGParareal performance. Analogous analysis has been carried out for other systems, such as Thomas labyrinth and FitzHugh-Nagumo PDE, which yielded similar results and hence are not reported here. For both $T=5$ and $T=5.9$, Figure~\ref{fig:Burges_perf_across_m_K} reports the histogram of the iterations to convergence across several choices of $m$ ranging from 11 to 30. For each $m$, we executed 100 independent runs of nnGParareal to account for the algorithm's randomness coming from the log-likelihood exploration (see Supplementary Material~\ref{app:impl_det}). The performance is largely unaffected by changes in $m$, with minor gains resulting from  higher $m$, as could be expected. Figure~\ref{fig:Burges_perf_across_m_speed} displays the empirical distribution of the speed-up corresponding to the simulations in Figure~\ref{fig:Burges_perf_across_m_K}. Although nnGParareal takes longer to converge than GParareal for Burgers' equation, on average, the speed-up gains are always greater due to the lower training cost of the model.
\begin{figure}[ht]
     \centering
    \includegraphics[width=0.49\textwidth]{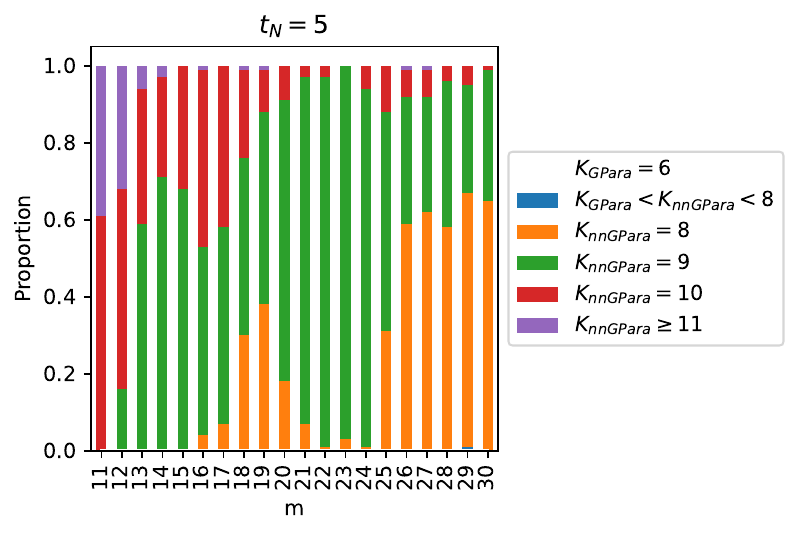}
    \includegraphics[width=0.49\textwidth]{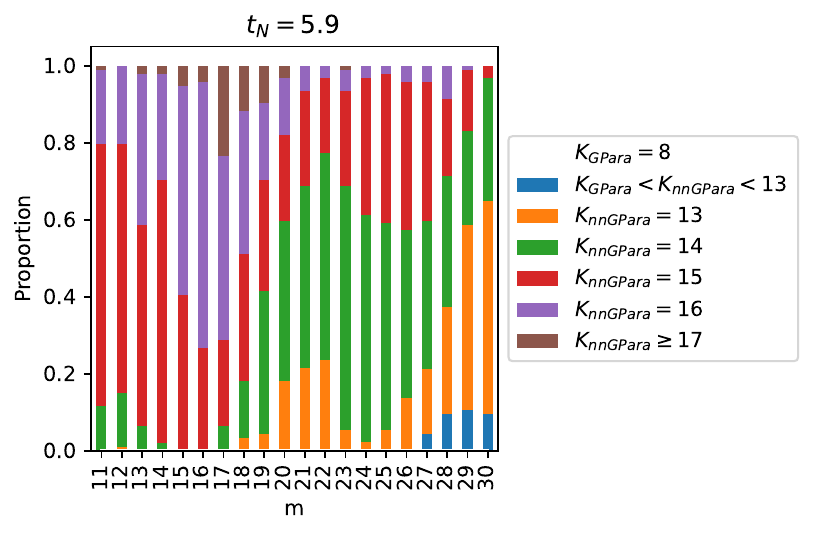}
    \caption{nnGParareal: empirical distribution of $K$ across values of $m$ for viscous Burgers' equation with 100 independent runs for each $m$. The shaded area indicates better performance than GParareal (never achieved, thus not visible).}
    \label{fig:Burges_perf_across_m_K}
\end{figure}

\begin{figure}[ht]
     \centering
    \includegraphics[width=0.49\textwidth]{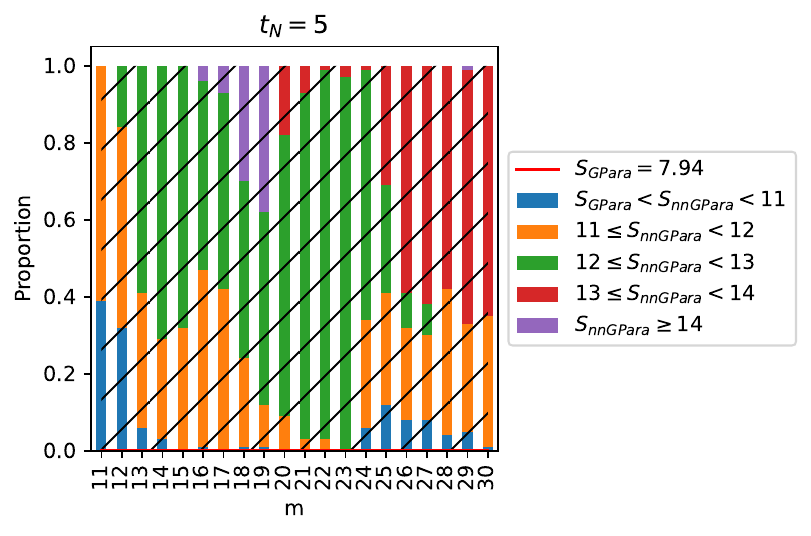}
    \includegraphics[width=0.49\textwidth]{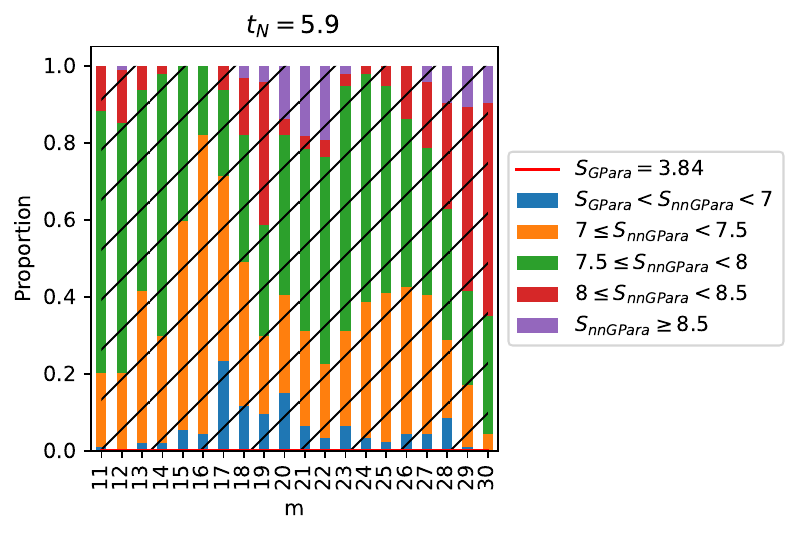}
    \caption{nnGParareal: empirical distribution of the speed-up across values of $m$ for viscous Burgers' equation with 100 independent runs for each $m$. The shaded area indicates better performance than GParareal.}
    \label{fig:Burges_perf_across_m_speed}
\end{figure}

\section{Heat equation}
\label{app:heat}
We consider the one‐dimensional heat equation on the spatial interval $x\in(0,L)$ over time $t\in(0,T]$,
\[
\frac{\partial u}{\partial t}(x,t) \;=\;\alpha\,\frac{\partial^2 u}{\partial x^2}(x,t),
\]
subject to homogeneous Dirichlet boundary conditions 
\[
u(0,t)=u(L,t)=0,\quad t\in(0,T],
\]
and initial condition $u(x,0)=u_0(x)$. The heat equation describes, among other phenomena, heat distribution in a homogeneous isotropic medium, where $u(x,t)$ represents temperature at position $x$ and time $t$, and $\alpha$ denotes thermal diffusivity.
We discretize the one dimensional spatial domain using finite difference~\cite{fornberg1988generation}, defining $d+1$ equally spaced points $x_{j+1}=x_j+\Delta x$, with $\Delta x = L/d$, $j=0,\ldots,d$, resulting in a $d$-dimensional ODE system. We choose parameters $L=1$, $T=2$, $\alpha=0.1$, $d=40$, and $u_0(x)=\sin(2\pi x)$ for our numerical experiments. Given the periodic initial condition, the exact solution is
\begin{equation}\label{heat_eq}u(x,t)=\exp\bigl(-0.4\pi^2\,t\bigr)\,\sin(2\pi x).
\end{equation}
\begin{table}
    \centering
    {\small
\begin{tabular}{llc}
\multicolumn{3}{c}{}\\\\
\hline
Parameter & Description \\
\hline
$N$ & Number of Parareal intervals & 300\\
$-$ & Number of processors for parallelization (set equal to $N$) & 300\\
$\mathscr{F}$ & Fine solver & Runge-Kutta 8\\
$\mathscr{G}$ & Coarse solver & Runge-Kutta 2\\
$N_{\mathscr{F}}$ & Number of temporal discretization steps per interval for $\mathscr{F}$ & 1000\\
$N_{\mathscr{G}}$ & Number of temporal discretization steps per interval for $\mathscr{G}$ & 2\\
$N \times N_{\mathscr{G}}$ & Total number of temporal steps for $\mathscr{G}$ across all intervals & 600\\
\hline
\end{tabular}}
\caption{Parareal simulation parameters}
\label{table1}
\end{table}
Table~\ref{table1} summarizes the Parareal simulation parameters. To ensure balanced discretization errors, we estimate temporal and spatial errors of the fine sequential solver $\mathscr{F}$ independently, proceeding as detailed before:
\begin{itemize}
\item Increase spatial (resp. temporal) resolution by a factor of $15$.
\item Keeping temporal (resp. spatial) resolution fixed, numerically approximate the solution.
\item Measure temporal (resp. spatial) discretization errors, $e_t$ (resp. $e_s$), using the Euclidean distance to the exact solution \eqref{heat_eq}.
\end{itemize}
With this approach, we obtained time and space discretization errors equal to $e_s=2.13 \times 10^{-4}$ and $e_t=8.32 \times 10^{-4}$, respectively, confirming that both errors are indeed of the same order of magnitude and balanced. Using this balanced resolution, we compared the runtimes of Parareal and nnGParareal against those of the sequential fine solver. The algorithms demonstrated significant speedups even with a minimal number of coarse solver steps, as shown in Table~\ref{tab_res}.
\begin{table}
    \centering
    {\small
\begin{tabular}{lcccccc}
\hline
Algorithm & $K_{\rm alg}$ & $T_{\mathscr{G}}$ & $T_{\mathscr{F}}$ & $T_{\rm model}$ & $T_{\rm alg}$ & $\widehat S_{\rm alg}$\\
\hline
Fine & $-$ & $-$ & $-$ & $-$ & 98.93 & 1\\
Parareal & 33 & 0.16 & 0.33 & 0.003 & 11.18 & 8.85\\
nnGParareal & 3 & 0.13 & 0.37 & 2.95 & 4.20 & \textbf{23.50}\\
\hline
\end{tabular}}
\caption{Performance of Parareal and nnGParareal ($m=18$) for the heat equation \eqref{heat_eq}. Times are in seconds. $T_{\mathscr{F}}$ and $T_{\mathscr{G}}$ represent the fine/coarse solvers' runtime per iteration, $K_{\rm alg}$ is iteration count to convergence, $T_{\rm model}$ includes training/hyperparameter selection/prediction time, and $T_{\rm alg}$ and $\widehat S_{\rm alg}$ represent total runtime and empirical parallel speedup, respectively.}
\label{tab_res}
\end{table}
While we acknowledge the moderate spatial dimension and relatively low time resolution, which typically do not warrant parallelization across $N=300$ processors, this controlled example demonstrates that our method, and PinT methods more broadly, do not rely on significant over-resolution in one discretization dimension to achieve meaningful parallel performance.

 \section{Jitter and hyperparameters optimization}
\label{app:impl_det}
The original GParareal implementation did not provide an automated strategy for tuning the jitter $\sigma_{\rm reg}^2$ in \eqref{eq:gp_posterior_m}. %
Using a fixed value may be too conservative and yield sub-optimal performance \cite{ameli2022noise,ranjan2011computationally,bostanabad2018leveraging}. The relevance of setting $\sigma_{\rm reg}^2>0$ in the case of the deterministic output, where the process has no intrinsic variance, is to reduce the numerical instability when inverting the GP covariance. There is extensive evidence of this phenomenon \cite{booker1999rigorous,  gramacy2012cases,   huang2006sequential, peng2014choice,pepelyshev2010role, taddy2009bayesian}. Numerical instability can arise when the condition number of the covariance matrix, defined as the ratio between the largest and the smallest singular value, is too large. Setting $\sigma_{\rm reg}^2>0$ affects the spectral properties of the (symmetric) covariance matrix by shifting the singular values by a constant proportional to it. This can decrease the condition number while leaving the inverse computation largely unaffected; refer to \cite{andrianakis2012effect} for a precise treatment.

We found the tuning of $\sigma_{\rm reg}^2$ to be important for (nn)GParareal, given the close distance between points in the dataset, as seen in Figure~\ref{fig:brus_dataset_vis_para_both}. This is a situation prone to numerical instability \cite{mohammadi2016analytic}. Some authors recommend setting a lower bound for the jitter to ensure the matrix does not leave the symmetric positive definite cone. These values are either found analytically \cite{peng2014choice} or via experiments \cite{ranjan2011computationally}. Others \cite{bostanabad2018leveraging, mohammadi2016analytic} advocate for the use of maximum likelihood (MLE) or leave-one-out cross-validation (LOO-CV), which can be computed in closed-form, for a data-driven tuning approach. In our application, we found MLE slightly more stable than LOO-CV in selecting $\sigma_{\rm reg}^2$, which is searched on a grid of size $n_{\rm reg}$,
with values as low as -20 on a log scale. All GParareal and nnGParareal results in this work feature automatic nugget selection. The remaining hyperparameters $(\sigma^2_{\rm i}, \sigma^2_{\rm o})$ are also tuned by MLE, using the Nelder-Mead optimization algorithm. Evidence of the advantages of MLE in the scope of hyperparameters optimization is shown in \cite{petit2021parameter}. Local minima are addressed using random restarts, i.e. by running the optimization procedure Nelder-Mead $n_{\rm start}$ times, independently, with different random initial guesses for the hyperparameters $\bs{\theta}$. Optimal hyperparameters computed during previous iterations can also be used as starting points for the optimization procedure.



\end{document}